\documentclass[prb,reprint,twocolumn,amsmath,amssymb]{revtex4-1}

\newcommand{\bK}{\boldsymbol{K}}
\newcommand{\bL}{\boldsymbol{L}}

\newcommand{\br}{\boldsymbol{\textbf{r}}}

\newcommand{\bz}{\boldsymbol{z}}
\newcommand{\bx}{\boldsymbol{\textbf{x}}}
\newcommand{\bA}{\textbf{A}}
\newcommand{\bAbar}{\bar{\textbf{A}}}
\newcommand{\bAtilde}{\tilde{\textbf{A}}}
\newcommand{\bB}{\textbf{B}}

\newcommand{\bH}{\textbf{H}}
\newcommand{\bHbar}{\bar{\textbf{H}}}
\newcommand{\Dbar}{\bar{D}}
\newcommand{\bGbar}{\bar{\textbf{G}}}
\newcommand{\bQ}{\textbf{Q}}
\newcommand{\bU}{\textbf{U}}
\newcommand{\bUbar}{\bar{\textbf{U}}}
\newcommand{\bV}{\textbf{V}}
\newcommand{\bVbar}{\bar{\textbf{V}}}

\newcommand{\bM}{\textbf{M}}

\newcommand{\bR}{\boldsymbol{\textbf{R}}}
\newcommand{\bS}{\boldsymbol{\textbf{S}}}
\newcommand{\bT}{\textbf{T}}

\newcommand{\bX}{\textbf{X}}
\newcommand{\bY}{\textbf{Y}}
\newcommand{\bphi}{\boldsymbol{\phi}}

\newcommand{\psibar}{\bar{\psi}}
\newcommand{\bpsi}{\boldsymbol{\psi}}
\newcommand{\bpsibar}{\bar{\boldsymbol{\psi}}}

\newcommand{\bgamma}{\boldsymbol{\gamma}}

\newcommand{\bxi}{\boldsymbol{\xi}}

\newcommand{\dr}{\,d\br}

\newcommand{\order}{\mathcal{O}}

\newcommand{\norm}[1]{\left\lVert#1\right\rVert}
\newcommand{\modulus}[1]{\left\vert#1\right\vert}
\newcommand{\Lone}{{L^1(\Omega)}}
\newcommand{\Ltwo}{{L^2(\Omega)}}
\newcommand{\Linf}{{L^\infty(\Omega)}}
\newcommand{\Hone}{{H^1(\Omega)}}
\newcommand{\al}{_{\alpha}}
\newcommand{\be}{_{\beta}}

\usepackage{subfiles}
\usepackage{graphicx}
\usepackage{graphics}
\usepackage{multirow}
\usepackage{fancyhdr, ifpdf}
\usepackage{amsxtra}
\usepackage{stmaryrd}
\usepackage{mathrsfs}
\usepackage{mathtools}
\usepackage{psfrag}
\usepackage{threeparttable}
\usepackage{subfigure}
\usepackage{epsfig}
\usepackage{rotating}
\usepackage{setspace}
\usepackage{float}
\usepackage{amsfonts}
\usepackage{amsbsy}
\usepackage{amscd}
\usepackage{amsthm}
\usepackage{color}
\usepackage{tabularx}
\usepackage{sidecap}
\usepackage{dcolumn}
\usepackage{footnote}
\usepackage{algorithm}
\usepackage{algorithmic}
\usepackage{enumitem}
\usepackage[titletoc,title]{appendix}
\newtheorem{prop}{Proposition}
\newtheorem{corollary}{Corollary}[prop]

\definecolor{hellgruen}{rgb}{0.2,0.7,0.2}

\newcolumntype{M}[1]{>{\centering\arraybackslash}m{#1}}
\newcolumntype{N}{@{}m{0pt}@{}}
\DeclareMathOperator*{\argmax}{arg\,max}
\begin{document}
\raggedbottom
\title{Real-time time-dependent density functional theory using higher-order finite-element methods}
\author{Bikash Kanungo}
\affiliation{Department of Mechanical Engineering, University of Michigan, Ann Arbor, Michigan 48109, USA}
\author{Vikram Gavini}
\affiliation{Department of Mechanical Engineering, University of Michigan, Ann Arbor, Michigan 48109, USA}
\affiliation{Department of Materials Science and Engineering, University of Michigan, Ann Arbor, Michigan 48109, USA}
\begin{abstract}
    We present a computationally efficient approach to solve the time-dependent Kohn-Sham equations in real-time using higher-order finite-element spatial discretization, applicable to both pseudopotential and all-electron calculations. To this end, we develop an \emph{a priori} mesh adaption technique, based on the semi-discrete (discrete in space but continuous in time) error estimate on the time-dependent Kohn-Sham orbitals, to construct an efficient finite-element discretization. Subsequently,
    we obtain the full-discrete error estimate to guide our choice of the time-step. We employ spectral finite-elements along with special reduced order quadrature to render the overlap matrix diagonal, thereby simplifying the inversion of the overlap matrix that features in the evaluation of the discrete time-evolution operator. We use the second-order Magnus operator as the time-evolution operator, wherein the action of the discrete Magnus operator, expressed as
    exponential of a matrix, on the Kohn-Sham orbitals is obtained efficiently through an adaptive Lanczos iteration. We observe close to optimal rates of convergence of the dipole moment with respect to spatial and temporal discretization, for both pseudopotential and all-electron calculations. We demonstrate a staggering 100-fold reduction in the computational time afforded by higher-order finite-elements over linear finite-elements, for both pseudopotential and all-electron calculations. Further, for similar level of accuracy, we obtain significant computational savings by our approach as compared to state-of-the-art finite-difference methods. We also demonstrate the competence of higher-order finite-elements for all-electron benchmark systems. Lastly, we observe good parallel scalability of the proposed method on many hundreds of processors. 
\end{abstract}
\maketitle
\section{Introduction}\label{sec:Intro}
Time-dependent density functional theory (TDDFT) extends the keys ideas of ground-state density functional theory (DFT) to electronic excitations and time-dependent processes. It relies on the Runge-Gross theorem ~\cite{Runge1984} to establish, for a given initial state, a one-to-one correspondence between the time-dependent external potential and the time-dependent electronic density, thereby making the electronic density the fundamental variable to define other physical quantities. Subsequently, one invokes the Kohn-Sham \emph{ansatz} ~\cite{Kohn1965} to reduce the many-electron time-dependent Schr\"odinger equation to a set of effective single electron equations, called the time-dependent Kohn-Sham (TDKS) equations. For all practical purposes, it requires the use of approximate exchange-correlation functionals, analogous to the ground-state case. However, TDDFT offers a great balance of accuracy and computational efficiency which have enabled the study of a wide-range of time-dependent phenomena---optical ~\cite{Appel2003} and higher-order responses ~\cite{Gonze1989, vanGisbergen1997}, electron transport ~\cite{Stefanucci2004, Kurth2005}, charge-transfer excitations ~\cite{Jamorski2003, Stein2009}, dynamics of chemical bonds ~\cite{Burnus2005}, multi-photon ionization ~\cite{Tong1998, Tong2001, Telnov2009}, to name a few. 

Given the practical significance of TDDFT calculations, there has been a growing interest in developing faster and more accurate numerical methods for solving the TDKS equations, over the past two decades. Broadly, these numerical methods can be classified into two categories, characterized by the strength of the light-matter interaction, namely, linear-response time-dependent density functional theory (LR-TDDFT) ~\cite{Casida1995,Petersilka1996} and real-time time-dependent density functional theory (RT-TDDFT) ~\cite{Theilhaber1992, Yabana1996,Baer2001}. The LR-TDDFT pertains to the case of weak interaction between the external field and the system, wherein the field induces a small perturbation from the ground-state. In such perturbative regime, one can compute the linear density response from the ground-state itself, which in turn can be used for the calculation of first-order response functions such as the absorption spectra. The RT-TDDFT, on the other hand, is a more generic framework which captures the electronic dynamics in real-time, thereby, allowing to handle both perturbative and non-perturbative regimes (e.g., harmonic generation, electron transport) in a unified manner. This involves propagating the TDKS equations in real-time without any restriction to the external field in terms of its frequency, shape or intensity. This work pertains to the more general RT-TDDFT.

Despite its generality in dealing with various time-dependent processes, there are two major challenges associated with RT-TDDFT. The first stems from the quality of the time-dependent exchange-correlation approximation used in the TDKS equations. The exact exchange-correlation functional is, in general, nonlocal in both space and time ~\cite{Vignale1995, Maitra2002,Marques2006} and has an initial-state dependence ~\cite{Maitra2001}. However, the lack of insight into its time nonlocality and initial-state dependence has necessitated the use of the adiabatic approximation, wherein the exchange-correlation functional is defined in terms of the instantaneous electronic density. Although the applicability of the adiabatic approximation to various systems and materials properties are yet to be understood, they have shown remarkable agreement in estimating the transition frequencies ~\cite{Appel2003}, and, in most cases, is the underlying approximation in existing RT-TDDFT softwares. As with most of the numerical implementations in RT-TDDFT, this work is restricted to the adiabatic approximation. The second challenge stems from the huge computational cost associated with the non-linear TDKS equations. Numerical simulations for large length- and time-scales are still computationally challenging, and warrant systematically improvable, accurate, efficient and scalable spatio-temporal discretization. Addressing these numerical challenges constitutes the main subject of this work. 

Significant efforts have been made towards efficient RT-TDDFT numerical schemes as extensions to popular ground-state DFT packages, borrowing from their respective spatial discretization. These include planewave basis in QBox ~\cite{QBox2008,Schleife2012}; linear combination of atomic orbitals (LCAO) in Siesta ~\cite{Siesta2002,Takimoto2007} and GPAW ~\cite{Kusima2015}; Gaussian basis in NWChem ~\cite{NWChem2010,Lopata2011}; and finite-difference based approaches in Octopus ~\cite{Octopus2006}
and GPAW ~\cite{GPAW2005,Walter2008}. The planewave basis, owing to its completeness, provides systematic convergence, and affords an efficient treatment of the electrostatic interactions through fast Fourier transforms. However, they remain restricted to only periodic geometries and boundary conditions, thereby ill-equipped to describe systems with defects, and non-periodic systems like isolated molecules and nano-clusters. Additionally, the nonlocality of the basis greatly hinders its parallel
scalability. Atomic-type orbitals, such as LCAO and Gaussian basis, owing to their atom-specific basis, are well-suited to describe molecules and nano-clusters for both pseudopotential as well as all-electron calculations. However, owing to the incompleteness of such basis, systematic convergence for all materials systems remains a concern. 
The finite difference discretization (FD) provides systematic convergence, can handle a broad range of boundary conditions, and exhibits improved parallel scalability in comparison to planewave and atomic-type orbital basis. However, incorporating adaptive spatial resolution in FD through a non-uniform grid remains non-trivial. This lack of adaptive spatial resolution in FD also renders it inefficient for an accurate treatment of singular potentials (as in the case of all-electron calculations), thereby restricting the applicability of FD to only pseudopotential calculations. On the other hand, the finite-element basis ~\cite{BrennerScott2007,Hughes2012}, being a local-piecewise polynomial basis, offers several key advantages---it provides systematic convergence; is amenable to adaptive spatial resolution, and thereby suitable for both pseudopotential and all-electron calculations; exhibits excellent parallel scalability owing to the locality of the basis; and admits arbitrary geometries and boundary conditions. We add that many of these advantages of finite-element basis are also shared by the wavelets basis~\cite{Genovese2008}. While, at present, the use of wavelets basis has been restricted LR-TDDFT~\cite{Natarajan2012}, we expect them to be a competent basis for RT-TDDFT as well.  

The efficacy of the finite-element basis in terms of its accuracy, efficiency, scalability and relative performance with other competing methods (e.g., planewaves, Gaussian basis, FD), have been thoroughly studied in the context of ground-state DFT, for both pseudopotential ~\cite{White1989, Tsuchida1998, Pask1999, Pask2001, Pask2005, Fattebert2007, Zhang2008, Suryanarayana2010, Fang2012, Bao2012, Chen2013, Motamarri2013, Chen2014, Motamarri2018, Motamarri2019} and all-electron calculations ~\cite{White1989, Motamarri2013, Motamarri2018, Yamakawa2005, Bylaska2009, Lehtovaara2009, Schauer2013, Motamarri2014, Kanungo2017, Motamarri2017, Motamarri2019}. A similarly comprehensive study on the efficacy of the finite-element basis for RT-TDDFT is, however, lacking. While two recent studies ~\cite{Lehtovaara2011, Bao2015} demonstrate the accuracy of finite-elements for RT-TDDFT, they remain restricted to only linear and quadratic finite-elements. As known from prior studies in ground-state DFT ~\cite{Hermansson1986, Bylaska2009, Motamarri2013}, lower-order (linear and quadratic) finite-elements require a large number of basis functions ($50,000-500,000$ per atom for pseudopotential calculations) to achieve chemical accuracy, and hence, perform poorly in comparison to planewaves and other real-space based methods. However, this shortcoming of linear and quadratic finite-elements for ground-state DFT calculations has been shown to be alleviated by the use of higher-order finite-elements ~\cite{Motamarri2013}. In this work, we extend the use of higher-order finite-elements to RT-TDDFT calculations and demonstrate the resulting advantages over lower-order finite-elements as well as finite-difference based methods.

The keys ideas in this work can be summarized as: (i) developing an \emph{a priori} mesh-adaption based on semi-discrete (discrete in space, continuous in time) error analysis of the TDKS equations, and subsequently, obtaining an efficient finite-element discretization for the problem; (ii) use of spectral finite-elements in conjunction with Gauss-Legendre-Lobatto quadrature to render the overlap matrix diagonal, thereby simplifying the evaluation of the inverse of the overlap matrix that features in the discrete time-evolution operator; (iii) obtaining an efficient temporal discretization using a full-discrete error analysis of the TDKS equations, in the context of second-order Magnus time-evolution operator; and (iv) using an adaptive Lanczos iteration to efficiently compute the action of the Magnus propagator on the Kohn-Sham orbitals. 
The \emph{a priori} mesh-adaption in this work is performed by minimizing the discretization error in the observable of importance, subject to fixed number of elements in the finite-element mesh. In particular, we minimize the semi-discrete error in the dipole moment of the system with respect to the mesh-size distribution, $h(\br)$, to obtain an efficient \emph{a priori} spatial discretization. Having obtained the spatial discretization, an efficient temporal discretization is obtained through a full-discrete error analysis, in the context of second-order Magnus time-evolution operator. This is, to the best of our knowledge, the first work that guides the spatio-temporal discretization for the RT-TDDFT problem using error estimates.  

We study the key numerical aspects of the proposed higher-order finite-element discretization for benchmark systems involving both nonlocal pseudopotential and all-electron calculations. To begin with, we study the numerical rates of convergence of the dipole moment with respect to spatial and temporal discretization. We use two benchmark systems: (i) a pseudopotential calculation on methane molecule; and (ii) an all-electron calculation on lithium hydride molecule, to demonstrate the rates of convergence for linear, quadratic and fourth-order finite-elements. We observe numerical rates of convergence in the dipole moment close to the optimal rates obtained from our error analysis. Next, we assess the computational advantage afforded by higher-order finite-elements over linear finite-element, using the same benchmark systems. We observe an extraordinary $100$-fold speedup in terms of the total computational time for the fourth-order finite-element over linear finite-element, for calculations in the regime of chemical accuracy. We also compare the relative performance of the finite-element discretization against finite-difference method for pseudopotential calculations. We use aluminum clusters ($\text{Al}_2$ and $\text{Al}_{13}$), and the Buckminsterfullerene (C60) molecule as our benchmark pseudopotential systems. The finite-difference based calculations are done using the Octopus package ~\cite{Octopus2006}. Depending on the benchmark system, the finite-element discretization shows a $3$- to $60$-fold savings in computational time as compared to the finite-difference approach, for pseudopotential calculations. We also demonstrate the efficacy of finite-elements for systems subjected to strong perturbation by studying higher harmonic generation in $\text{Mg}_2$. Additionally, we demonstrate the competence of finite-elements for all-electron calculations on two benchmark systems---methane and benzene molecule. Lastly, we study the strong scaling of our implementation and observe good parallel scalability with $\sim75\%$ efficiency at 768 processors for a benchmark system of a Buckminsterfullerene molecule containing $3.5$ million degrees of freedom. 

The rest of the paper is organized as follows. In Section ~\ref{sec:Formulation}, we briefly discuss the TDKS equations and the form of the exact time-evolution operator. In Section ~\ref{sec:Discrete}, we introduce the notion of semi- and full-discrete solution to the TDKS equation. In Section ~\ref{sec:Error}, we provide formal spatial and time discretization error estimates in the Kohn-Sham orbitals. Sec. ~\ref{sec:OptimalDiscretization} provides an efficient spatio-temporal discretization scheme guided by the error estimates. In Section ~\ref{sec:Numerics}, we describe the various numerical implementation aspects pertaining to spectral finite-elements and the discrete second-order Magnus
operator. Section ~\ref{sec:Results} details the convergence, accuracy, efficiency and parallel scalability of the higher-order finite-elements along with its relative performance against the finite-difference method. Finally, we summarize our findings and outline the future scope in Section ~\ref{sec:Summary}. 

\section{Time-dependent Kohn-Sham Equations}\label{sec:Formulation}
TDDFT relies on the Runge-Gross theorem ~\cite{Runge1984} and the Kohn-Sham \textit{ansatz} ~\cite{Kohn1965} to reduce the many-electron time-dependent Schr\"odinger equation to a set of effective single electron equations, called the time-dependent Kohn-Sham (TDKS) equations. These equations prescribe the evolution of an auxiliary system of non-interacting electrons that yield the same time-dependent electronic charge density, $\rho(\br,t)$, as that of the interacting system. The TDKS equations,
in atomic units, are given as
\begin{equation} \label{eq:TDKS}
\begin{split}
    i\frac{\partial{\psi\al(\br,t)}}{\partial{t}} &= H_{KS}[\rho](\br,t;\bR)\psi\al(\br,t) \\
    & \coloneqq \left[-\frac{1}{2}\nabla^2+ V_{KS}[\rho](\br,t;\bR)\right]\psi\al(\br,t) \,,
\end{split}
\end{equation}
where $H_{KS}[\rho](\br,t;\bR)$, $V_{KS}[\rho](\br,t;\bR)$ and $\psi\al(\br,t)$ represent the time-dependent Kohn-Sham Hamiltonian, potential and orbitals, respectively, with the index $\alpha$ spanning over all the $N_e$ electrons in the system. $\bR=\{\bR_1,\bR_2,...,\bR_{N_a}\}$ denotes the collective representation for the positions of the $N_a$ atoms in the system. The electron density, $\rho(\br,t)$, is given in terms of the Kohn-Sham orbitals as
\begin{equation}
    \rho(\br,t)=\sum_{\alpha=1}^{N_e}|\psi\al(\br,t)|^2\,.
\end{equation}
In the present work, we restrict ourselves to only non-periodic (clusters and molecules) as well as spin-unpolarized systems. However, we note that all the ideas discussed subsequently can be generalized to spin-polarized systems as well. \\

The time-dependent Kohn-Sham potential, $V_{KS}[\rho](\br,t;\bR)$ in Eq. ~\ref{eq:TDKS}, is given by 
\begin{equation}\label{eq:VKS}
\begin{split}
    V_{KS}[\rho](\br,t;\bR)&=V_{ext}(\br,t;\bR)+V_{H}[\rho](\br,t)+V_{XC}[\rho](\br,t)\,,
\end{split}
\end{equation}
where $V_{ext}(\br,t;\bR)$ denotes the external potential, $V_{H}[\rho](\br,t)$ denotes the Hartree potential, and $V_{XC}[\rho](\br,t)$ represents the exchange-correlation potential. The exchange-correlation potential, $V_{XC}[\rho](\br,t)$, in general, is nonlocal in both space and time ~\cite{Vignale1995, Maitra2002,Marques2006}, and has a dependence on the initial many-electron wavefunction ~\cite{Maitra2001}. However, in absence of the knowledge of its true form, most of the existing approximations use locality in time (adiabatic exchange-correlation) and non-dependence on the initial many-electron wavefunction. This allows for direct use of the existing exchange-correlation approximations used in ground-state DFT. In the present work, we use the adiabatic local-density approximation (ALDA) ~\cite{Gross1985}, which is local in both space and time. Specifically, we use the Ceperley-Alder form \cite{Ceperley1980}. \\

In Eq.~\ref{eq:VKS}, the Hartree potential is given by
\begin{equation} \label{eq:VH}
    V_{H}[\rho](\br,t)=\int{\frac{\rho(\br',t)}{|\br-\br'|}\dr'}\,.
\end{equation}
The external potential comprises of the nuclear potential $V_N(\br;\bR)$ and the external field $V_{field}(\br,t)$. The nuclear potential is given by
\begin{equation}
    V_N(\br;\bR) = \begin{cases}
        V_N^{ae} = -\sum\limits_{I=1}^{N_a}\frac{Z_I}{|\br-\bR_I|}, \quad \text{for all-electron}\,,  \\ \\
        V_N^{psp}(\bR), \quad \text{for pseudopotential} \,, 
    \end{cases}
\end{equation}
where $Z_I$ and $\bR_I$ represent the atomic charge and position of the $I^{th}$ nucleus. For a typical pseudopotential calculation, $V_N^{psp}$ comprises of a local part, $V_{psp}^{loc}$, and a nonlocal part, $V_{psp}^{nl}$. For the nonlocal part, the action on a function $\phi(\br)$, written in the Kleinman-Bylander form ~\cite{Kleinman1982}, is given by 
\begin{widetext}
\begin{equation} \label{eq:VPSP}
    V_{psp}^{nl}(\bR)\phi(\br)=\sum_{I=1}^{N_a}\sum_{l=0}^{L_I}\sum_{m=-l}^l\left(\frac{\int{u_{lm}^I(\br')\delta V^I_l(\br')\phi(\br')\,d\br'}}{\int{u_{lm}^I(\br')\delta V_l^I(\br')u_{lm}^I(\br')}\,d\br'}\right)\delta V_l^I(\br) u_{lm}^I(\br)\,,
\end{equation}
\end{widetext}
where $l$ and $m$ denote the angular and magnetic quantum number, respectively. $u_{lm}^I(\br)$ is a pseudo-atomic eigenfunction for the atom at $\bR_I$, $\delta V^I_l(\br)$ is the specified $l$ angular component short-ranged potential for the atom at $\bR_I$, and $L_I$ is the maximum angular quantum number specified for the atom at $\bR_I$.
The external field, $V_{field}(\br,t)$, is typically provided as a monochromatic laser pulse of the form 
\begin{equation} \label{eq:VField}
    V_{field}(\br,t)=-\mathbf{E_0}(t)\cdot\br\,,
\end{equation}
where $\mathbf{E_0}(t)$ represents the time-dependent electric field. 

We note that both the electrostatic potentials---Hartree and nuclear (all-electron)---are extended in real space. However, using the fact that the $\frac{1}{|\br|}$ kernel in these extended interactions is the Green's function of the Laplace operator, one can recast their evaluation as the solutions to the following Poisson equations:
\begin{subequations} \label{eq:Poisson}
  \begin{equation} \label{eq:Hartree_possion}
    -\frac{1}{4\pi}\nabla^2V_{H}(\br,t)=\rho(\br,t)\,,\quad V_H(\br,t)\rvert_{\partial\Omega}=f(\br,\bR) \,,
  \end{equation}    
  \begin{equation} \label{eq:Nuclear_poisson}
      -\frac{1}{4\pi}\nabla^2V_{N}^{ae}(\br;\bR)=b(\br,\bR)\,, \quad  V_N^{ae}(\br)\rvert_{\partial\Omega}=-f(\br,\bR)\,.
  \end{equation}
\end{subequations}
In the above equation, $b(\br;\bR)=-\sum\limits_{I=1}^{N_a}{Z_I\delta(\br;\bR_I)}$, where $\delta(\br;\bR_I)$ is a bounded regularization of the Dirac-delta distribution with compact support in a small ball around $\bR_I$ and satisfies $\int \delta(\br;\bR_I)\dr=1$; $f(\br,R)=\sum_{I=1}^{N_a}\frac{Z_I}{\modulus{\br-\bR_I}}$; and $\partial \Omega$ denotes the boundary of a sufficiently large bounded domain $\Omega \in \mathcal{R}^3$. We refer to previous works on finite-elements based ground-state DFT calculations ~\cite{Pask1999, Pask2005, Suryanarayana2010, Motamarri2013, Motamarri2012} for a detailed treatment of the local reformulation of the electrostatic potentials into Poisson equations.

Formally, the solution to Eq.~\ref{eq:TDKS} can be written as
\begin{equation} \label{eq:Propagator}
\begin{split}
        &\psi\al(\br,T)=U(T,t_0)\psi\al(\br,t_0)\\
        &=\mathcal{T}\text{exp}\left\{-i\int_{t_0}^T{H_{KS}[\rho](\br,\tau)d\tau}\right\}\psi\al(\br,t_0)\,,
\end{split}
\end{equation}
where $U(T,t_0)$ represents the time-evolution operator (propagator) and $\mathcal{T}$ denotes the time-ordering operator. Although the above equation provides a formal way to directly evaluate the orbitals at any time, $t$, resolving the implicit time-dependence of the Kohn-Sham Hamiltonian through the density is too difficult. However, one can exploit the following composition property of the propagator,
\begin{equation} \label{eq:PropagatorComposition1}
    U(t_2,t_0)=U(t_2,t_1)U(t_1,t_0)\,, \quad \quad \quad t_0<t_1<t_2\,,
\end{equation}
to accurately resolve the implicit time-dependence in $H_{KS}[\rho](\br,t)$. To elaborate, the above property allows us to rewrite the propagator $U(T,t_0)$ as
\begin{equation}\label{eq:PropagatorComposition2}
    U(T,t_0)=\prod_{i=0}^{N-1}{U(t_{i+1},t_i)}\,, 
\end{equation}
where $t_N=T$ and $t_{i+1}-t_i=\Delta t_i$, with $\Delta t_i$ denoting the variable time step. Consequently, one can divide the evaluation of the orbitals at $T$ into $N$ short-time propagation, given by
\begin{equation} \label{eq:ShortTimePropagation}
\begin{split}
    &\psi\al(\br,t+\Delta t)=U(t+\Delta t)\psi\al(\br,t)\\
    &=\mathcal{T}\text{exp}\left\{-i\int_t^{t+\Delta t}{H_{KS}[\rho](\br,\tau)d\tau}\right\}\psi\al(\br,t)\,.
\end{split}
\end{equation}
In addition to resolving the implicit time-dependence in $H_{KS}[\rho](\br,t)$, the short time propagation provides the numerical advantage of containing the norm of the exponent in Eq. ~\ref{eq:Propagator}. To elaborate, any efficient numerical scheme to compute the action of the propagator on a wavefunction involves either a power series expansion or a subspace projection of the propagator, wherein the number of terms in the power series or the dimension of the subspace required for a given accuracy are dependent on norm of the exponent.
Moreover, there is a physical upper bound imposed on the time step based on the maximum frequency, $\omega_{max}$, that one wants to resolve in their calculations, i.e., $\Delta t_{max}=\frac{1}{\omega_{max}}$. Typically, $\omega_{max}$ is determined by the
eigen-spectrum of the ground-state Hamiltonian or by the frequency of the applied field, $V_{field}$. We note that, in practice, one uses a time step $\Delta t \ll \Delta t_{max}$ owing to the need of containing time-discretization errors that arise in approximating the continuous propagator, $\mathcal{T}\text{exp}\left\{-i\int_t^{t+\Delta t}{H_{KS}[\rho](\br,\tau)d\tau}\right\}$. We discuss these approximations and their associated time-discretization errors in greater detail in Sec. ~\ref{sec:Discrete} and Sec.~\ref{sec:Error}.                  

\section{Semi- and full-discrete solutions} \label{sec:Discrete}
In this section, we introduce the notion of semi-discrete (discrete in space but continuous in time) and full-discrete solution to the TDKS equation. The full-discrete solution is provided in the context of second-order Magnus propagator. 

To begin with, we provide some of the finite-element essentials. In the finite-element method, the spatial domain of interest ($\Omega \in \mathbb{R}^3$) is divided into non-overlapping sub-domains, known as finite-elements. Each finite-element ($e$) is characterized by its spatial extent ($\Omega_e$) and size ($h_e$). Subsequently, the finite-element basis is constructed from piecewise Lagrange interpolating polynomials that have a compact support on the finite elements (i.e., on $\Omega_e$), thus rendering locality to these basis functions. We note that there is an abundance of choice in terms of the form and order of the polynomial functions that can be used in constructing the finite-element basis. We refer to Refs. ~\cite{Hughes2012} and ~\cite{Bathe2006} for a comprehensive discourse on the subject.

\subsection{Semi-discrete solution} \label{sec:SemiDiscrete}
To begin with, we express the semi-discrete time-dependent Kohn-Sham orbitals, $\psi\al^h(\br,t)$, as  
\begin{equation} \label{eq:PsiExpansion}
\psi\al^h(\br,t) = \sum_{j=1}^{n_h}N_j(\br)\psi\al^j(t)\,,\,\, \text{s.t. } \psi\al^h(\br,t)\rvert_{\partial\Omega}=0\,\, \forall t\geq 0\,,
\end{equation}
where $\{N_j(\br)\}$ represents the set of finite-element basis functions, each of polynomial order $p$; and  $\psi\al^j(t)$ denote the time-dependent expansion coefficient corresponding to the $N_j$ basis function. We refer to the Appendix for a formal discussion on the appropriate function space for $\psi\al^h(\br,t)$. Using the discretization of Eq.~\ref{eq:PsiExpansion} in the TDKS equation (Eq. ~\ref{eq:TDKS}) results in following discrete equation,
\begin{equation} \label{eq:TDKSMatVec}
i\bM\dot{\bpsi}\al(t) = \bH\bpsi\al(t),
\end{equation}
where $\bH$ and $\bM$ denote the discrete Hamiltonian and overlap matrix, respectively, and $\bpsi\al(t)$ denotes the vector containing the coefficients $\psi\al^j(t)$. The solutions to the above equation are called the \emph{semi-discrete} solutions to the TDKS equation. In the above equation, the discrete Hamiltonian $H_{jk}$ is given by
\begin{equation}\label{eq:DiscreteHamiltonian}
\begin{split}
    H_{jk}&= \frac{1}{2}\int_{\Omega} \nabla N_j(\br) \cdot \nabla N_k(\br) \dr \\
    \quad &+ \int_{\Omega} V_{KS}^h[\rho^h](\br,t)N_j(\br) N_k(\br) \dr\,, 
\end{split}
\end{equation}
where $V_{KS}^h[\rho^h](\br,t)$ is the discrete Kohn-Sham potential corresponding to the semi-discrete density, $\rho^h(\br,t)$ (i.e., evaluated from the solutions of Eq. ~\ref{eq:TDKSMatVec}). $V_{KS}^h[\rho^h](t)$ is, in turn, given by 
\begin{equation}\label{eq:VKSSemiDiscrete}
\begin{split}
    V_{KS}^h[\rho^h](\br,t) &= V_H^h[\rho^h](\br,t) + V_N^{h}(\br) \\
    \quad & + V_{XC}[\rho^h](\br,t) + V_{field}(\br,t) \,,
    \end{split}
\end{equation}
where $V_H^h[\rho^h](\br,t)$ and $V_N^h(\br)$ denote the discrete Hartree and nuclear potential, respectively. We note that for the pseudopotential case, $V_N^h$ is same as the continuous potential $V_N^{psp}$ and hence $V_N^h$ is relevant only in the all-electron case. Similar to $\psi\al^h(\br,t)$, the discrete electrostatic potentials ($V_H^h[\rho^h](\br,t)$ and $V_N^h(\br)$) are obtained by discretizing them in the finite-element basis, i.e., 
\begin{equation} \label{eq:VHartreeExpansion}
V_H^h[\rho^h](\br,t) = \sum_{j=1}^{n_h}N_j(\br)\phi_H^j(t)\,,
\end{equation} 
\begin{equation} \label{eq:VNExpansion}
V_N^{ae,h}(\br) = \sum_{j=1}^{n_h}N_j(\br)\phi_N^j\,,
\end{equation}
satisfying the boundary conditions presented in Eq.~\ref{eq:Poisson} in a discrete sense. We refer to the Appendix for a formal discussion on the appropriate function spaces for $V_H^h[\rho^h](\br,t)$ and $V_N^h(\br)$. Subsequently, the expansion coefficients $\phi_H^j(t)$ and $\phi_N^j$ can be obtained by using the above expressions in Eq. ~\ref{eq:Poisson}, and solving the resulting system of linear equations
\begin{equation}\label{eq:VHDiscretePoisson}
    \bK\bphi_H(t) = 4\pi\textbf{c}(t)\,, ~\text{and}
\end{equation}
\begin{equation}\label{eq:VNDiscretePoisson}
    \bK\bphi_N = 4\pi\textbf{d}\,,
\end{equation}
where $K_{jk}=\int_{\Omega}\nabla N_j(\br).\nabla N_k(\br)\dr$; $\bphi_H$ and $\bphi_N$ are the vectors containing the coefficients $\phi_H^j(t)$ and $\phi_N^j$, respectively; $c_j(t)=\int_{\Omega}\rho^h(\br,t)N_j(\br)\dr$; and $d_j=\int_{\Omega}b(\br,\bR)N_j(\br)\dr$. 

\subsection{Full-discrete solution}\label{sec:FullDiscrete}
We now discuss the full-discrete solution to the TDKS equations, in the context of second-order Magnus propagator. To begin with, we note that the overlap matrix $\bM$, being positive definite, guarantees the existence of a unique positive definite square root, $\bM^{1/2}$. This allows us to rewrite Eq. ~\ref{eq:TDKSMatVec} as
\begin{equation} \label{eq:TDKSMatVecStandard}
    i\dot{\bpsibar}\al(t) = \bHbar \bpsibar\al(t)\,,
\end{equation}
where $\bpsibar\al(t) = \bM^{1/2}\bpsi\al(t)$ and $\bHbar = \bM^{-1/2}\bH\bM^{-1/2}$. To put it differently, $\bpsibar\al(t)$ is the representation of $\psi\al^h(\br,t)$ in a L\"owdin orthonormalized basis ~\cite{Lowdin1950}. We remark that the practicality of the above reformulation in terms of $\bpsibar\al$ is contingent upon the efficient evaluation of $\bM^{-1/2}$. To that end, we present an efficient scheme for computing $\bM^{-1/2}$ in Sec. ~\ref{sec:Numerics}. 

We invoke the Magnus \emph{ansatz} to write the solution of Eq. ~\ref{eq:TDKSMatVecStandard} as
\begin{equation} \label{eq:MagnusAnsatz} 
    \bpsibar\al(t) = \text{exp}(\bA(t)) \bpsibar\al(0)\,, \qquad \forall t \geq 0 \,.
\end{equation}
The operator $\text{exp}\left(\bA(t)\right)$ is termed as the \emph{Magnus propagator}, where $\bA(t)$ is given explicitly as ~\cite{Blanes2009, Hochbruck2003}
\begin{equation} \label{eq:MagnusExpansion}
\begin{split}
    	\bA(t) &= \int_0^t -i\bHbar(\tau)d\tau \\
    	\quad &- \frac{1}{2}\int_0^t\left[\int_0^{\tau} -i\bHbar(\sigma)d\sigma,-i\bHbar(\tau)\right]d\tau 
    	+ \ldots \quad \,,
\end{split}
\end{equation}
where $[\bX,\bY]=\bX\bY-\bY\bX$ denotes the commutator. Although known explicitly, the above equation is not practically useful, given the difficulty in resolving the implicit dependence of $\bHbar(t)$ on $\rho^h(\br,t)$. As mentioned in Section~\ref{sec:Formulation}, we resolve the implicit dependence by using the composition property of a propagator (cf. Eq. ~\ref{eq:PropagatorComposition2}). This allows us to rewrite the Magnus propagator as 
\begin{equation} \label{eq:MagnusAnsatzSplit}
    \text{exp}(\bA(t)) = \prod_{n=1}^{N}\exp(\bA_n) \,,
\end{equation}
where $\bA_n$ is given by Eq. ~\ref{eq:MagnusExpansion}, albeit with the limits of integration modified to $[t_{n-1},t_{n}]$.

In practice, one replaces the exact $\bA_n$ with an approximate operator $\bAtilde_n$, which involves, first, a truncation of the Magnus expansion (defined in Eq. ~\ref{eq:MagnusExpansion}), and second, an approximation for the time integrals in the truncated Magnus expansion. Truncating the Magnus expansion after the first $p$ terms results in a time-integration scheme of order $2p$. In this work, we restrict ourselves to the second-order Magnus propagator, i.e.,
obtained by truncating the Magnus expansion after the first term. Furthermore, we use a mid-point integration rule to evaluate $\int_{t_{n-1}}^{t_n}-i\bHbar(\tau)d\tau$. In particular, the action of the second-order Magnus propagator, with a mid-point integration rule, on the set of Kohn-Sham orbitals $\{\bpsibar_1,\bpsibar_2,\ldots,\bpsibar_{N_e}\}$ which defines the density $\rho^h(\br,t)$, is given by
\begin{equation} \label{eq:ApproxMagnus}
    \begin{split}
        e^{\bAtilde_n}\bpsibar\al(t) &= e^{-i \bHbar\left[\rho^h\left(t_{n-1}+\frac{\Delta t}{2}\right)\right]\Delta t}\bpsibar\al(t)\,,
    \end{split}
\end{equation}
where $\Delta t = t_n-t_{n-1}$ and $\bHbar\left[\rho^h\left(t_{n-1}+\frac{\Delta t}{2}\right)\right]$ is the time-continuous Kohn-Sham Hamiltonian described by $\rho^h(\br,t)$ at the future time instance $t_{n-1}+\Delta t/2$. We remark that $\bHbar\left[\rho^h\left(t_{n-1}+\frac{\Delta t}{2}\right)\right]$, being dependent on a future instance of the density, is evaluated either by an extrapolation of $\bHbar$ using $m (> 2)$ previous steps or by a second (or higher) order predictor-corrector scheme. 

Thus, time-discrete approximation to $\bpsibar\al(t_n)$, denoted by $\bpsibar\al^{n}$, is given by 
\begin{equation} \label{eq:PsiTimeDiscrete}
    \bpsibar\al^{n} = \text{exp}(\bAtilde_n)\bpsibar\al^{n-1}.
\end{equation}
Consequently, the orbitals $\psi\al^{h,n}(\br)$ defined by the coefficient vectors $\bpsi\al^n=\bM^{-1/2}\bpsibar\al^n$ represent the \emph{full-discrete} solution to the TDKS equation. 

\section{Discretization errors} \label{sec:Error}
In this section, we provide the discretization error in the Kohn-Sham orbitals which will later on form the basis of our efficient spatio-temporal discretization. To begin with, we decompose the discretization error in the Kohn-Sham orbitals into two parts, one arising due to spatial discretization and the other due to temporal discretization. 
To elaborate, if $\psi\al^h(\br,t_n)$ and $\psi\al^{h,n}(\br)$ represent the semi-discrete (discrete in space but continuous in time) and full-discrete solution to $\psi\al(\br,t_n)$, respectively, then one decompose the discretization error in $\bpsi\al(\br,t_n)$ as 
\begin{equation} \label{eq:ErrorSplit}
\begin{split}
    \psi\al(\br,t_n)-\psi\al^{h,n}(\br) &=\left(\psi\al(\br,t_n)-\psi\al^{h}(\br,t_n)\right) \\ 
    \quad & + \left(\psi\al^h(\br,t_n)-\psi\al^{h,n}(\br)\right)  
\end{split}
\end{equation}
In the following two subsections, we present error estimates on the two right-hand terms in the above equation,  respectively.

\subsection{Spatial discretization error}
If $\psi\al^h(\br,t)$ denotes the semi-discrete solution to $\psi\al(\br,t)$, then the following bound holds true (see Appendix for the derivation)
\begin{widetext}
        \begin{equation} \label{eq:SemiDiscreteErrorH1GS}
            \begin{split}
                \sum_{\alpha}^{N_e}\norm{\psi\al-\psi\al^h}_{\Hone} (t) & \leq  C(t) \sum_e h_e^{p}\sum_{\alpha=1}^{N_e}\left(\modulus{\psi\al}_{p+1,\Omega_e}(s_{1,\alpha}) + \modulus{\psi\al}_{p+1,\Omega_e}(s_{2,\alpha})  +  \modulus{\psi\al}_{p+3,\Omega_e}(s_{2,\alpha})\right) \\
                & \quad + C(t)\sum_e h_e^{p} \left(\modulus{V_H[\rho^h]}_{p+1,\Omega_e}(s_3) + \modulus{V_N}_{p+1,\Omega_e}\right)\,,
            \end{split}
        \end{equation}
        \begin{equation*}
            \text{for some}\, \{s_{1,\alpha}\} \,, \{s_{2,\alpha}\}\,,\text{and}\, s_3 \in [0,t]\,.
        \end{equation*}
\end{widetext}
In the above equations, $h_e$ and $\Omega_e$ denote the size and spatial-extent of the $e-$th finite-element, respectively. $C(t)$ is a time-dependent constant independent of the finite-element mesh. $\modulus{.}_{p,\Omega_e}$ is the semi-norm in $H^{p}(\Omega_e)$.
The importance of the above equations lies in relating the semi-discrete error to the mesh parameters (i.e., $h_e$ and $p$), and hence, is instrumental in obtaining an efficient spatial discretization (discussed in Sec. ~\ref{sec:OptimalDiscretization}). In particular, the above equation informs that the semi-discrete error in  $\norm{\psi\al-\psi\al^h}_{\Hone}$ decays as $\mathcal{O}({h_e^p})$.

\subsection{Time discretization error}
We now present the formal bounds on the time discretization error in $\psi\al(\br,t)$. Assuming each time interval $[t_{n-1},t_n]$ to be of length $\Delta t$, we obtain the following bound for the time-discretization error for a second-order Magnus propagator with a mid-point integration rule (see Appendix for the derivation)
\begin{equation} \label{eq:PsiTimeDiscreteError}
\begin{split}
        \norm{\psi\al^h(t_n)-\psi\al^{h,n}}_{\Ltwo} \leq 
        \quad C (\Delta t)^2 t_n \max_{0\leq t \leq t_n}\norm{\psi\al^h(t)}_{\Hone}\,.
\end{split}
\end{equation}
The essence of the above relation lies in relating the time-discretization error to $\norm{\psi\al^h(t)}_{\Hone}$, which in turn is related to the spatial discretization. Thus, the above equation, is crucial in selecting an efficient $\Delta t$, for a given finite-element mesh (see Sec. ~\ref{sec:OptimalDiscretization}).

\section{Efficient Spatio-temporal Discretization} \label{sec:OptimalDiscretization}
We now utilize our spatial and temporal discretization error estimates (Eqs. ~\ref{eq:SemiDiscreteErrorH1GS} and ~\ref{eq:PsiTimeDiscreteError}) to obtain an efficient spatio-temporal discretization. We follow along the lines of ~\cite{Radovitzky1999,Motamarri2012, Motamarri2013} to obtain an efficient spatial discretization by minimizing the semi-discrete error in the dipole moment at a given time, subject to a fixed number of finite-elements. We remark that the choice of dipole moment as an observable for this exercise is solely a matter of convenience, and any observable which can be inexpensively evaluated in terms of the density or the Kohn-Sham orbitals can be used instead. Representing the x-component of the continuous and the semi-discrete dipole moment at time $t$ as $\mu_x(t)$ and $\mu_x^h(t)$, respectively, we have 
\begin{equation} \label{eq:DipoleMomentSemiDiscrete1}
\begin{split}
|\mu_x(t)-\mu_x^h(t)| &\leq \norm{x}_{\Ltwo} \norm{\rho(t)-\rho^h(t)}_{\Ltwo} \\
&\leq C \norm{\rho(t)-\rho^h(t)}_{\Ltwo} \\
&\leq C \sum_{\alpha=1}^{N_e}\norm{\psi\al-\psi\al^h}_{\Hone}(t)\,.
\end{split}
\end{equation} 
Now using Eq. ~\ref{eq:SemiDiscreteErrorH1GS} in the above equation, results in   
\begin{widetext}
\begin{equation} \label{eq:DipoleMomentSemiDiscrete2}
    \begin{split}
        |\mu_x(t)-\mu_x^h(t)|  & \leq  C_{1}(t) \sum_e h_e^{p}\sum_{\alpha=1}^{N_e}\left(\modulus{\psi\al}_{p+1,\Omega_e}(s_{1,\alpha}) + \modulus{\psi\al}_{p+1,\Omega_e}(s_{2,\alpha})  +  \modulus{\psi\al}_{p+3,\Omega_e}(s_{2,\alpha})\right) \\
                & \quad + C_1(t)\sum_e h_e^{p} \left(\modulus{V_H[\rho^h]}_{p+1,\Omega_e}(s_3) + \modulus{V_N}_{p+1,\Omega_e}\right)\,,
    \end{split}
\end{equation}
for some $\{s_{1,\alpha}\},\{s_{2,\alpha}\},~\text{and}~s_3 \in [0,t]$.
We now use the definition of the semi-norm (in terms of partial spatial derivative) and introduce an element size distribution, $h(\br)$, to rewrite the above equation as 
\begin{equation} \label{eq:DipoleMomentSizeDist}
    \begin{split}
        |\mu_x(t)-\mu_x^h(t)|  & \leq  C_{1}(t) \int_{\Omega} h^{p}(\br) \left[\sum_{\alpha=1}^{N_e}\left(\Dbar^{p+1}\psi\al(s_{1,\alpha}) + \Dbar^{p+1}\psi\al(s_{2,\alpha})  + \Dbar^{p+3}\psi\al(s_{2,\alpha})\right)\right] \dr \\
                & \quad + C_1(t) \int_{\Omega} h^{p}(\br) \left(\Dbar^{p+1}V_H[\rho^h] + \Dbar^{p+1}V_N\right)\dr\,,
    \end{split}
\end{equation}
where $\Dbar^{k}f=\sum_{|l|=k}|D^lf|$.
Thus, obtaining the optimal element size distribution, for a fixed number of elements ($N_{elem}$), reduces to the following minimization problem,
\begin{equation} \label{eq:OptimalMeshMin}
    \begin{split}
        & \min_{h(\br)} \int_{\Omega} h^{p}(\br) \left[\sum_{\alpha=1}^{N_e}\left(\Dbar^{p+1}\psi\al(s_{1,\alpha}) + \Dbar^{p+1}\psi\al(s_{2,\alpha})  +  \Dbar^{p+3}\psi\al(s_{2,\alpha})\right)+\Dbar^{p+1}V_H[\rho^h](s_3) + \Dbar^{p+1}V_N\right] \dr \\
    &\qquad \qquad \qquad \qquad \qquad\qquad \qquad \text{subject to}: \int_{\Omega}\frac{\dr}{h^3(\br)}=N_{elem}\,.
    \end{split}
\end{equation}
Solving the Euler-Lagrange equation corresponding to the above optimization problem yields the following element size distribution,
\begin{equation} \label{eq:OptimalMeshSize}
    h(\br) = E \left[\sum_{\alpha=1}^{N_e}\left(\Dbar^{p+1}\psi\al(s_{1,\alpha}) + \Dbar^{p+1}\psi\al(s_{2,\alpha})  +  \Dbar^{p+3}\psi\al(s_{2,\alpha})\right)+\Dbar^{p+1}V_H[\rho^h](s_3) + \Dbar^{p+1}V_N\right]^{-1/p+3}\,,
\end{equation}
\end{widetext}
where the constant $E$ is evaluated from the constraint on the number of elements. We remark that although the above discretization approach requires \emph{a priori} knowledge of the unknown $\psi\al(s_{1,\alpha})$, $\psi\al(s_{2,\alpha})$, and $V_H[\rho^h](s_3)$, we use the ground-state atomic solutions to the Kohn-Sham orbitals and the electrostatic potentials to construct the adaptive finite-element mesh for all calculations. Although, this does not inform us about the optimal discretization required near the nuclei, this affords an efficient strategy to accurately handle the regions away
from the nuclei, wherein the time-dependent Kohn-Sham orbitals, typically, have similar decay properties as their ground-state counterparts. We note that, in practice, the finite-element mesh obtained can deviate from $h(\br)$ due to conformity and quality requirements, especially in the context of hexahedral elements that are employed in this work. However, the resulting finite-element mesh broadly captures the optimal coarse-graining rate, and has the general adaptive characteristics that significantly enhance the computational efficiency, as demonstrated in Sec.~\ref{sec:Results}. Figure~\ref{fig:MeshAl2} shows an adaptive mesh for all-electron $\text{Al}_2$ generated using this approach.

\begin{figure}[h!]
\centering     
\includegraphics[scale=0.24]{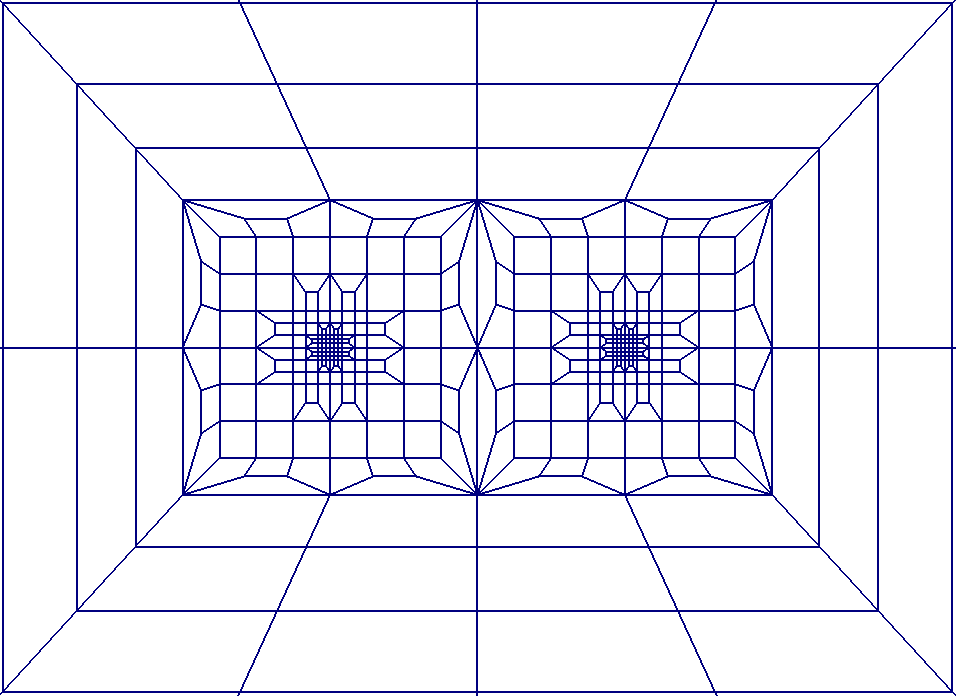}
\caption{\small{Adaptive finite-element mesh for all-electron $\text{Al}_2$ (slice shown on the plane of the molecule)}  
\label{fig:MeshAl2}}
\end{figure}

Having determined the required spatial discretization, a suitable temporal discretization for the given finite-element mesh can be estimated by using the time-discrete error bound for the dipole moment. To elaborate, if $\mu^{h,n}_{x}$ denotes the x-component of the full-discrete dipole moment, then using the result in Eq. ~\ref{eq:PsiTimeDiscreteError} it is straightforward to show,
\begin{equation} \label{eq:FullDiscreteDipoleMoment}
|\mu^h_x(t_n)-\mu^{h,n}_{x}|\leq C t_n (\Delta t)^2 \sum_{\alpha=1}^{N_e}\max_{0\leq t \leq t_n}\norm{\psi\al^h(t)}_{\Hone}\,.
\end{equation} 
As is evident from the above relation, our choice of $\Delta t$ is intrinsically tied to the spatial discretization through $\norm{\psi\al^h(t)}_{\Hone}$. Furthermore, we remark that although the presence of $t_n$ in the above inequality indicates increasing time-discretization error with time, it does not pose a limitation in containing the errors due to the fact that $t_n \leq T$, where, typically, $T$ lies between $10-30$ fs. Now, in order to evaluate a suitable $\Delta
t$ that can contain the above error bound to a fixed tolerance, we need to estimate the values of $\norm{\psi\al^h(t)}_{\Hone}$ and the value of the constant, $C$, featuring in it. The value of $\norm{\psi\al^h(t)}_{\Hone}$ can be reliably approximated from its ground-state value, i.e., $\norm{\psi\al^h(0)}_{\Hone}$. The characteristic value of the constant, $C$, is determined from atomic RT-TDDFT calculations at different $\Delta t$. To
elaborate, the constant can be evaluated from a linear fit to the log-log plot of the error $|\mu^h_x(t_n)-\mu^{h,n}_x|$ with respect to $\Delta t$. For a multi-atom system, we use the least $\Delta t$ obtained for each of its constituent atomic species.

\section{Numerical Implementation} \label{sec:Numerics}
We now discuss some of the key numerical aspects involved in our implementation of the finite-element discretization of the TDKS equations. 
\subsection{Higher-order spectral finite-elements} \label{sec:SpectralFE}
Finite-elements, with their varied choices of forms and orders ~\cite{Bathe2006,Hughes2012}, have been widely used in several engineering applications. While the use of linear finite-elements remains popular in engineering applications that warrant moderate levels of accuracy, it remains computationally inefficient to attain chemical accuracy in electronic structure calculations. To highlight, the use of linear finite-elements have been shown to require large number of basis functions per atom
($\sim 100,000$) to achieve chemical accuracy in ground-state DFT calculations ~\cite{Hermansson1986, Bylaska2009}. However, this shortcoming has been demonstrably mitigated by the use of higher-order finite-elements ~\cite{Motamarri2012, Motamarri2013}. In this work, we explore the possibility of similar gains from using higher-order finite-elements for RT-TDDFT calculations. Unlike  conventional finite-elements, we employ spectral finite-elements for the spatial discretization of the TDKS
equations. To elaborate, while the conventional finite-elements are constructed from a tensor product of Lagrange polynomials interpolated through equidistant nodal points in an element, spectral finite elements employ a distribution of nodes obtained from the roots of the derivative of Legendre polynomials or the Chebyshev polynomials ~\cite{Boyd2001}. In our work, we use the roots of the derivative of Legendre polynomials along with boundary nodes, so as to maintain $C^0$ continuity. The
resulting distribution is known as the Gauss-Legendre-Lobatto node distribution. The spectral finite-elements are known to afford better conditioning with increasing polynomial degree ~\cite{Boyd2001}, and have been used to gain significant computational efficiency in ground-state DFT calculations ~\cite{Motamarri2012}. Moreover, a major advantage of the spectral finite-elements is realized when used in conjunction with the Gauss-Legendre-Lobatto (GLL) quadrature rule for evaluating the integrals
involved in the overlap matrix $\mathbf{M}$, wherein the quadrature points are coincident with the nodal points in the spectral finite-elements. Such a combination renders $\mathbf{M}$ diagonal, thereby greatly simplifying the evaluation of $\mathbf{M}^{-1/2}$ that features in the discrete TDKS equations (cf. Eq. ~\ref{eq:TDKSMatVecStandard}). We note that while an $n$ point rule in the conventional Gauss quadrature rule integrates polynomials exactly up to degree $2n-1$, an $n$ point GLL
quadrature rule integrates polynomials exactly only up to degree $2n-3$. Therefore, we employ the GLL quadrature rule only in the construction of $\mathbf{M}$, while the more accurate Gauss quadrature rule is used for all other integrals featuring in the discrete TDKS equations. We refer to Motamarri et. al. ~\cite{Motamarri2013} for a discussion on the accuracy and sufficiency of GLL quadrature in the evaluation of overlap matrix $\mathbf{M}$. For the sake of brevity, we use the term finite-elements instead of spectral finite-elements in all subsequent discussions.

\subsection{Approximating the second-order Magnus operator}
The form of the Magnus operator, as shown in Eq. ~\ref{eq:PsiTimeDiscrete}, calls for efficient means of evaluating $\text{exp}(\bAtilde_n)\bpsibar\al^n$. Direct means of evaluating the matrix $\text{exp}(\bAtilde_n)$ remain computationally prohibitive beyond small sizes of $\bAtilde_n$. Alternatively, one can attempt to evaluate the action of $\text{exp}(\bAtilde_n)$ on $\bpsibar\al^n$ in an iterative fashion. Several such schemes are in use in RT-TDDFT calculations, namely, Taylor series expansion,
Chebyshev polynomial expansion, split-operator techniques, and Krylov subspace projection method. We refer to Castro et. al. ~\cite{Castro2004} and references there-within for a detailed discussion on each of these schemes. In this work, we adopt the Krylov subspace projection method for its superior efficiency and robustness compared to the other methods. To elaborate, the Krylov subspace projection allows for an \emph{a posteriori} error control mechanism based on error estimates that will be
described below. On the other hand, polynomial expansion methods such as Taylor series or Chebyshev polynomial expansion offer no such \emph{a posteriori} mechanism. While one can use \emph{a priori} estimates, based on the spectral radius of $\bAtilde_n$, the number of terms required in the polynomial expansion, for a desired accuracy, remain highly over-estimated. Furthermore, in the case of split-operator, the efficacy of it rests on operating back and forth between Fourier and real
space, so as to diagonalize the kinetic and the potential part of the Kohn-Sham Hamiltonian in succession. Thus, it involves the use of fast Fourier transforms (FFTs) which pose serious challenges to the parallel scalability of the code. The Krylov subspace projection, on the other hand involves operations only in real space and affords good parallel scalability (as will be shown in Sec. ~\ref{sec:Scalability}). 

We now discuss the details of the Krylov subspace projection approach for the second-order Magnus operator. To begin with, a $k$-dimensional Krylov subspace for the matrix $\bAtilde_n$ and the vector $\bpsibar$ (a generic representation for the Kohn-Sham vectors $\bpsibar\al^n$) is given by
\begin{equation} \label{eq:KrylovSubspace}
\mathcal{K}_k(\bAtilde_n,\bpsibar)=span\{\bpsibar,\bAtilde_n\bpsibar,\bAtilde_n^2\bpsibar,\ldots,\bAtilde_n^{k-1}\bpsibar\}.
\end{equation}
The Lanczos iteration provides a recipe for generating an orthonormal set of vectors $\bQ_k=\{q_1,q_2,\ldots,q_k\}$, with $q_1=\bpsibar/\norm{\bpsibar}$, that spans the same space as $\mathcal{K}_k(\bAtilde_n,\bpsibar)$. In particular, the Lanczos iteration, allows for the following approximation to $e^{\bAtilde_n}\bpsibar$, denoted by $\bz_k \in  \mathcal{K}_k(\bAtilde_n,\bpsibar)$, given by  
\begin{equation} \label{zk}
    \bz_k=\norm{\bpsibar}\bQ_k e^{\bQ_k^T \bAtilde_n \bQ_k} e_1 = \norm{\bpsibar}\bQ_k e^{\bT_k} e_1\,,
\end{equation}
where $\bT_k=\bQ_k^T\bAtilde_n\bQ_k$ is a tridiagonal matrix, and $e_1$ is the $i$-th unit vector in $\mathbb{C}^k$. As is evident from the above form, the problem is now reduced to the evaluation of $\text{exp}(\bT_k)$, wherein $\bT_k$ is a small matrix of size $k \times k$, and hence, $\text{exp}(\bT_k)$ can be evaluated inexpensively either through Taylor series expansion or exact eigenvalue decomposition of $\bT_k$. The error, $\epsilon_k$, incurred in the above approximation is given by ~\cite{Hochbruck1998}
\begin{equation}\label{eq:AdaptiveLanczos}
\epsilon_k=\norm{e^{\bAtilde_n}\bpsibar - \norm{\bpsibar}\bQ_k e^{\bT_k}e_1} \approx \beta_{k+1,k}\norm{\bpsibar}\modulus{\left[e^{\bT_k}\right]_{k,1}}\,,
\end{equation} 
where $\beta_{k+1,k}$ is the $(k+1,k)$ entry of $\bT_{k+1}=\bQ_{k+1}^T\bAtilde_n\bQ_{k+1}$.
Thus, the above relation provides a robust and inexpensive scheme to adaptively determine the dimension of the Krylov subspace by checking if $\epsilon_k$ is below a set tolerance. An economic choice for the tolerance for $\epsilon_k$ is determined from atomic RT-TDDFT calculations, such that it achieves $<10$ meV accuracy in the excitation energies. For a multi-atom system, we employ the lowest such tolerance obtained for each of the constituent atomic species.

Finally, we comment upon the numerical details of the second-order Magnus propagator with midpoint integration rule. As discussed in Sec. ~\ref{sec:FullDiscrete}, the use of second-order Magnus propagator with midpoint integration rule, i.e., $e^{\bAtilde_n}$, requires the knowledge of $\bHbar$ at a future time instant i.e., $\bHbar[t_{n-1}+\Delta t/2]$, which is \emph{a priori} unknown. A fully consistent approach involves, for a given $\bpsibar^{n-1}$, the following steps: (i) approximate
$\bHbar[t_{n-1}+\Delta t/2]$ through extrapolation over previous instants of $\bHbar$; (ii) use it to obtain $\bpsibar^{n}$, and then evaluate $\bHbar[t_{n}]$; (iii) re-evaluate $\bHbar[t_{n-1}+\Delta t/2]$ by interpolating between $\bHbar[t_{n-1}]$ and $\bHbar[t_{n}]$; and (iv) repeat steps (ii)--(iii) until convergence. Although robust and accurate, this approach comes at a huge computational cost arising out of the Lanczos procedure at each iterate of the
self-consistent iteration. An efficient and sufficiently accurate approach is to use a predictor-corrector method to, first, predict $\bHbar[t_{n-1}+\Delta t/4]$ through an extrapolation (linear or higher-order) from previous instants of $\bHbar$, use it to propagate $\bpsibar^{n-1}$ to $\bpsibar^{n-1/2}$, which is then used to evaluate $\bHbar[t_{n-1}+\Delta t/2]$. We refer to ~\cite{Cheng2006} for the details of the predictor-corrector scheme. We remark that this predictor-corrector scheme is accurate to $\mathcal{O}(\Delta t^2)$, and hence, does not affect the results of our time-discretization error estimates. 


\section{Results} \label{sec:Results}
In this section, we discuss the accuracy, rate of convergence, computational efficiency and the parallel scalability of higher-order finite-element discretization in conjunction with second-order Magnus propagator, for both pseudopotential and all-electron RT-TDDFT calculations. Based on the system, we use hexahedral spectral finite-elements of polynomial order $1$ to $5$, denoted as HEX8, HEX27, HEX64SPEC, HEX125SPEC, and HEX216SPEC, respectively. For the pseudopotential calculations, we provide
comparison, in terms of accuracy and
performance, of the higher-order finite-elements against the finite-difference method. The finite-difference based calculations are performed using the Octopus ~\cite{Octopus2006} software package. In all our finite-difference based calculations, we have used a stencil of order $4$ in each direction (default stencil order in Octopus). All the pseudopotential calculations are done using the norm-conserving Troullier-Martins 
pseudopotentials~\cite{Troullier1991}. For all calculations, the ground-state Kohn-Sham orbitals are used as the initial states. We use the Chebyshev filter acceleration technique (refer ~\cite{Zhou2006a, Zhou2006b, Motamarri2013}) to efficiently compute the ground-state, for all the calculations done using finite-elements. All our scalability as well as benchmark studies demonstrating the computational efficiency are conducted on a parallel computing cluster with the following configuration: Intel Xeon
Platinum 8160 (Skylake) CPU nodes with 48 processors
(cores) per node, 192 GB memory per node, and Infiniband networking between all nodes for fast MPI communications.

\subsection{Rates of Convergence} \label{sec:Convergence}
In this section, we study the rates of convergence of the dipole moment with respect to both finite-element mesh-size, $h$, as well as time-step, $\Delta t$. We use methane and lithium hydride molecules as our benchmark systems for
studying the rates of convergence, for pseudopotential and all-electron calculations, respectively.

We note that in order to study the convergence with respect to mesh-size, the dominant error must arise from spatial-discretization. To this end, we contain other sources of error, namely, time-discretization error and Krylov subspace projection error, by choosing a very small time-step of $\Delta t= 10^{-4}$, and using a small tolerance of $10^{-12}$ for the Krylov subspace error (cf. Eq. ~\ref{eq:AdaptiveLanczos}). In effect, we mimic a semi-discrete (discrete in space but continuous in time) error analysis. We employ finite-elements of three different orders ($p$)---HEX8, HEX27, and HEX125SPEC---in all our convergence studies. For each $p$, we start with a given $N_{elem}$ and use the superposition of ground-state atomic orbitals and electrostatic potentials, to determine the coarsening rate of the mesh, as per Eq.~\ref{eq:OptimalMeshSize}. The resultant mesh forms the coarsest mesh. Subsequently, we increase the value of $N_{elem}$ to obtain a sequence of increasingly refined meshes.
We remark that although we propose for using the ground-state electronic fields to determine the mesh coarsening rate, it, nevertheless, forms a reliable and cost effective way for discretizing the mesh as opposed to any ad-hoc coarsening or using uniform discretization. Furthermore, it allows us to use large computational domain sizes without significantly increasing the number of elements. This, in turn, allows us to circumvent the need for artificial absorbing boundary conditions, otherwise essential to tackle wave reflection effects that are observed while dealing with small computational domains.

In order to perform the convergence study with respect to mesh-size, and compare to the semi-discrete error estimate obtained in Eq. ~\ref{eq:DipoleMomentSemiDiscrete2}, we require the knowledge of the continuous value of the dipole moment, $\mu_x(t)$. To this end, we use the discrete dipole moment $\mu_x^h(t)$, obtained from a sequence of increasingly refined HEX125SPEC finite-element meshes to obtain a least-square fit of the form
\begin{equation} \label{eq:DipoleSemiDiscreteLeastSquare}
    \frac{\modulus{\mu_{x}(t)-\mu_x^h(t)}}{\modulus{\mu_{x}(t)}}=Ch_{min}^q\,.
\end{equation}
to determine $\mu_{x}(t)$, $C$, and $q$. In the above equation $h_{min}$ represents the minimum element size, $h_e$, present in the mesh. The obtained $\mu_{x}(t)$ represents the extrapolated continuum limit (continuous in space) for the dipole moment computed using HEX125SPEC element, and is used as the reference value to compute $\frac{\modulus{\mu_{x}(t)-\mu_x^h(t)}}{\modulus{\mu_{x}(t)}}$ for HEX8 and HEX27 finite-elements.

Next, we consider the convergence with respect to temporal discretization, i.e., $\Delta t$. To this end, we use a sufficiently refined HEX125SPEC finite-element mesh and use increasingly refined $\Delta t$ to obtain a least-square fit of the form
\begin{equation} \label{eq:DipoleFullDiscreteLeastSquare}
    \frac{\modulus{\mu^{h}_x(t_n)-\mu^{h,n}_{x}}}{\modulus{\mu^{h}_x(t_n)}} = C (\Delta t)^q \,.
\end{equation}
to determine  $\mu^{h}_x(t_n)$, $C$, and $q$. The value of $\mu^{h}_x(t_n)$ obtained from the above equation represents the extrapolated continuum  
limit (continuous in time) for the dipole moment at $t_n$.

\subsubsection{All-electron calculations: Lithium Hydride} \label{sec:LiH}
In this example, we conduct all-electron RT-TDDFT study on a lithium hydride molecule (LiH) with Li-H bond-length of 3.014 a.u.\,. A large cubical domain of length of 50 a.u. is chosen to ensure that the electron density decays to zero on the domain boundary, thereby, allowing us to impose Dirichlet boundary condition on the time-dependent Kohn-Sham orbitals and the Hartree potential. We use the aforementioned adaptive mesh generation strategy to construct a sequence of HEX8, HEX27 and HEX125SPEC meshes.
For all the meshes under consideration, we first obtain the ground-state and employ a weak delta-kick to excite the system. To elaborate, we
use an electric field of the form $\textbf{E}_0(t)=\kappa\delta(t)\hat{x}$, with $\kappa=10^{-3}$ a.u., where $\delta(t)$ is the Dirac-delta distribution and $\hat{x}$ is the unit vector along x-axis. This amounts to perturbing the ground-state Kohn-Sham orbitals, $\psi\al^{GS}$, by a factor $e^{-i\kappa x}$. Thus, our initial-states are defined as $\psi\al(0)=e^{-i\kappa x}\psi\al^{GS}$. Figure ~\ref{fig:LiHSpatial} depicts the rates of convergence for the dipole moment at $t=1.0$ a.u. for different orders of finite-elements. For all the three types of finite-elements under consideration, we observe close to optimal rates of convergence, $\mathcal{O}(h^p)$, where $p$ is the
degree of the finite-element interpolating polynomial. As is evident from Figure ~\ref{fig:LiHSpatial}, much higher accuracies are obtained with HEX125SPEC when compared to HEX8 and HEX27 of the same mesh size. In particular, $200,000$ HEX8 elements ($210,644$ degrees of freedom) are required to achieve relative errors of $10^{-2}$, whereas we achieve relative error of $10^{-3}$ with just $3,000$ HEX125SPEC elements ($83,156$ degrees of freedom). Next, we study the rate of convergence of the dipole moment with respect to temporal discretization. To
this end, we used a sufficiently refined HEX125SPEC mesh which affords $10^{-4}$ relative error with respect to spatial discretization. We then propagate the initial-states using second-order Magnus propagator with different $\Delta t$. Figure ~\ref{fig:LiHTemporal} depicts the rate of convergence of the dipole moment with respect to $\Delta t$ at $t_n=1.0$ a.u.. We obtain a rate of convergence of $q=1.96$ (defined in Eq. ~\ref{eq:DipoleFullDiscreteLeastSquare}), which agrees remarkably well with the quadratic rate of convergence for second-order Magnus propagator (cf. Eq. ~\ref{eq:FullDiscreteDipoleMoment}). 
\begin{figure}[h!]
    \begin{minipage}{.48\textwidth}
        \begin{center}
            \includegraphics[scale=0.6]{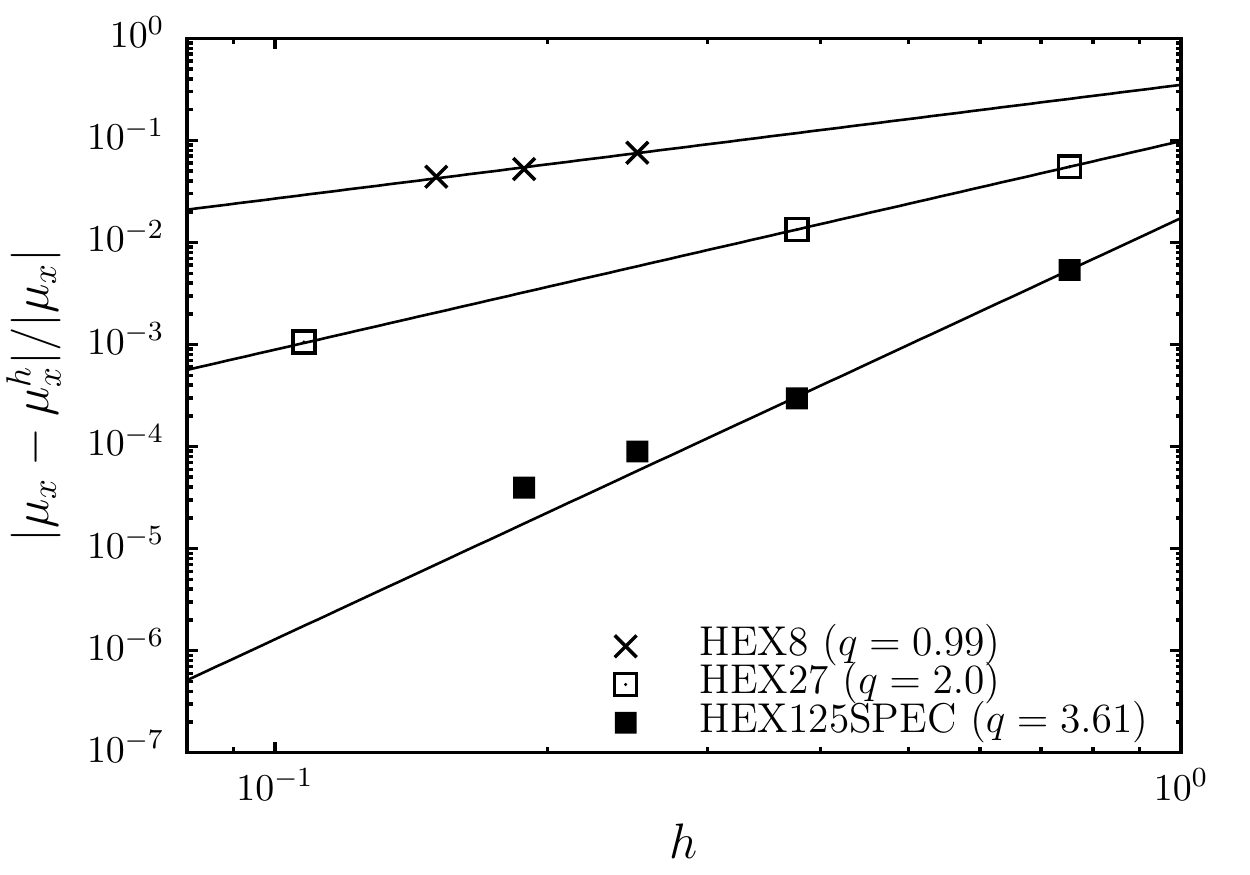}
              \caption{\small Rates of convergence with respect to spatial discretization for LiH}
                \label{fig:LiHSpatial}
            \end{center}
    \end{minipage}%
    \hfill
    \begin{minipage}{.48\textwidth}
        \begin{center}
            \includegraphics[scale=0.6]{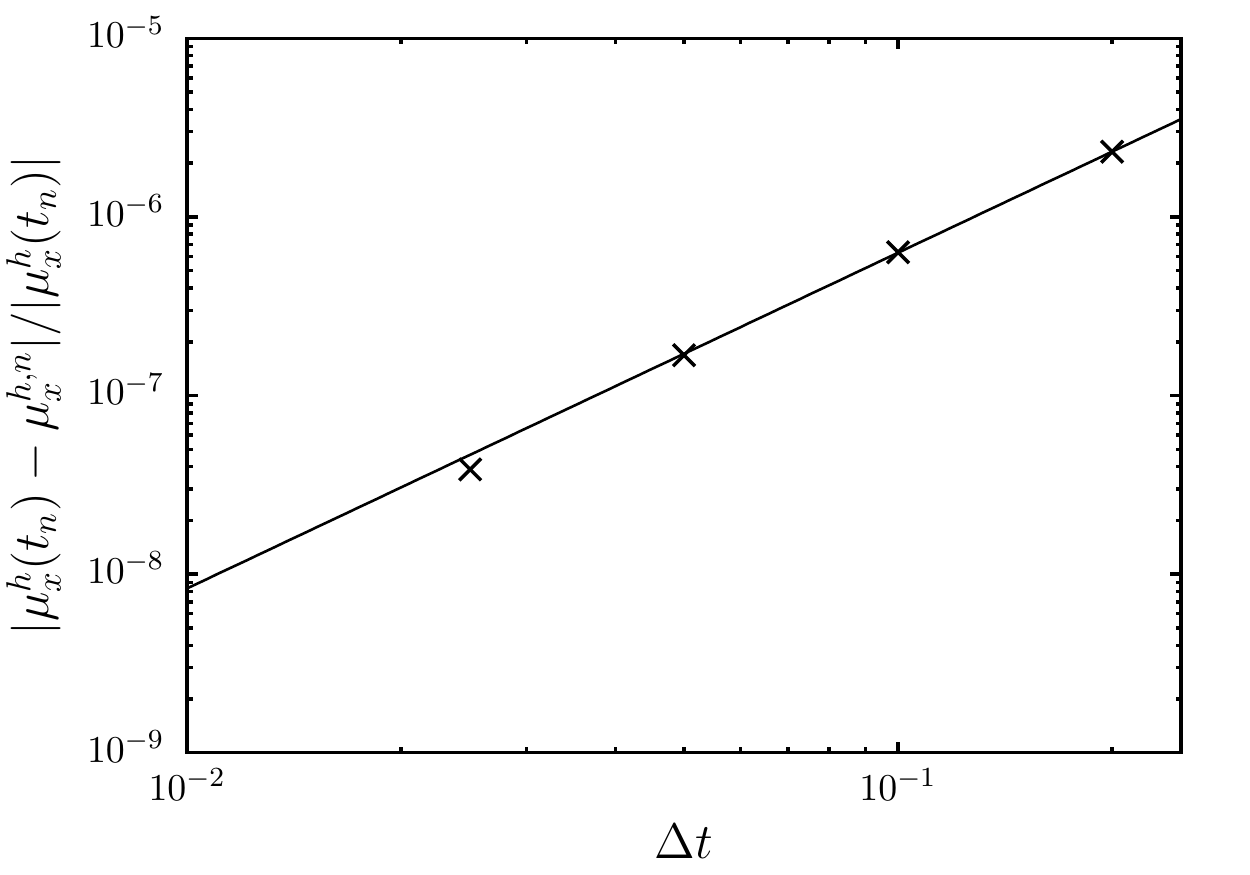}
              \caption{\small Rate of convergence with respect to temporal discretization for LiH}
                \label{fig:LiHTemporal}
            \end{center}
    \end{minipage}
\end{figure}

\subsubsection{Pseudopotential calculations: Methane ($\text{CH}_4$)}\label{sec:CH4}
We now turn to examining rates of convergence for the pseudopotential case. We use a methane molecule with C-H bond-length of 2.07846 a.u. and a H-C-H tetrahedral angle of $109.4712^{\circ}$ as our benchmark system. Similar to lithium hydride, we use the ground-state single-atom electronic fields to obtain a sequence of adaptively refined HEX8, HEX27, and HEX125SPEC meshes.
We, once again, make use of a large cubical domain of length 50 a.u. to mimic simulations in $\mathbb{R}^3$. For all the meshes, we first, obtain the ground-state and then excite the system using a Gaussian electric field of the form $\textbf{E}_0(t)=\kappa e^{(t-t_0)^2/w^2} \hat{x}$, with
$\kappa=2\times 10^{-5}$ a.u., $t_0=3.0$ a.u., and $w=0.2$ a.u.\,. Figure ~\ref{fig:CH4Spatial} illustrates the rates of convergence of the dipole moment at $t=5.0$ a.u. for different orders of finite-elements. As in the case of lithium hydride, we obtain close to optimal rates of convergence, and observe significantly higher accuracies for HEX125SPEC over HEX8 and HEX27. Next, we study the rate of convergence afforded by the second-order Magnus propagator with respect to the time-step using a sufficiently refined HEX125SPEC mesh. We propagate the ground-state Kohn-Sham orbitals under the  influence of the same Gaussian electric field using different $\Delta t$. Figure ~\ref{fig:CH4Temporal} shows the rate of convergence of the dipole moment with respect to $\Delta t$ at $t_n=5.0$ a.u.\,. As was the case with lithium hydride, we obtain a convergence rate of $q=1.98$, which is remarkably close to the optimal (i.e., quadratic) rate of convergence (cf. Eq.~\ref{eq:DipoleFullDiscreteLeastSquare}). 
\begin{figure}[h!]
    \begin{minipage}{.48\textwidth}
        \begin{center}
            \includegraphics[scale=0.6]{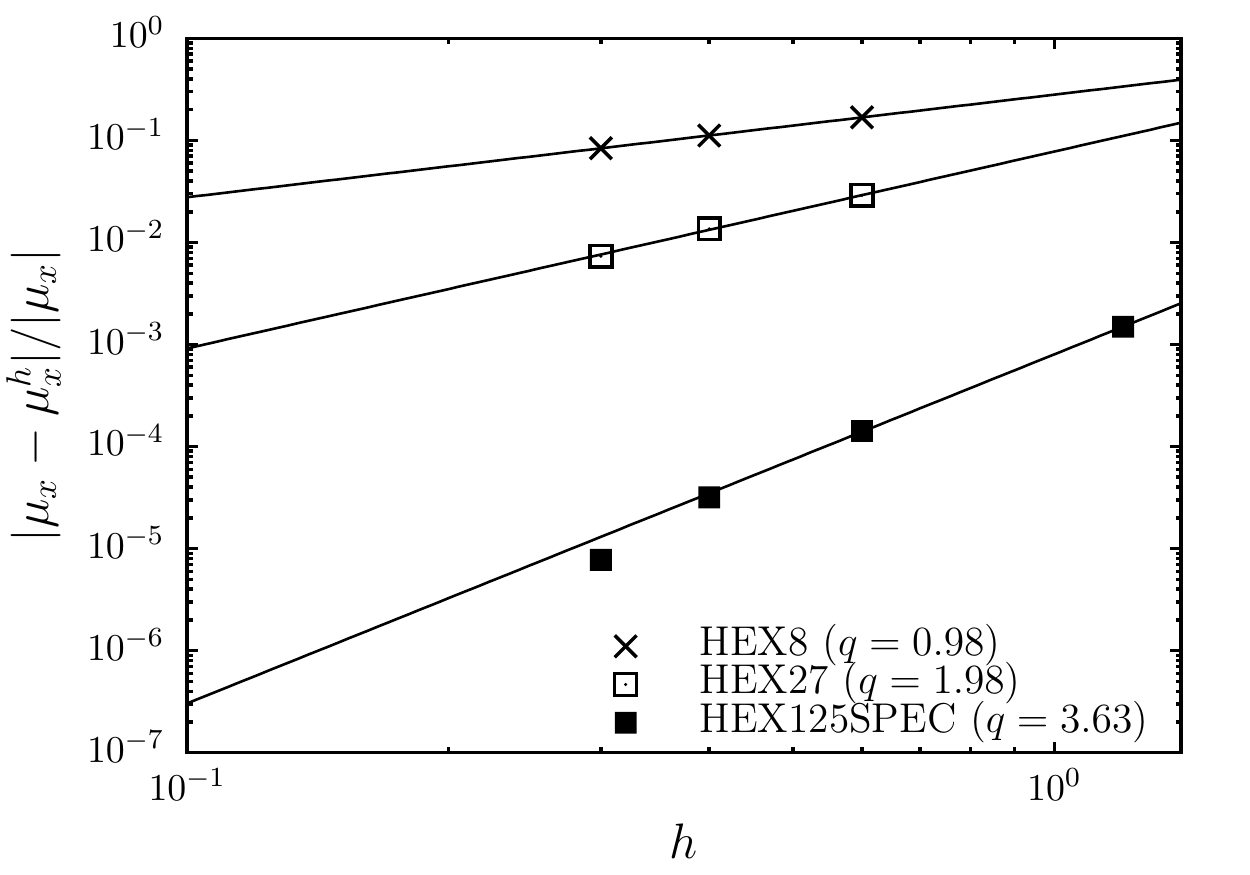}
              \caption{\small Rates of convergence with respect to spatial discretization for $\text{CH}_4$}
                \label{fig:CH4Spatial}
            \end{center}
    \end{minipage}%
    \hfill
    \begin{minipage}{.48\textwidth}
        \begin{center}
            \includegraphics[scale=0.6]{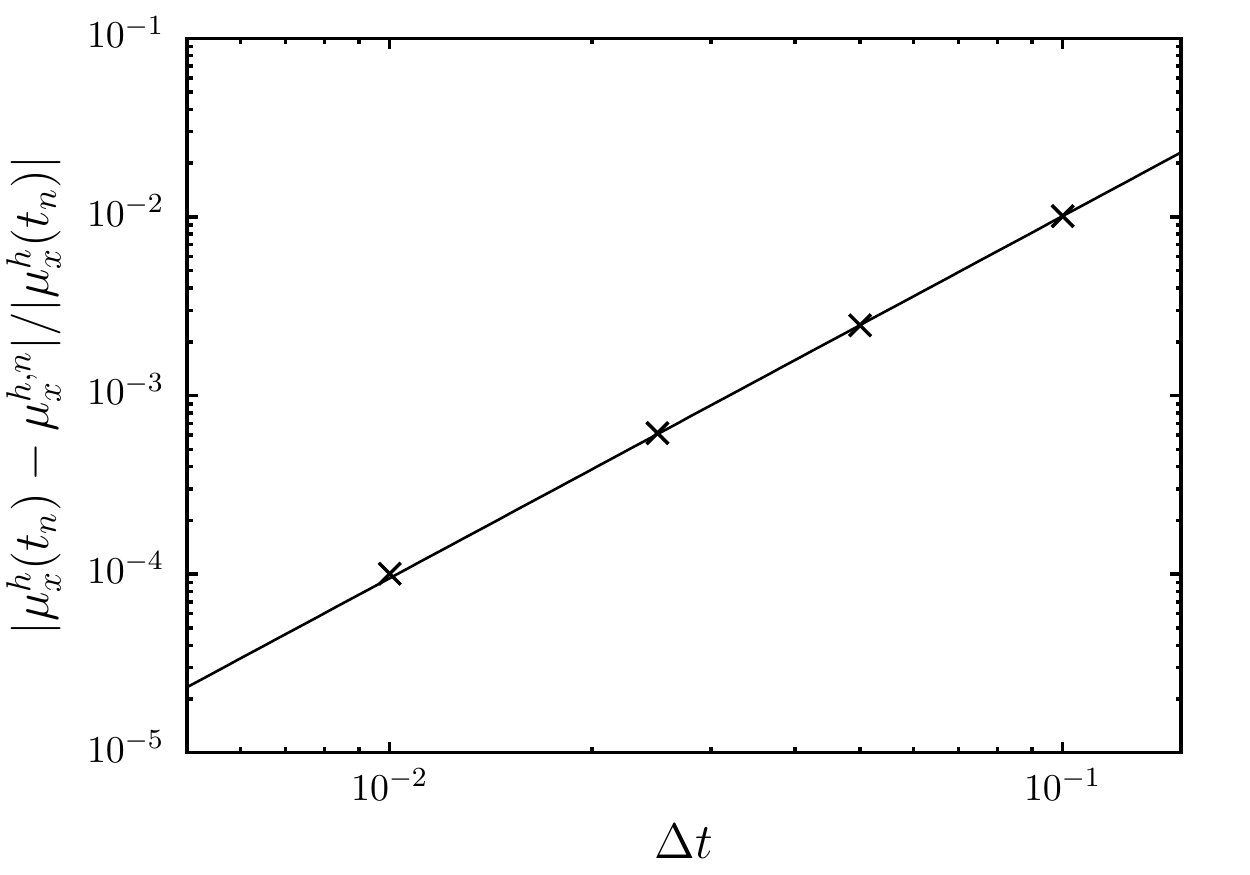}
              \caption{\small Rate of convergence with respect to temporal discretization for $\text{CH}_4$}
                \label{fig:CH4Temporal}
            \end{center}
    \end{minipage}
\end{figure}

\subsection{Computational Cost} \label{sec:ComputationalCost}
In this section, we investigate the relative computational efficiency afforded by higher-order finite-elements over linear finite-element. We consider the previous two systems, lithium hydride and methane, for all-electron and pseudopotential calculations, respectively. We use the same mesh adaption strategy as detailed in Sec. ~\ref{sec:Convergence}. Since the objective of this study is to compare the relative performance of various orders of finite-elements, we eliminate any time-discretization effect by setting $\Delta t= 10^{-4}$. Furthermore, we use a tolerance of $10^{-12}$ for the adaptive Lanczos (cf. Eq.~\ref{eq:AdaptiveLanczos}) in order to eliminate any Krylov subspace projection error influencing the spatial discretization error. We repeat the previous numerical studies by exciting the lithium hydride molecule with a delta-kick (see Sec. ~\ref{sec:LiH}), and the methane molecule with a Gaussian electric field (see Sec.
~\ref{sec:CH4}). Figures ~\ref{fig:LiHTime} and ~\ref{fig:CH4Time} show the relative error in the dipole moment against the normalized computational time, for three different orders of finite-elements. The
normalization of the computational time is done with respect to the longest time among the various meshes under consideration. As is evident, the relative computational efficiency afforded by higher-order finite-elements improves as the desired accuracy is increased. In particular, for a relative accuracy of $10^{-3}$, HEX125SPEC outperforms HEX8 and HEX27 by factor $150-200$ and $10-18$, respectively. This underscores the efficacy of higher-order finite-elements for RT-TDDFT calculations, an aspect which had, heretofore, remained unexplored for RT-TDDFT.  
\begin{figure}[h!]
    \begin{minipage}{.48\textwidth}
        \begin{center}
            \includegraphics[scale=0.6]{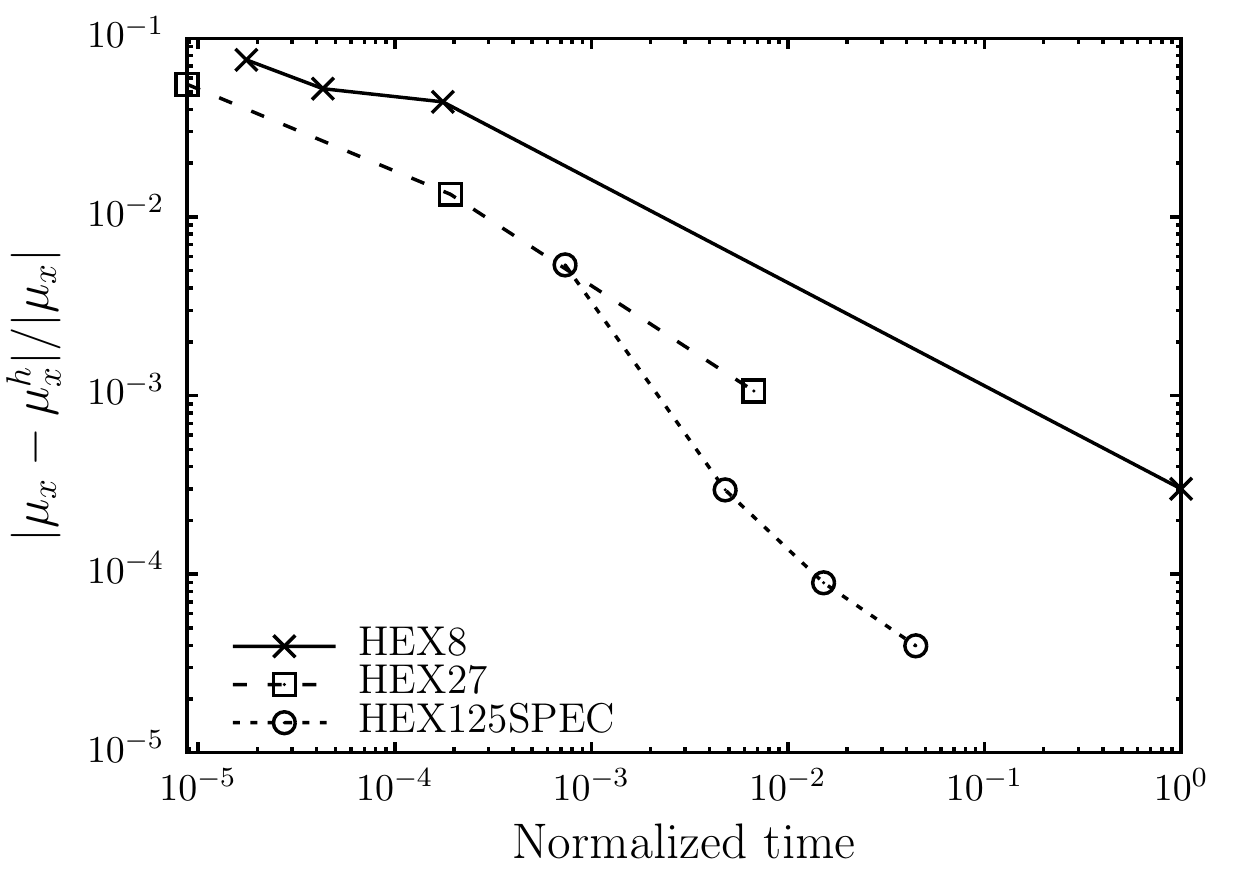}
              \caption{\small Computational efficiency of various orders finite-elements for LiH}
                \label{fig:LiHTime}
            \end{center}
    \end{minipage}%
    \hfill
    \begin{minipage}{.48\textwidth}
        \begin{center}
            \includegraphics[scale=0.6]{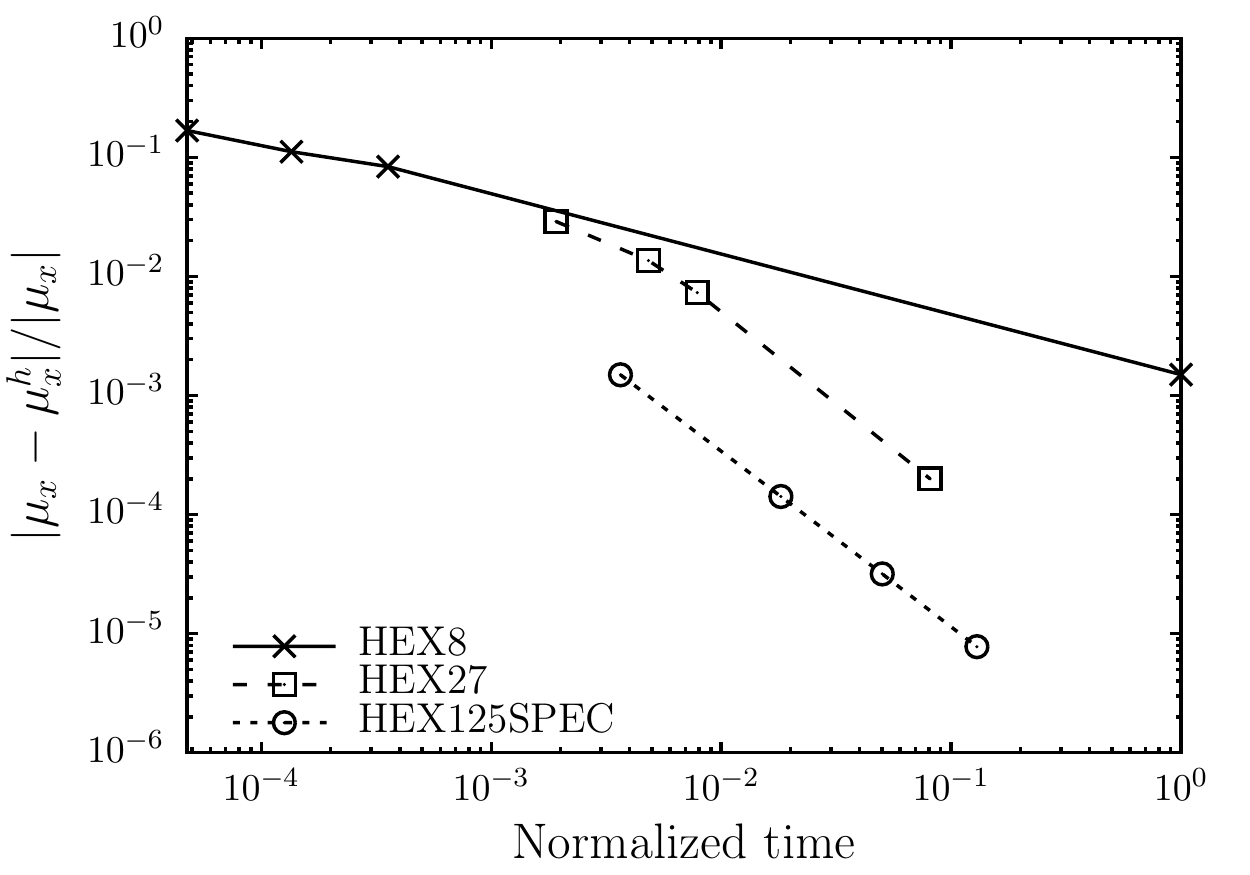}
              \caption{\small Computational efficiency of various orders finite-elements for $\text{CH}_4$}
                \label{fig:CH4Time}
            \end{center}
    \end{minipage}
\end{figure}

\subsection{Other materials systems} \label{sec:LargeSystems}
In this section, we investigate the accuracy and computational efficiency afforded by higher-order finite-elements for other materials systems, in both pseudopotential and all-electron RT-TDDFT calculations. We use $\text{Al}_{2}$, $\text{Al}_{13}$, and $\text{Mg}_2$ as the benchmark metallic systems for pseudopotential calculations. Furthermore, we use Buckminsterfullerene ($\text{C}_{60}$) as our benchmark insulating system for pseudopotential calculations. For the all-electron case, we use methane and benzene as our benchmark systems. Additionally, for the all-electron calculations we provide comparison, in the absorption spectrum, with their pseudopotential counterparts. For all the above systems under consideration, except $\text{Mg}_2$, we use weak electric fields to excite them. For $\text{Mg}_2$, we use a strong laser pulse to study the efficacy of higher-order finite-elements for nonlinear response. Table~\ref{tab:Params} lists the important simulation parameters, for all the benchmark systems under consideration. For pseudopotential systems, we also provide comparison, wherever possible, against calculations based on a finite-difference discretization, by employing the same propagator (i.e., second-order Magnus) and simulation details (as listed in Table~\ref{tab:Params}). To this end, we use the Octopus ~\cite{Octopus2006} software package to perform the finite-difference based calculations.

We now briefly discuss about the choice of domain sizes in our calculations. Typically, one needs a larger domain for RT-TDDFT calculations than ground-state calculations, so as to avoid reflection  at  the  domain  boundaries. For finite-elements, owing to adaptive meshing capability, choosing a large enough domain has little bearing on its computational expense. However, for finite-difference, wherein Octopus uses a uniform mesh, the use of large domain sizes can significantly effect its computational cost. In  order  to  obtain  a  suitable grid in Octopus, we first obtain the optimal grid-spacing and domain size that achieves an accuracy of 10 meV in the  ground-state  energy  per  atom, commensurate  with the accuracy targeted in the finite-element discretization. We then increase the domain size until it achieves $<10$ meV accuracy in the excitation energies (defined in Sec~\ref{sec:Al2}). The calculation based on the resulting Octopus mesh is considered as the point of comparison (for both accuracy and efficiency) against the corresponding finite-elements based calculation. We add that, while dealing with uniform mesh, a typical workaround to the large domain requirement is to use a smaller domain with absorbing  boundary conditions. Hence, to better assess the effects of absorbing boundary conditions, we employ them in finite-difference based calculations for some of the benchmark systems discussed below.    
\begin{table} [h!]
\caption{\small Simulation details for both pseudopotential (PSP) and all-electron (AE) benchmark systems: Type of the electric field $\mathbf{E}_0(t)$; time-step ($\Delta t$ in a.u.); tolerance for Krylov subspace projection error ($\epsilon$ , cf. Eq.~\ref{eq:AdaptiveLanczos}); total duration of simulation ($T$ in fs)} 
\begin{threeparttable}
\begin{tabular}{M{0.25\columnwidth}M{0.3\columnwidth}M{0.12\columnwidth}M{0.12\columnwidth}M{.12\columnwidth}}
\hline 
\hline
    System & Field type & $\Delta t$ & $\epsilon$ & $T$ \\ \hline
    $\text{Al}_2$ (PSP) & Weak-Gaussian\tnote{1} & $0.05$ & $10^{-8}$ & $10$ \\ \hline
    $\text{Al}_{13}$ (PSP) & Weak-Gaussian\tnote{1} & $0.05$ & $10^{-8}$ & $10$ \\ \hline
    $\text{C}_{60}$ (PSP) & Weak-Gaussian\tnote{1} & $0.05$ & $10^{-8}$ & $10$ \\ \hline
    $\text{Mg}_2$(PSP)& Strong-Sinusoidal\tnote{2} & $0.025$ & $10^{-8}$ & $25.33$ \\ \hline
    $\text{CH}_{4}$ (PSP) & Weak-Gaussian\tnote{1} & $0.05$ & $10^{-8}$ & $10$ \\ \hline
    $\text{CH}_{4}$ (AE) & Weak-Gaussian\tnote{1} & $0.025$ & $10^{-8}$ & $10$ \\ \hline
    $\text{C}_6\text{H}_{6}$ (PSP) & Weak-Gaussian\tnote{1} & $0.05$ & $10^{-8}$ & $10$ \\ \hline
    $\text{C}_6\text{H}_{6}$ (AE) & Weak-Gaussian\tnote{1} & $0.025$ & $10^{-8}$ & $10$ \\ \hline \hline
\end{tabular}
   \begin{tablenotes}
     \item[1] $\mathbf{E}_0(t)=\kappa e^{(t-t_0)^2/\omega^2}\hat{x}$, with $\kappa=2\times10^{-5}$, $t_0=3.0$, and $\omega=0.2$ (all in a.u.). \\ 
     \item[2] $\mathbf{E}_0(t) = \kappa \text{sin}^2(\pi/T)\text{sin}(\omega t)\hat{x}$, with $\kappa = 0.01$, $\omega=0.03$, $T=5\times(2\pi/\omega)$ (all in a.u.).  
   \end{tablenotes}
   \end{threeparttable}
   \label{tab:Params}
\end{table}

\subsubsection{Pseudopotential calculations: $\text{Al}_2$}\label{sec:Al2}
We consider an aluminum dimer ($\text{Al}_2$) of bond-length $4.74$ a.u.\,. In order to generate a suitable mesh, we use an adaptive HEX64SPEC finite-elements discretization that follows the coarsening rate obtained from Eq.~\ref{eq:OptimalMeshSize} and is commensurate with an accuracy of $10$ meV
 in the ground-state energy per atom. We use a cubical domain of length 60 a.u. to ensure that the wavefunctions decay to zero, and thereby, avoid any reflection effects. We excite the ground-state using the simulation parameters listed in Table~\ref{tab:Params}. We use the Fourier transform of the dipole moment to obtain the dynamic polarizability, $\alpha_{a,b}(\omega)$, where $a$ is the index of the electric field's polarization direction and $b$ is the index of the measurement
direction
of the dipole. Subsequently, we obtain the absorption spectrum (dipole strength function), $S(\omega)$, given by $S(\omega)=\frac{2\omega}{3\pi}\text{Tr}\left[\text{Im}[\boldsymbol{\alpha}(\omega)]\right]$.
The peaks in the absorption spectrum correspond to the excitation
energies. We also assess the performance of higher-order finite-elements by comparing against the finite-difference scheme of Octopus~\cite{Octopus2006}.  
In order to highlight the effects of domain size for the finite-difference mesh, we use two cubical domains of sizes 38 a.u and 46 a.u, both with a grid-spacing of 0.2 a.u.\,. Furthermore, to understand the effect of absorbing boundary conditions, we perform
an additional finite-difference calculation on the 38 a.u. mesh with a negative imaginary potential (NIP) near the boundaries. In particular, we use a potential of the following form
\[
V_{NIP}(\bx)=
	 \begin{cases} 
      0 & |x|\leq L \\
      -i\, \eta\, \text{sin}^2\left(\frac{2\pi (x-L)}{L}\right) & L < |x| \leq L + \Delta L 
   	\end{cases}
\]
with $\eta = 0.4$, $L=18.0$ and $\Delta L = 1.0$ (all in atomic units). For clarity, we refer to the three finite-difference calculations, namely, with domain size 46 a.u., with domain size 38 a.u., and with domain size 38 a.u. along with NIP absorbing boundary condition as FD-46.0, FD-38.0, and FD-38.0-ABS, respectively. 
We use the simulation details, namely, time-step, duration of propagation, choice of propagator, and tolerance for Krylov subspace, are same as those used for the finite-element case. Fig.
~\ref{fig:Al2Abs} compares absorption spectrum obtained from finite-elements against finite-difference. We have used a Gaussian window of the form $g(t) = e^{-\alpha t^2}$, with $\alpha = 0.005$ a.u., in the Fourier transform of the dipole moment to artificially broaden the peaks. As is evident from the figure, we get good agreement with the finite-difference based results for the domain size of 46 a.u.\,. The
finite-difference calculation with domain size 38 a.u., with and without the absorbing boundary condition, provides qualitatively different
results with two peaks around 5~eV. We attribute this discrepancy to possible 
reflection effects from the boundary, as a domain size of 38 a.u. may not be sufficient to avoid finite-domain size effects. Furthermore, comparing FD-38 and FD-38-ABS curves, we observe that the use of NIP based absorbing boundary condition, on its own, hardly improves the answer. This suggests that, for the system under consideration, one cannot rely, solely, on absorbing boundary conditions to avoid reflection effects, and hence, must use a larger domain. Table ~\ref{tab:FEFDAl2}
compares the first two excitation peaks, the degrees of freedom and the total computational time for the finite-element and the finite-difference (46 a.u. domain size) based calculations. As is evident from the table, both finite-element and finite-difference based results agree to within $10$ meV in the excitation energies. Furthermore, in terms of computational efficiency, we observe a $\sim65$-fold speedup for finite-elements over finite-difference. We remark that this superior efficiency
for the finite-elements is largely attributed to fewer 
degrees of freedom that one can afford in finite-elements due to adaptive resolution of the mesh, as opposed to a uniform mesh in finite-difference. We underline this by noting that while finite-difference requires over $12$ million degrees of freedom, the finite-elements require only 31,411 degrees of freedom to attain similar accuracies.
\begin{figure}[h!]
    \begin{center}
        \includegraphics[scale=0.7]{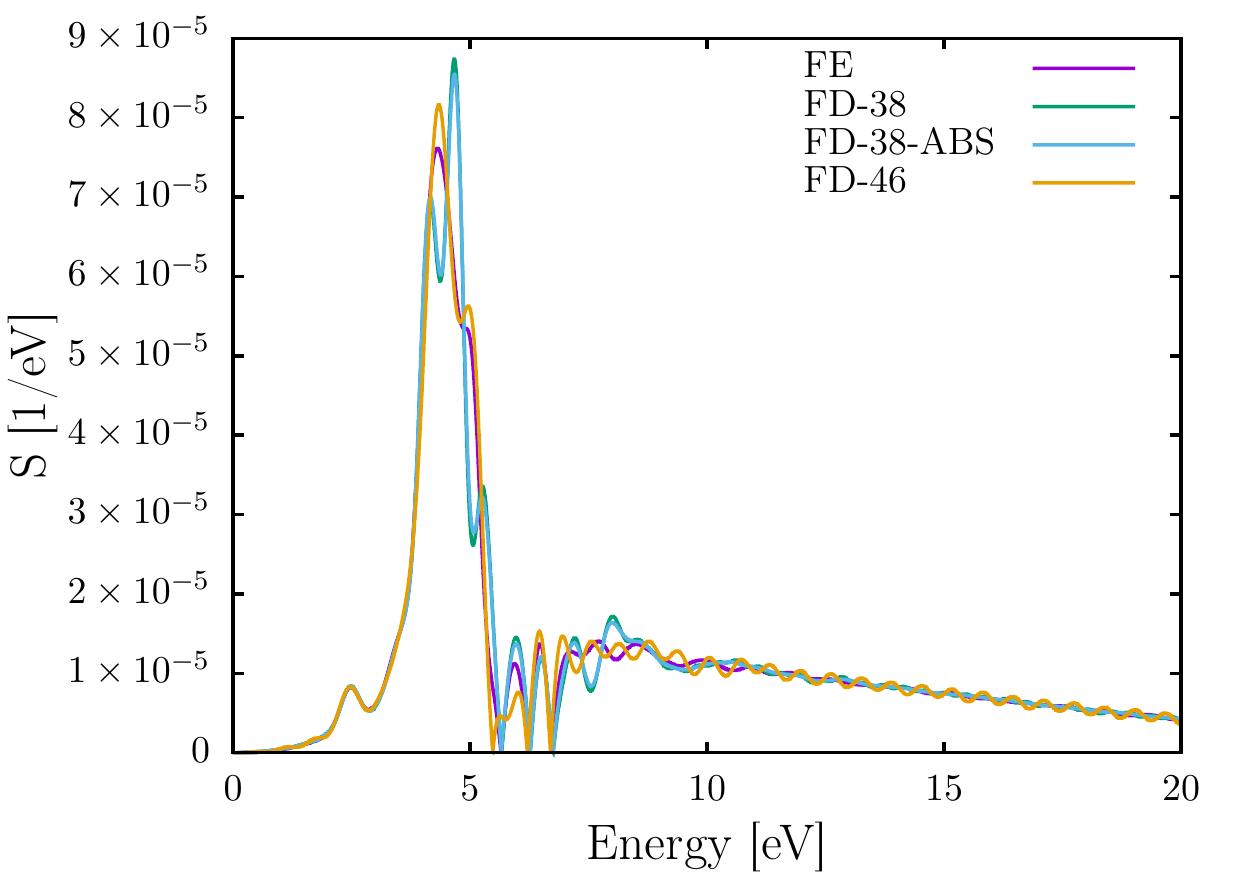}
            \caption{\small Absorption spectra for $\text{Al}_2$}
            \label{fig:Al2Abs}
    \end{center}
\end{figure}

\begin{table}[h!]
    \caption{\small Comparison of finite-element (FE) and finite-difference (FD) for $\text{Al}_2$: First and second excitation energies ($E_1$, $E_2$, respectively, in eV), degrees of freedom (DoF), and total computation CPU time (in CPU hours).} 
\begin{tabular}{M{0.1\columnwidth}M{0.15\columnwidth}M{0.15\columnwidth}M{0.25\columnwidth}M{.25\columnwidth}}
\hline 
\hline
    Method & $E_1$ & $E_2$ & DoF & CPU Hrs\\ \hline
    FE & $2.477$ & $4.325$ & $31,411$ & $2.11$ \\
    FD & $2.486$ & $4.332$ & $12,326,391$ & $138.8$\\ \hline
\end{tabular}
\label{tab:FEFDAl2}
\end{table}

\subsubsection{Pseudopotential calculations: $\text{Al}_{13}$}\label{sec:Al13}
We now consider a 13 atom aluminum cluster with an icosahedral symmetry. We use the same characteristic finite-element mesh as that of $\text{Al}_2$ but with a cubical domain of length 70 a.u., to avoid reflection effects. We excite the system from its ground-state using the parameters listed in Table~\ref{tab:Params}. We, once again, provide a comparative study
against finite-difference based calculation by using a uniform cubical mesh of size 56 a.u. and grid-spacing 0.2 a.u.\,. Fig.~\ref{fig:Al13Abs} compares absorption spectrum obtained from finite-elements against finite-difference. We have used the same Gaussian window as in the case of $\text{Al}_2$. As is evident from the figure, the peaks for both finite-element and finite-difference are in good agreement. Table ~\ref{tab:FEFDAl13} compares the first two excitation peaks, degrees of freedom, and the total computational time for the finite-element and the finite-difference based calculations. Both the methods agree to within 10 meV in the first two excitation energies. In terms of computational efficiency, the finite-elements attain an $\sim8$-fold savings in the computational
time against finite-difference, once again, attributed to the fewer degrees of freedom in finite-elements owing to adaptive resolution of the mesh. In particular, the finite-elements afford $\sim 30$-fold fewer degrees of freedom as compared to finite-difference. 
\begin{figure}[h!]
    \begin{center}
        \includegraphics[scale=0.7]{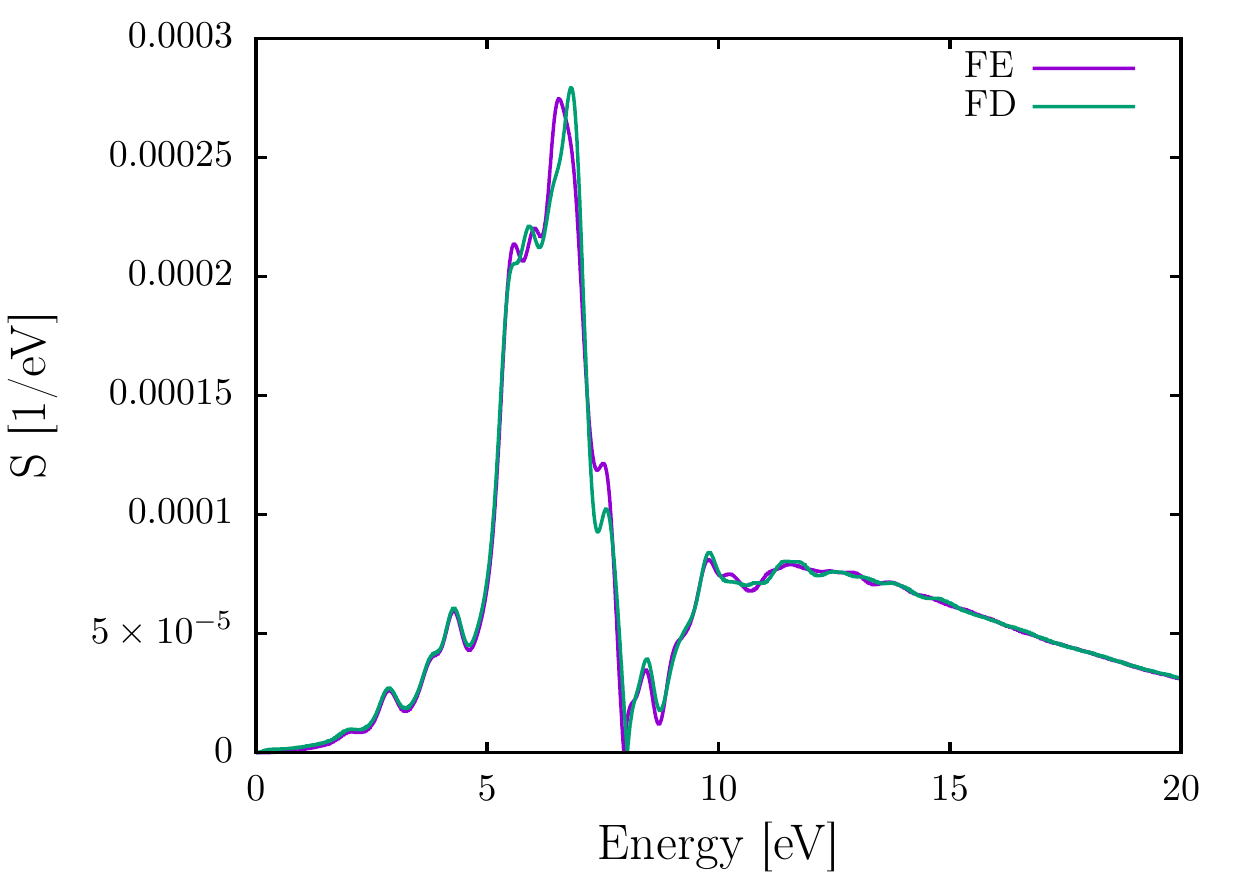}
            \caption{\small Absorption spectra for $\text{Al}_{13}$.}
            \label{fig:Al13Abs}
    \end{center}
\end{figure}

\begin{table}[h!]
    \caption{\small Comparison of finite-element (FE) and finite-difference (FD) for $\text{Al}_{13}$: First and second excitation energies ($E_1$, $E_2$, respectively, in eV), degrees of freedom (DoF), and total computation CPU time (in CPU hours).} 
\begin{tabular}{M{0.1\columnwidth}M{0.15\columnwidth}M{0.15\columnwidth}M{0.25\columnwidth}M{.25\columnwidth}}
\hline 
\hline
    Method & $E_1$ & $E_2$ & DoF & CPU Hrs\\ \hline
    FE & $2.876$ & $4.280$ & $698,782$ & $82.2$ \\
    FD & $2.880$ & $4.282$ & $22,188,041$ & $624.6$\\ \hline
\end{tabular}
\label{tab:FEFDAl13}
\end{table}

\subsubsection{Pseudopotential calculations: Buckminsterfullerene}\label{sec:Fullerene}
In this example, we consider the Buckminsterfullerene molecule comprising of 60 carbon atoms (240 electrons) packed into the shape of a buckyball. As with $\text{Al}_2$, we use an adaptive HEX64SPEC finite-elements discretization that follows the coarsening rate obtained from Eq.~\ref{eq:OptimalMeshSize} and is commensurate with an accuracy of $10$ meV
 in the ground-state energy per atom. We use a cubical domain of length 50 a.u. to eliminate any reflection effects from the boundaries.
We use the simulation parameters listed in Table~\ref{tab:Params} to excite the system from its ground-state. As with previous cases, we also assess the performance of higher-order finite-elements by comparing
against finite-difference based method, as implemented in the Octopus package. We assess the effects of domain size for the finite-difference mesh, by using two cubical domains of sizes 30 and 36 a.u., both with a grid-spacing of 0.15 a.u.\,. Furthermore, we study the effect of absorbing boundary conditions by performing an additional finite-difference calculation on the 30 a.u. mesh with a negative imaginary potential (NIP) near the boundaries. 
We use an NIP of the same form as used in $\text{Al}_2$  (see Sec. ~\ref{sec:Al2}), albeit with $L=14.0$ a.u.\,. We denote these three finite-difference calculations, namely, with domain size 36 a.u., with domain size 30 a.u., and with domain size 30 a.u. along with NIP absorbing boundary condition as FD-36, FD-30, and FD-30-ABS, respectively. 
Figure ~\ref{fig:FullereneAbs} shows the absorption spectrum obtained from finite-element and the three different finite-difference based calculations. We have used a same Gaussian window of the form $g(t)=e^{-\alpha t^2}$, with $\alpha=0.01$ a.u., to artificially broaden the peaks. As is evident from the figure, there is good agreement between the finite-element and FD-36 for all the excitation peaks. On the other hand, while FD-30 and FD-30-ABS have good agreement with finite-elements for the first two peaks, they differ for the rest, possibly because of reflection effects. Furthermore, comparing FD-30 and FD-30-ABS, we
remark that the use of NIP based absorbing boundary condition is not improving the absorption spectrum. This, once again, indicates that one cannot always dispense with the need for a larger domain by, solely, using absorbing boundary conditions. Table ~\ref{tab:FEFDFullerene} compares the first two excitation peaks, degrees of freedom, and the computational time for finite-elements against that of FD-36. Both finite-element and FD-36 based results
match within $30$ meV in the first two peaks. Furthermore, the excitation energies are also in good agreement with results presented in~\cite{Kawashita2009} (the first two excitation peaks, as we estimate from the absorption spectrum reported in~\cite{Kawashita2009}, are $\sim5.6$ eV and $\sim11.5$ eV, respectively.). In terms of computational efficiency, finite-elements attain a $\sim3$-fold speedup over FD-36. This higher efficiency of the finite-elements, is once again, attributed to a $\sim 9$-fold fewer degrees of freedom required by the finite-elements against that of finite-difference. 
\begin{figure}[h!]
    \begin{center}
        \includegraphics[scale=0.7]{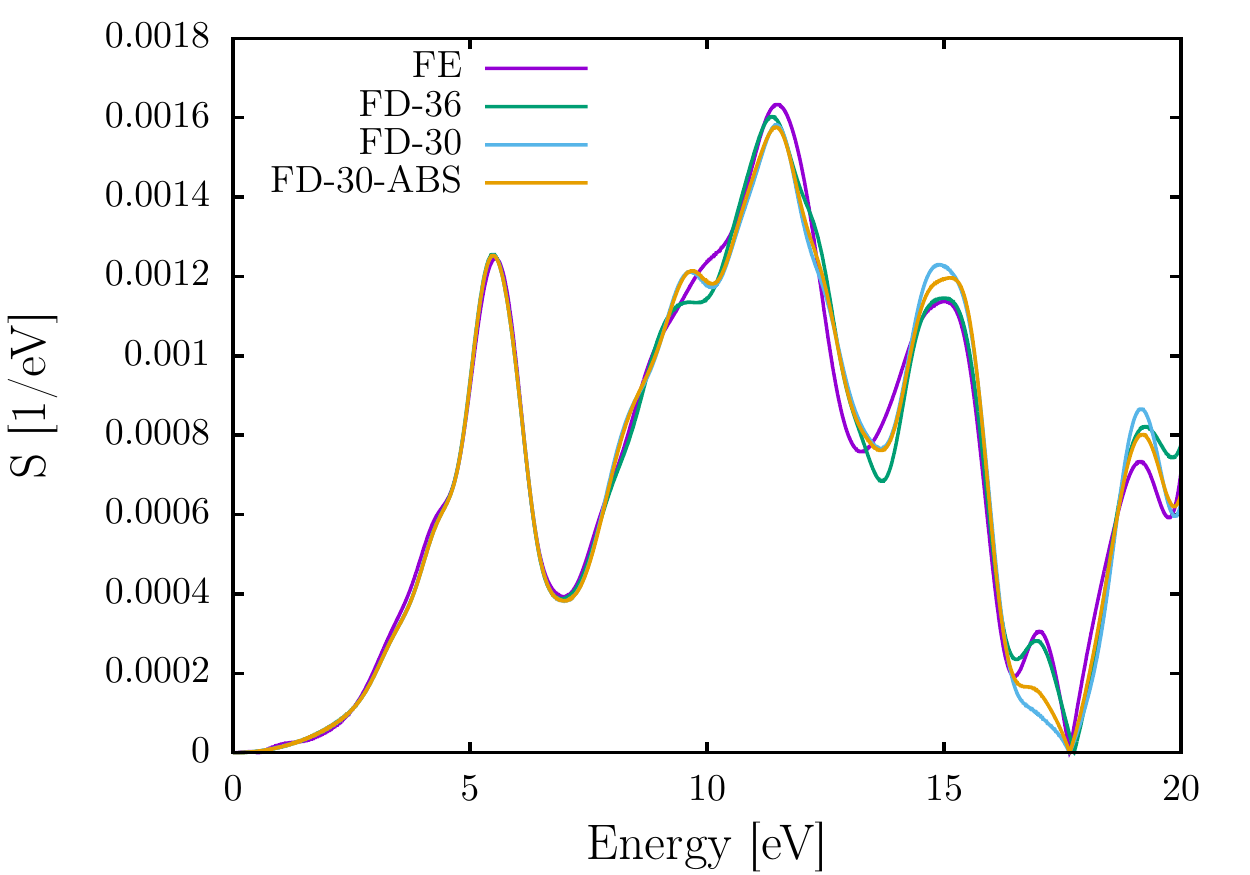}
            \caption{\small Absorption spectra of Buckminsterfullerene}
            \label{fig:FullereneAbs}
    \end{center}
\end{figure}

\begin{table}[h!]
    \caption{\small Comparison of finite-element (FE) and finite-difference (FD) for $\text{C}_{60}$: First and second excitation energies ($E_1$, $E_2$, respectively, in eV), degrees of freedom (DoF), and total computation CPU time (in CPU hours).} 
\begin{tabular}{M{0.1\columnwidth}M{0.15\columnwidth}M{0.15\columnwidth}M{0.25\columnwidth}M{.25\columnwidth}}
\hline 
\hline
    Method & $E_1$ & $E_2$ & DoF &CPU Hrs\\ \hline
    FE & $5.499$ & $11.412$ & $1,548,073$ & $5,200$ \\
    FD & $5.476$ & $11.439$ & $13,997,521$ &$15,361$\\ \hline
\end{tabular}
\label{tab:FEFDFullerene}
\end{table}

\subsubsection{Pseudopotential calculations: $\text{Mg}_2$}\label{sec:Mg2}
In this example, we study the higher harmonic generation in a magnesium dimer with bond-length of $4.74$ a.u.\,. Unlike the previous examples, we use a strong laser pulse to excite the system from its ground-state (see Table~\ref{tab:Params} for the simulation details). We use an adaptive HEX125SPEC mesh with the coarsening rate determined by Eq.~\ref{eq:OptimalMeshSize}. Furthermore, we use a cubical domain of length 100 a.u. to eliminate any reflection effects from the boundaries. We obtain the dipole power spectrum, $P(\omega)$, of the system by taking the imaginary part of the Fourier transform of the acceleration of the dipole moment, $\mu(t)$. To elaborate, $P(w)=\text{Im}\left(\int_0^T e^{-i\omega t}\frac{d^2}{dt^2}\mu(t)\,dt\right)$. Theoretically, for a system with spatial inversion symmetry, only odd multiples of the frequency of the exciting laser pulse must be emitted. We verify this in Figure ~\ref{fig:Mg2Power} wherein the peaks in the power spectrum coincide with odd
harmonics. Furthermore, we observe that the decay of the intensity of the peaks flattens beyond the $13$-$th$ harmonic, which corroborates well with the plateau phenomenon, typically observed in experiments ~\cite{Brabec2000}. We emphasize that despite the large domain size used in this calculation, we require only $\sim 60,000$ basis functions. This underlines the efficacy of higher-order finite-elements for even nonlinear regime in RT-TDDFT.   

\begin{figure}[h!]
    \begin{center}
        \includegraphics[scale=0.7]{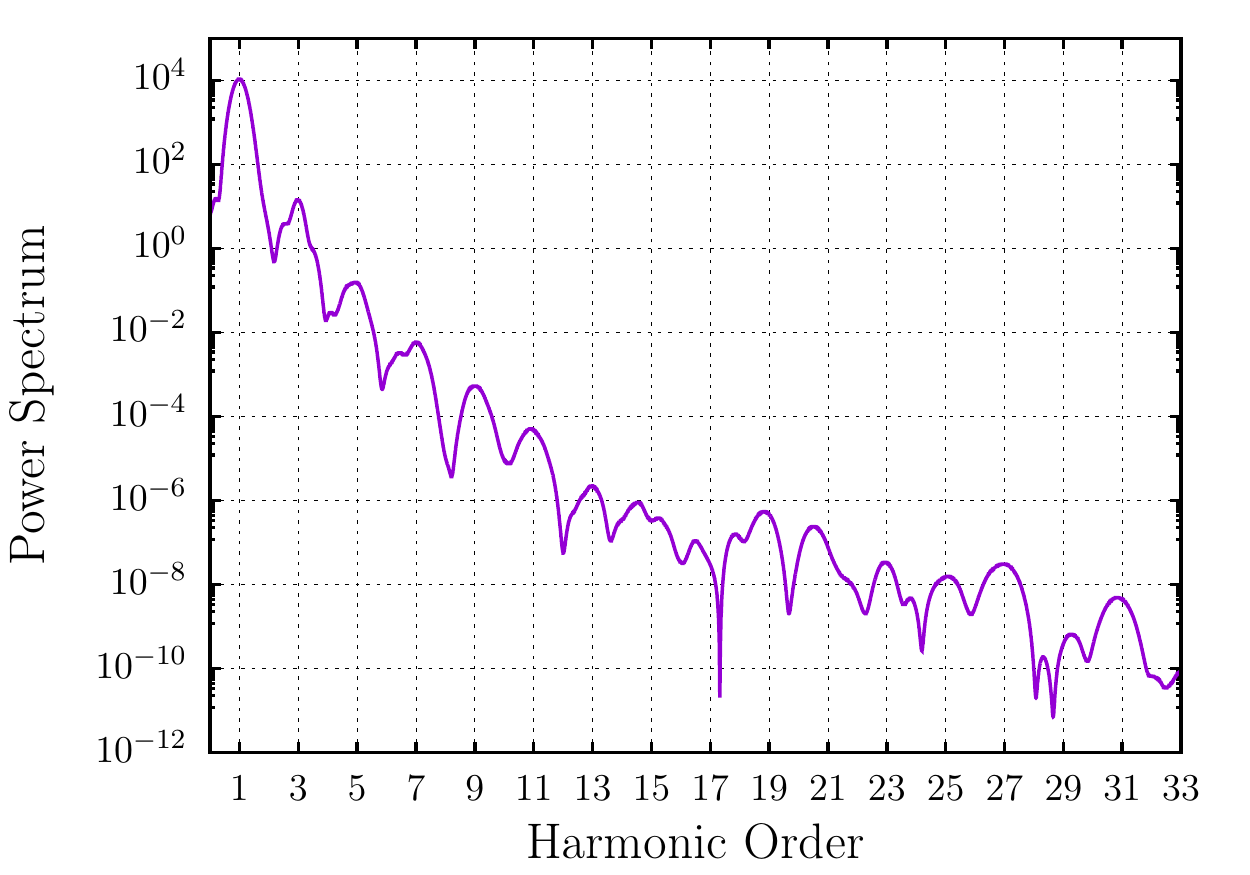}
            \caption{\small Dipole power spectrum of $\text{Mg}_2$}
            \label{fig:Mg2Power}
    \end{center}
\end{figure}

\subsubsection{All-electron calculations: Methane ($\text{CH}_4$)}\label{sec:CH4AE}
We now examine the competence of higher-order finite-elements for all-electron RT-TDDFT calculations by providing a comparative study with its pseudopotential counterpart. In this example, we consider a methane molecule with the same geometry as described in Sec. ~\ref{sec:CH4}. We use HEX64SPEC and HEX125SPEC elements for the pseudopotential and all-electron case, respectively. For both all-electron and pseudopotential cases, we use the same mesh adaption strategy
as used in all previous examples. For both the meshes, we use a large cubical domain of length 40 a.u., so as to eliminate reflection from the boundaries. Both the systems are excited from their respective ground-states using the simulation details listed in Table~\ref{tab:Params}. The absorption spectra for both the calculations are shown in Figure ~\ref{fig:MethaneAbs}. We used the same Gaussian window as in the case of Buckminsterfullerene (see Sec. ~\ref{sec:Fullerene}), to artificially broaden the peaks. As evident from the figure, we obtain remarkable agreement between the all-electron and pseudopotential results, i.e., the two curves are almost
identical. Table ~\ref{tab:AEPSPMethane} we list the first two excitation peaks, degrees of freedom, and total computational time for both the calculations. The first two excitation
peaks agree to within $10$ meV. We remark that the all-electron calculation requires $\sim100\text{x}$ more computational time as compared to the pseudopotential case. This large computational expense for the all-electron calculation stems primarily from the need of a highly refined mesh near the nuclei, so as to accurately capture the sharp variations in the electronic fields near the nuclei. This refinement has two major consequences: a) an increase in the degrees of freedom; and b) increase in $\norm{\psi\al^h}_{\Hone}$
, which in turn, warrants a smaller time-step (cf. Eq. ~\ref{eq:FullDiscreteDipoleMoment}) as well as a larger Krylov subspace to achieve the prescribed accuracy. In particular, for the case of methane, we required $\sim4\text{x}$ degrees of freedom and $\sim10\text{x}$ the size of the Krylov subspace as compared to that of the pseudopotential case. We emphasize that while finite-elements are expensive for the all-electron calculation, they provide the desired accuracy and offer systematic convergence (see Sec. ~\ref{sec:LiH}). Moreover, one can mitigate the need of a refined mesh for the all-electron calculation by using an enriched finite-element basis, wherein the standard (classical) finite-element basis are augmented with numerical atom-centered basis~\cite{Pask2011, Pask2012, Pask2017, Kanungo2017}. This idea has successfully attained $100-300\text{x}$ speedup over the standard (classical) finite-elements for ground-state DFT calculations~\cite{Kanungo2017}, and can be extended to RT-TDDFT to further the capabilities of
finite-elements.

\begin{figure}[h!]
    \begin{center}
        \includegraphics[scale=0.7]{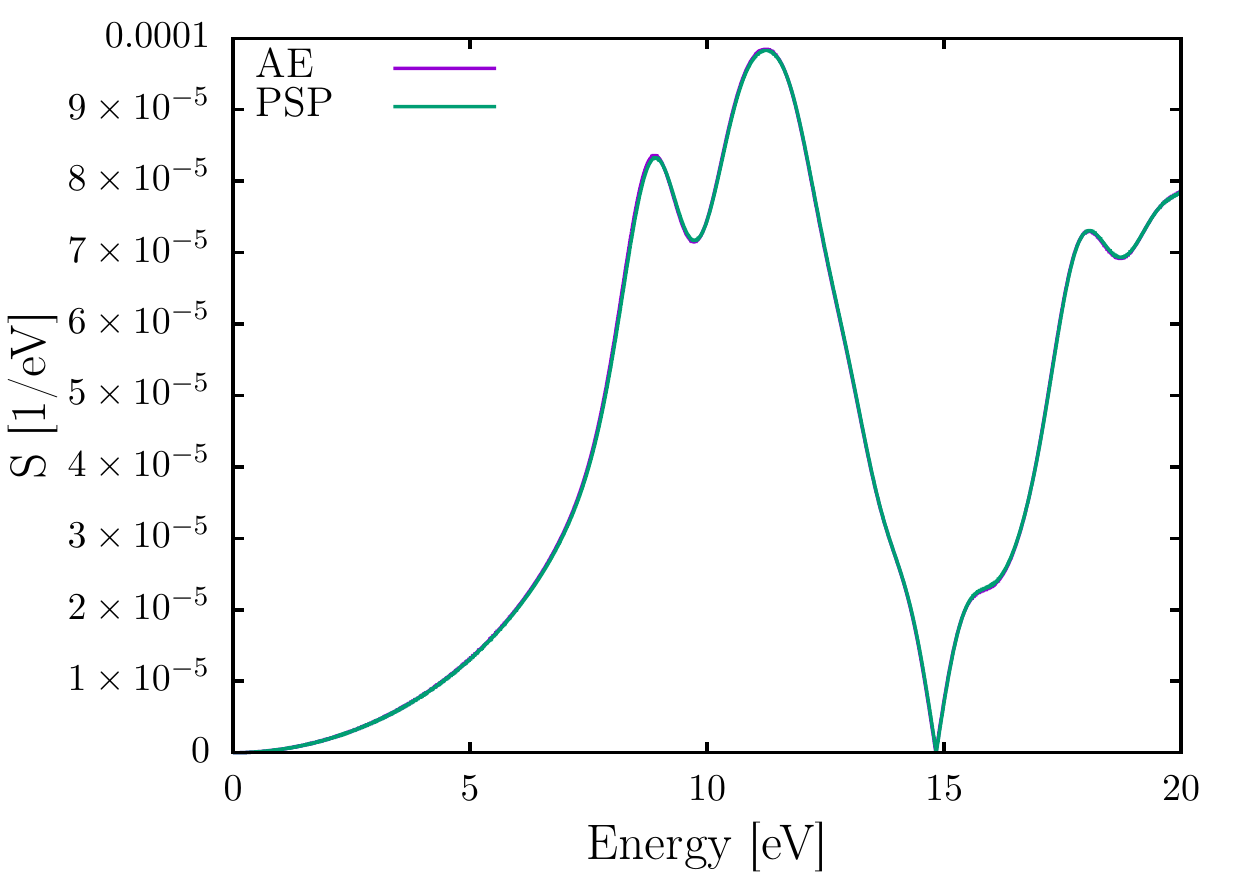}
            \caption{\small Absorption spectra of methane}
            \label{fig:MethaneAbs}
    \end{center}
\end{figure}

\begin{table}[h!]
    \caption{\small Comparison of all-electron (AE) and pseudopotential (PSP) calculations for methane: First and second excitation energies ($E_1$, $E_2$, respectively, in eV), degrees of freedom (DoF), and total computation CPU time (in CPU hours).} 
\begin{tabular}{M{0.1\columnwidth}M{0.15\columnwidth}M{0.15\columnwidth}M{0.25\columnwidth}M{.25\columnwidth}}
\hline 
\hline
    Method & $E_1$ & $E_2$ & DoF &CPU Hrs\\ \hline
    AE & $8.898$ & $11.238$ & $348,289$ & $13,653$ \\
    PSP & $8.907$ & $11.244$ & $80,185$ &$145$\\ \hline
\end{tabular}
\label{tab:AEPSPMethane}
\end{table}

\subsubsection{All-electron calculations: Benzene}\label{sec:BenzeneAE}
In this example, we perform similar comparative studies between all-electron and pseudopotential calculations for benzene molecule. As with the methane molecule, we use HEX64SPEC and HEX125SPEC finite-elements for the pseudopotential and all-electron calculation, respectively. Furthermore, we use the same characteristic mesh features (i.e., refinement near the nuclei, coarsening rate, simulation domain), in both the meshes, as their counterparts in the methane calculation (see Sec.~\ref{sec:CH4AE}). The simulation details, for both the cases, are listed in Table~\ref{tab:Params}. Figure ~\ref{fig:BenzeneAbs} compares the absorption spectra from the all-electron and pseudopotential calculations. Both the spectra compares well with thre results presented in~\cite{Lehtovaara2011}, in terms of first two excitation peaks (the first two excitation peaks, as we estimate from the absorption spectrum reported in~\cite{Lehtovaara2011}, are $\sim6.6$ eV and $\sim10$ eV, respectively). We remark that while there is qualitative agreement between the pseudpotential and all-electron calculations, quantitatively the predictions from all-electron and pseudopotential calculations differ. In particular, the first two excitation
peaks (see Table ~\ref{tab:AEPSPBenzene}) differ up to $\sim 0.2$ eV. This suggests that one ought to carefully test for the transferability of the pseudopotential approximation used, to provide reliable quantitative predictions from RT-TDDFT calculations. We take note that a more careful comparison of pseudopotential and all-electron calculations warrants a scan through a range of pseudopotential approximation. Nevertheless, the objective of this exercise is to highlight the fact that
finite-elements, by treating both pseudopotential and all-electron calculations on an equal footing, allows for a robust tool for such transferability studies.
\begin{figure}[h!]
    \begin{center}
        \includegraphics[scale=0.7]{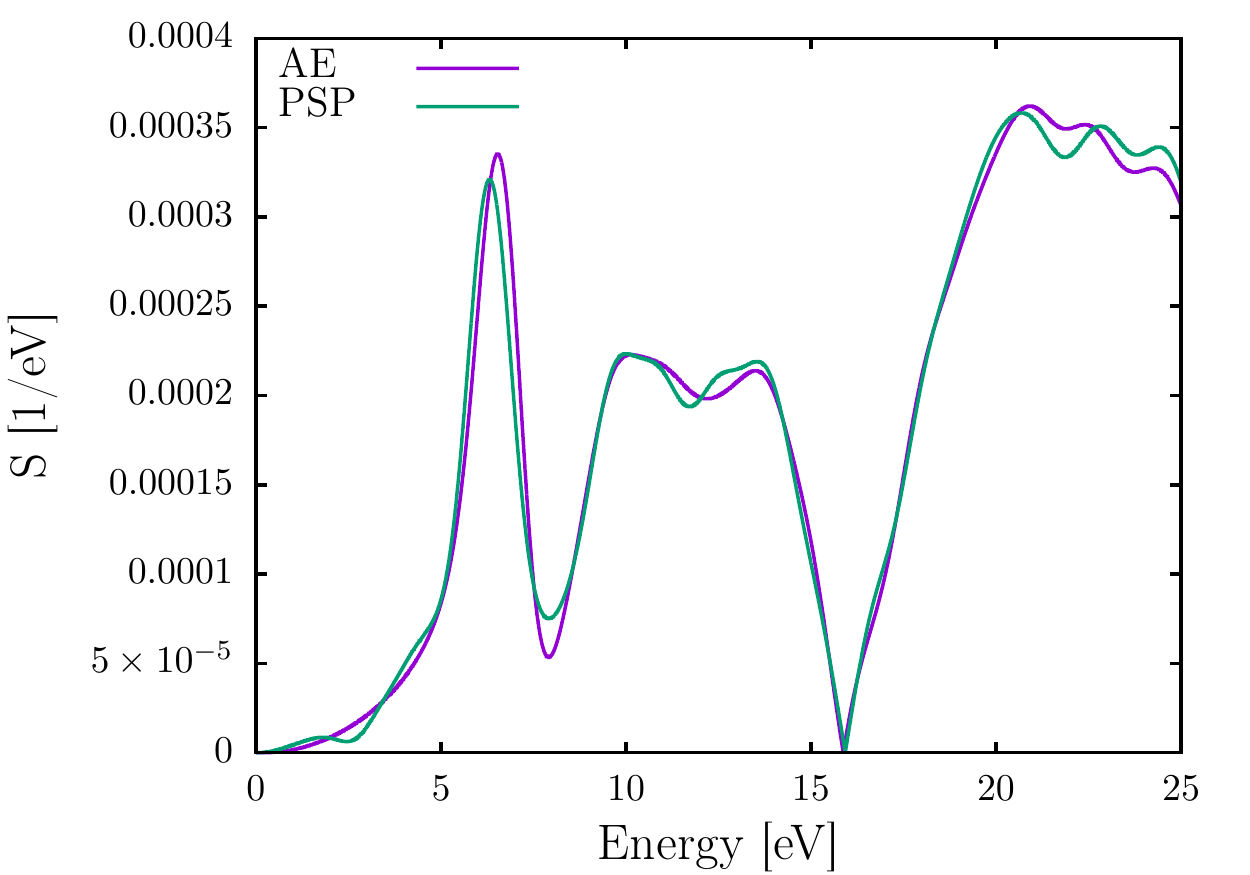}
            \caption{\small Absorption spectra of benzene}
            \label{fig:BenzeneAbs}
    \end{center}
\end{figure}

\begin{table}[h!]
    \caption{\small Comparison of all-electron (AE) and pseudopotential (PSP) calculations for benzene: First and second excitation energies ($E_1$, $E_2$, respectively, in eV), degrees of freedom (DoF), and total computation CPU time (in CPU hours).} 
\begin{tabular}{M{0.1\columnwidth}M{0.15\columnwidth}M{0.15\columnwidth}M{0.25\columnwidth}M{.25\columnwidth}}
\hline 
\hline
    Method & $E_1$ & $E_2$ & DoF &CPU Hrs\\ \hline
    AE & $6.521$ & $10.131$ & $989,649$ & $153,600$ \\
    PSP & $6.316$ & $10.007$ & $257,473$ & $1,574$\\ \hline 
\end{tabular}
\label{tab:AEPSPBenzene}
\end{table}

\subsection{Scalability} \label{sec:Scalability}
Lastly, we demonstrate the parallel scalability (strong scaling) of the proposed finite-element basis in Figure ~\ref{fig:Scalability}. We choose the Buckminsterfullerene molecule containing $\sim3.5$ million degrees of freedom (number of basis functions) as our fixed benchmark system and report the relative speedup with respect to the wall time on 24 processors. The use of any number of processors below 24 was unfeasible owing to the memory requirement posed by the system. As is evident from the figure, the
scaling is in good agreement with the ideal linear scaling behavior up to 384 processors, at which we observe a parallel efficiency of 86.2$\%$. However, we observe a deviation from linear scaling behavior at 768 processors with a parallel efficiency of 74.2$\%$. This is attributed to the fact that, at 768 processors, the number of degrees of freedom possessed by each processor falls below 5000, which is low to achieve good parallel scalability.  
\begin{figure}[h!] 
  \centering
  \includegraphics[scale=0.7]{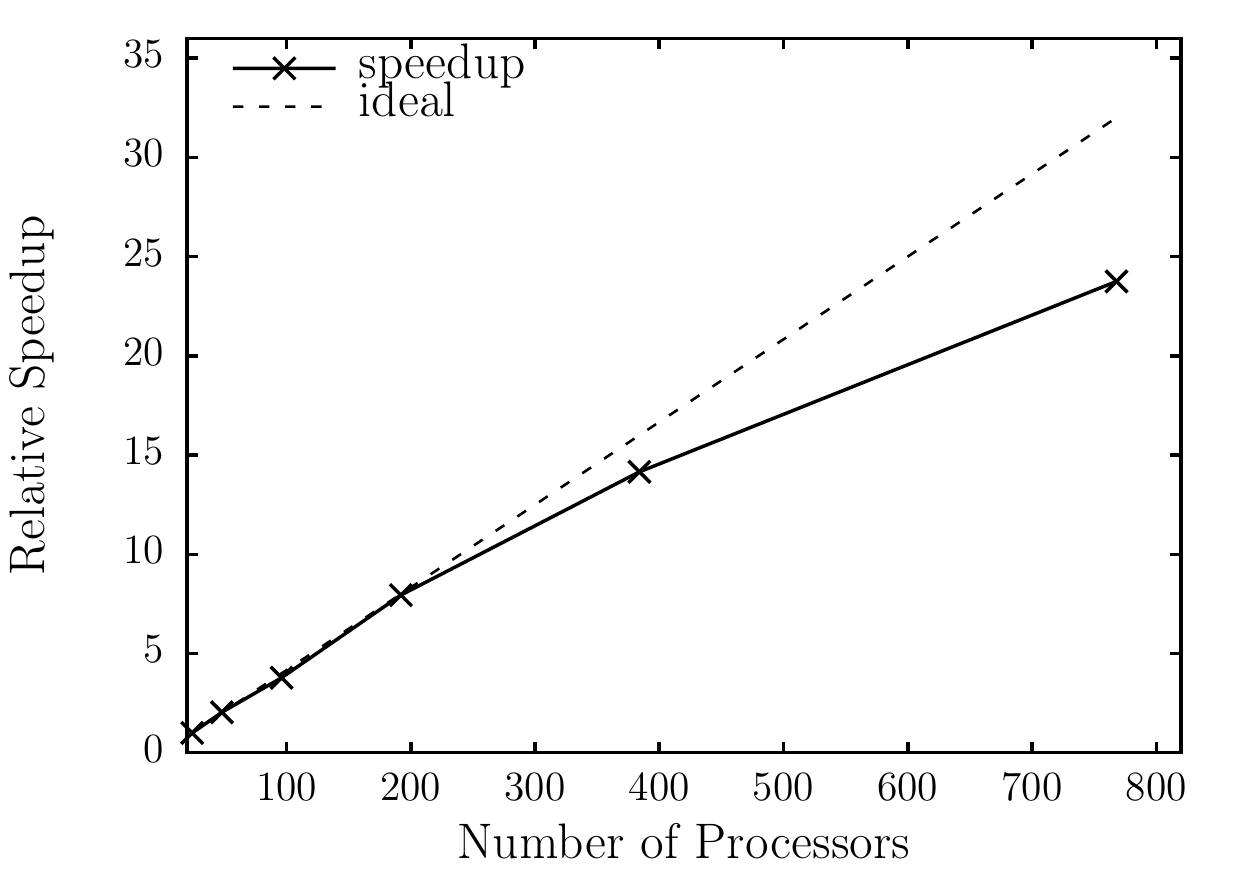}
  \caption{\small Parallel scalability of the higher-order finite-element implementation.}
  \label{fig:Scalability}
\end{figure}

\section{Summary} \label{sec:Summary}
In summary, we have investigated the accuracy, computational efficiency and scalability of higher-order finite-elements for the RT-TDDFT problem, for both pseudopotential and all-electron calculations. We presented an efficient \emph{a priori} spatio-temporal scheme guided by the discretization errors in the time-dependent Kohn-Sham orbitals, in the context of second-order Magnus propagator. In particular, we used the knowledge of the ground-state electronic fields to determine an efficient adaptively resolved finite-element mesh. This adaptive resolution is crucial in affording the use of large simulation domains without significant increase in the number of basis functions, and hence, allows us to circumvent the use of any artificial absorbing boundary conditions. A key aspect of the finite-element discretization in this work is the use of higher-order spectral finite-elements, which while providing a better conditioned basis also renders the overlap matrix diagonal when combined with special quadrature rules for numerical integration. This, in turn, enabled an efficient construction of the Magnus propagator (or any exponential time-integrator) for finite-element discretization. Furthermore, we employed an adaptive Lanczos subspace projection to evaluate the action of the Magnus propagator, defined as exponential of a matrix, on the Kohn-Sham orbitals.

We demonstrated the accuracy of the proposed approach through numerical convergence studies on both pseudopotential and all-electron systems, where we obtained close to optimal rates of convergence with respect to both spatial and temporal discretization, as determined by our error estimates. The computational efficiency afforded by using higher-order finite-element discretization was established, where a staggering $10-100$ fold speedup was obtained on benchmark systems by using a fourth-order finite-element in comparison to linear and quadratic finite-elements. Furthermore, we assessed the accuracy and efficiency afforded by our approach against the finite-difference based method of Octopus software package, for pseudopotential calculations. Across all the benchmark systems considered, we obtained good agreement in the absorption spectrum with calculations using the Octopus package. In terms of computational efficiency, we obtained $3-60$ fold speedup over finite-difference, which is largely attributed to the adaptive spatial resolution afforded by our approach. We also demonstrated the efficacy of finite-elements, especially its efficient handling of large domains, for nonlinear response by studying the higher harmonic generation under a strong electric field.  We also demonstrated the competence of higher-order finite-elements for the all-electron RT-TDDFT calculations. This underscores the versatility of finite-elements in handling both pseudopotential and all-electron calculations on an equal footing. Lastly, in terms of parallel scalability, we obtained good parallel efficiency up to 768 processors for a benchmark system comprising of the Buckminsterfullerene molecule containing $\sim 3.5$ million basis functions.

Thus, the proposed approach offers a computationally efficient, systematically improvable, and scalable basis for RT-TDDFT calculations, applicable to both pseudopotential and all-electron cases. We remark that, for the all-electron case, the need for a highly refined mesh near the nuclei increases the computational cost, as observed from the reported studies. For systems with heavier atoms, the mesh requirements become even more exacting. However, this can be alleviated by augmenting the finite-element basis with numerical atom-centered basis. This idea has been successfully used for ground-state DFT~\cite{Pask2011, Pask2012, Pask2017, Kanungo2017}, and its extension to RT-TDDFT is currently being investigated. Further, assessing the transferability of pseudopotentials for electron dynamics, enabled by the unified treatment of all-electron and pseudoptential calculations, is another interesting direction for future investigation.  


\acknowledgements
We are grateful for the support of Army Research Office through Grant number W911NF-15-1-0158, under the auspices of which the mathematical formulation, error analysis and numerical implementation was developed, and the pseudopotential studies were conducted. We also gratefully acknowledge the support from the Department of Energy, Office of Basic Energy Sciences, under grant number DE-SC0017380, for supporting the all-electron studies. This work used the Extreme Science and Engineering Discovery Environment (XSEDE), which is supported by National Science Foundation Grant number ACI-1053575. This research used resources of the National Energy Research Scientific Computing Center, a DOE Office of Science User Facility supported by the Office of Science of the U.S. Department of Energy under Contract No. DE-AC02-05CH11231. We also acknowledge the Advanced Research Computing at University of Michigan for providing additional computing resources through the Flux computing platform, part of which were performed using the computing cluster constructed from the DURIP Grant number W911NF-18-1-0242.      

\onecolumngrid
\section*{Appendix: Derivation of spatial and time discretization error estimates} \label{sec:appendix}
In this section we provide the detailed derivation of the spatial and time discretization error estimates presented in the main text (i.e., Eqs. ~\ref{eq:SemiDiscreteErrorH1GS} and ~\ref{eq:PsiTimeDiscreteError}).  

\subsection{Notations, assumptions and preliminaries} \label{sec:appendA}
For a  bounded closed domain $\Omega$ and bounded time interval $[0,T]$, we denote $\Omega_T=\Omega\times[0,T]$. For any two complex-valued functions $f(\br,t), g(\br,t):\Omega_T\rightarrow \mathbb{C}$, the inner product $(f,g)(t) = \int_{\Omega}{f(\br,t)g^\dag(\br,t)\;d\br}$, where $g^\dag(\br,t)$ denotes the complex conjugate
of $g(\br,t)$. Correspondingly, the norm $||f||_{L^2(\Omega)}(t) = \sqrt{(f,f)(t)}$. Thus, we extend the definition of the standard $L^2(\Omega)$ and $H^1(\Omega)$ spaces to define
\begin{subequations} \label{eq:Spaces}
    \begin{equation}
        L^2(\Omega_T)=\left\{f(\br,t) \, \middle| \, ||f||_{L^2(\Omega)}(t) \leq \infty, \quad \forall t \in [0,T] \right\}\,,
    \end{equation}
    \begin{equation}
        H^1(\Omega_T)=\left\{f(\br,t) \, \middle| \, f,\frac{\partial{f}}{\partial t},Df \in L^2(\Omega_T) \right\}\,,
    \end{equation}
    \begin{equation}
        H_0^1(\Omega_T)=\left\{f(\br,t) \, \middle| \, f\in H^1(\Omega_T), f(\br,t)\vert_{\partial \Omega}=0, \quad \forall t \in [0,T]\right\}\,,
    \end{equation}
\end{subequations}
where $Df$ denotes the spatial partial derivatives of $f$, and $\partial \Omega$ denotes the boundary of $\Omega$.
Additionally, we define two more spaces relevant to the Poisson problem (Eq. ~\ref{eq:Poisson}),
\begin{subequations} \label{eq:SpacesPoisson}
    \begin{equation}
        H^1_Z(\Omega_T)=\left\{f(\br,t) \, \middle| \, f \in H^1(\Omega_T), f(\br,t)\vert_{\partial \Omega}=\sum_{I=1}^{N_a}\frac{Z_I}{|\br-\bR_I|} \quad \forall t \in [0,T] \right\}
    \end{equation}
    \begin{equation}
        H^1_{-Z}(\Omega)=\left\{f(\br) \, \middle| \, f \in H^1(\Omega), f(\br)\vert_{\partial \Omega}=\sum_{I=1}^{N_a}\frac{-Z_I}{|\br-\bR_I|} \right\}.
    \end{equation}
\end{subequations}\\
For conciseness of notation, in all our subsequent discussion, we drop the argument $t$ from the inner product as well as all the $L^p$ and $H^1$ norms. Thus, any occurrence of $(.\,,\,.)$, $||.||_{L^p(\Omega)}$, and $||.||_{H^1(\Omega)}$ are to treated as time-dependent, unless otherwise specified.\\

We list certain weak assumptions that we invoke throughout our error-estimates. 
\begin{enumerate}[label=$\mathcal{A}$\arabic*]
    \item The time-dependent Kohn-Sham orbitals and their spatial derivatives are bounded and have a compact support on $\Omega$, which, in turn, is a large but a bounded subset of $\mathbb{R}^3$. To elaborate, $\psi\al\in H^1_0(\Omega_T)\cap L^{\infty}(\Omega_T)$. \label{assm:psiBound}
    \item The nuclear potential (in the all-electron case), due to the use of regularized nuclear charge distribution $b(\br;\bR)$ (defined in Eq. ~\ref{eq:Nuclear_poisson}), is bounded, i.e., $V_N^{ae} \in L^{\infty}(\mathbb{R}^3)$. \label{assm:VNAllElectron} 
    \item The local part of the pseudopotential is bounded, i.e., $V_{psp}^{loc} \in L^{\infty}(\mathbb{R}^3)$. \label{assm:VPSPLocal}
    \item The short-ranged potentials appearing in the nonlocal part of the pseudopotential are bounded, i.e., $\delta V^I_l \in L^{\infty}(\Omega)$. \label{assm:VPSPNonLocal}
    \item The exchange-correlation potential and its derivative with respect to density are both bounded, i.e., $V_{XC}[\rho], V_{XC}'[\rho] \in L^{\infty}(\mathbb{R}^3), \forall t \in [0,T]$. \label{assm:VXC}
    \item The external field is bounded, i.e., $V_{field} \in L^{\infty}(\mathbb{R}^3), \forall t \in[0,T]$. \label{assm:VField}
    \item The induced operator (or matrix) norm of the Kohn-Sham Hamiltonian and the Laplace operator are equivalent, i.e., $\exists$ time-independent bounded constants $C_1, C_2$ such that: \\
    $C_1 \norm{\nabla^2\phi}_{\Ltwo} \leq \norm{H_{KS}\phi}_{\Ltwo} \leq  C_2 \norm{\nabla^2\phi}_{\Ltwo} \,, \forall \phi \in H^1_0(\Omega), \,\forall t \in [0,T]$. \label{assm:NormEqui} 
    \item The first and second time-derivatives of the Kohn-Sham potential are bounded, i.e., $\norm{\frac{d}{dt}V_{KS}(t)}_{\Ltwo} \leq C_1$ and $\norm{\frac{d^2}{dt^2}V_{KS}(t)}_{\Ltwo} \leq C_2\,, \forall t \in [0,T]$, where $C_1, C_2$ are time-independent bounded constants. \label{assm:VKSTimeDerBound}
\end{enumerate}
We remark that while the validity of~\ref{assm:psiBound} and~\ref{assm:NormEqui} are apparent in the case of pseudopotential calculations, for the all-electron case, it is reasonable to assume the same owing to the use of regularized nuclear charge distribution $b(\br;\bR)$.
Using these assumptions, we derive certain formal bounds that will subsequently be used in deriving the error estimates. To this end, given two different densities $\rho_{\Psi_1}(\br,t)$ and $\rho_{\Psi_2}(\br,t)$ defined by the set of orbitals $\Psi_1=\{\psi_{1,1},\psi_{1,2},\ldots,\psi_{1,N_e}\}$ and $\Psi_2=\{\psi_{2,1},\psi_{2,2},\ldots,\psi_{2,N_e}\}$, respectively, we seek to bound
$\norm{V_{KS}[\rho_{\Psi_1}]\psi_{1,\alpha}-V_{KS}[\rho_{\Psi_2}]\psi_{2,\alpha}}_{\Ltwo}$ in terms of $(\psi_{1,\alpha}-\psi_{2,\alpha})$ and $(\rho_{\Psi_1}-\rho_{\Psi_2})$. We remark that all the subsequent results hold $\forall \alpha \in \{1,2,\ldots,N_e\}$, unless otherwise specified. Moreover, the constants $C$, its subscripted forms (i.e., $C_1$, $C_2$, etc.), and primed forms ($C'$), that appear subsequently, are positive and bounded.

To begin with, we note, through straightforward use of Cauchy-Schwarz and Sobolev inequalities, that
\begin{subequations} \label{eq:RhoInequality}
\begin{equation} \label{eq:RhoInequality1}
        \norm{\rho_{\Psi_1}-\rho_{\Psi_2}}_{\Lone} \leq C \sum_{\alpha=1}^{N_e}\norm{\psi_{1,\alpha}-\psi_{2,\alpha}}_{\Ltwo}  \,,
\end{equation}
\begin{equation} \label{eq:RhoInequality2}
        \norm{\rho_{\Psi_1}-\rho_{\Psi_2}}_{\Ltwo} \leq C \sum_{\alpha=1}^{N_e}\norm{\psi_{1,\alpha}-\psi_{2,\alpha}}_{\Hone}  \,.
\end{equation}
\end{subequations}
Furthermore, for the convolution integral of $\rho$ and $\frac{1}{|\br|}$, denoted by $|\br|^{-1}*\rho=\int_{\Omega}\rho(\bx)\frac{1}{|\br-\bx|}d\bx$, we have
\begin{subequations} \label{eq:HartreeConv}
\begin{equation} \label{eq:HartreeConv1}
    \norm{|\br|^{-1}*\rho}_{\Linf} \leq C\norm{|\br|^{-1}}_{\Ltwo} \norm{\rho}_{\Ltwo}\,,
\end{equation}
\begin{equation} \label{eq:HartreeConv2}
    \norm{|\br|^{-1}*\rho}_{\Ltwo} \leq C\norm{|\br|^{-1}}_{\Ltwo} \norm{\rho}_{\Lone}\,,
\end{equation}
\end{subequations}
where we have used the Young's inequality along with the fact that $|\br|^{-1} \in L^2(\Omega)$. 

We now bound $\norm{V_{KS}[\rho_{\Psi_1}]\psi_{1,\alpha}-V_{KS}[\rho_{\Psi_2}]\psi_{2,\alpha}}_{\Ltwo}$, by decomposing $V_{KS}$ into its Hartree ($V_H$), nuclear ($V_N$), exchange-correlation ($V_{XC}$) and field ($V_{field}$) components, and bounding each of the components. For the Hartree potential, we have, for $\forall v \in H_0^1(\Omega_T)$, 
\begin{equation} \label{eq:HartreeDotProd}
    \left(V_H[\rho_{\Psi_1}]\psi_{1,\alpha}-V_H[\rho_{\Psi_2}]\psi_{2,\alpha},v\right) = \left(V_H[\rho_{\Psi_1}](\psi_{1,\alpha}-\psi_{2,\alpha}),v\right) + \left(V_H[\rho_{\Psi_1}-\rho_{\Psi_2}]\psi_{2,\alpha},v\right) \,.
\end{equation}
Thus, using result of Eq. ~\ref{eq:HartreeConv2} along with the fact that $\psi_{2,\alpha}\in L^{\infty}(\Omega_T)$ (from ~\ref{assm:psiBound}) and  $V_H[\rho_{\Psi_1}] \in L^{\infty}(\Omega_T)$ (from Eq.~\ref{eq:HartreeConv1}), it follows that
\begin{equation} \label{eq:HartreeInequality}
\begin{split}
    \modulus{\left(V_H[\rho_{\Psi_1}]\psi_{1,\alpha}-V_H[\rho_{\Psi_2}]\psi_{2,\alpha},v\right)} 
    & \leq C\left(\norm{\psi_{1,\alpha}-\psi_{2,\alpha}}_{\Ltwo}\norm{v}_{\Ltwo} + \norm{\rho_{\Psi_1}-\rho_{\Psi_2}}_{\Lone}\norm{v}_{\Ltwo}\right)\,.
\end{split}
\end{equation}

Next, for the exchange-correlation potential, we use the mean value theorem to note that 
\begin{equation} \label{eq:VXCMeanValue}
    V_{XC}[\rho_{\Psi_1}]\psi_{1,\alpha}-V_{XC}[\rho_{\Psi_2}]\psi_{2,\alpha} = (V_{XC}[\rho_{X}]+2 \chi\al^2 V'_{XC}[\rho_{X}])(\psi_{1,\alpha}-\psi_{2,\alpha})\,,
\end{equation}
where $\rho_X$ is defined by the orbitals $\chi\al=\lambda\al\psi_{1,\alpha}+(1-\lambda\al)\psi_{2,\alpha},\text{~for~some~} \lambda\al \in [0,1]$. Using the above relation, we have, $\forall v \in H_0^1(\Omega_T)$,  
\begin{equation} \label{eq:VXCInequality}
    \begin{split}
        \modulus{\left(V_{XC}[\rho_{\Psi_1}]\psi_{1,\alpha}-V_{XC}[\rho_{\Psi_2}]\psi_{2,\alpha},v\right)} &= \modulus{\left((V_{XC}[\rho_{X}]+2 \chi\al^2 V'_{XC}[\rho_{X}])(\psi_{1,\alpha}-\psi_{2,\alpha}),v\right)} \\ 
        & \leq \norm{V_{XC}[\rho_{X}]+2 \chi\al^2 V'_{XC}[\rho_{X}]}_{\Linf}\norm{\psi_{1,\alpha}-\psi_{2,\alpha}}_{L^2(\Omega)} \norm{v}_{\Ltwo}\\
        & \leq C\norm{\psi_{1,\alpha}-\psi_{2,\alpha}}_{L^2(\Omega)}\norm{v}_{\Ltwo}\,,
    \end{split}
\end{equation}
where we have used the boundedness assumption on $V_{XC}$ and $V'_{XC}$ (assumption ~\ref{assm:VXC}).

Similarly, using the boundedness assumptions on $V_N^{ae}$ (\ref{assm:VNAllElectron}), $V_N^{psp}$( \ref{assm:VPSPLocal}, ~\ref{assm:VPSPNonLocal}), and $V_{field}$ (\ref{assm:VField}) it is easy to observe, $\forall v \in H_0^1(\Omega_T)$,  
\begin{equation} \label{eq:VNAllElectronInequality}
    \begin{split}
        \modulus{\left(V_N^{ae}\psi_{1,\alpha}-V_N^{ae}\psi_{2,\alpha},v\right)} & \leq C \norm{\psi_{1,\alpha}-\psi_{2,\alpha}}_{\Ltwo}\norm{v}_{\Ltwo}\,.
    \end{split}
\end{equation}
\begin{equation} \label{eq:VPSPInequality}
    \begin{split}
    \modulus{\left(V_N^{psp}\psi_{1,\alpha}-V_N^{psp}\psi_{2,\alpha},v\right)} & \leq C\norm{\psi_{1,\alpha}-\psi_{2,\alpha}}_{\Ltwo}\norm{v}_{\Ltwo} \,.   
    \end{split}
\end{equation}
\begin{equation} \label{eq:VFieldInequality}
\begin{split}
\modulus{\left(V_{field}\psi_{1,\alpha}-V_{field}\psi_{2,\alpha},v\right)} \leq C \norm{\psi_{1,\alpha}-\psi_{2,\alpha}}_{\Ltwo}\norm{v}_{\Ltwo} \,.
\end{split}
\end{equation}

We now define the weak solution of the TDKS equation (Eq.~\ref{eq:TDKS}) as follows: 
given an initial state $\psi\al(\br,0) \in H_0^1(\Omega)$, we seek $\psi\al(\br,t) \in H_0^1(\Omega_T)$ such that
\begin{equation} \label{eq:WeakSolution}
    i\left(\frac{\partial \psi\al}{\partial t}, v\right)=\frac{1}{2}\left(\nabla \psi\al,\nabla v \right) + \left(V_{KS}[\rho]\psi\al,v \right), \quad \forall v \in H_0^1(\Omega_T)\,, \, \text{and} ~ \forall t \in [0,T]\,.
\end{equation}
Similarly, the weak solutions to the Poisson problems defined in Eq. ~\ref{eq:Poisson} are defined to be $V_H(\br,t)\in H^1_Z(\Omega_T)$, and $V_N^{ae}(\br,\bR) \in H^1_{-Z}(\Omega)$, satisfying, 
\begin{subequations} \label{eq:WeakSolutionPoisson}
    \begin{equation} \label{eq:WeakSolutionHartree}
        (\nabla V_H,\nabla v)=4\pi(\rho,v), \quad \forall v \in H_0^1(\Omega_T), \, \text{and} ~ \forall t \in [0,T]
    \end{equation}
    \begin{equation} \label{eq:WeakSolutionNuclear}
        \left(\nabla V_N^{ae},\nabla v\right)=4\pi\left(b,v\right), \quad \forall v \in H_0^1(\Omega)\,.
    \end{equation}
\end{subequations}

\subsection{Derivation of spatial discretization error estimate} \label{sec:appendB}
We denote $X^{h,p} \in H^1(\Omega)$ to be the finite-dimensional space of dimension $n^h$, spanned by finite-element basis functions of order $p$. Further, we denote $X_0^{h,p}=X^{h,p} \cap H_0^1(\Omega)$. We now define the semi-discrete solution, $\psi\al^h(\br,t)$, to Eq. ~\ref{eq:WeakSolution} as follows: given an initial state $\psi\al^h(\br,0) \in X_0^{h,p}$, we seek $\psi\al^h(\br,t) \in X_0^{h,p} \times [0,T]$ such that
\begin{equation} \label{eq:SemiDiscreteWeakSolution}
    i\left(\frac{\partial \psi\al^{h}}{\partial t}, v^h\right)=\frac{1}{2}\left(\nabla \psi\al^h,\nabla v^h \right) + \left(V_{KS}^h[\rho^h]\psi\al^h,v^h \right), \quad \forall v^h \in X_0^{h,p}\times [0,T], \, \text{and} ~ \forall t \in [0,T] \,,
\end{equation}
where $\rho^h(\br,t)=\sum\limits_{\alpha=1}^{N_e}{\modulus{\psi\al^h(\br,t)}^2}$ and $V_{KS}^h[\rho^h](\br,t)=V_H^h[\rho^h](\br,t)+V_N^h(\br;\bR)+V_{XC}[\rho^h](\br,t)+V_{field}(\br,t)$. 

We now elaborate on the different terms appearing in the expression for $V_{KS}^h[\rho^h](\br,t)$. First, to define appropriate boundary conditions for $V_H^h[\rho^h](\br,t)$ and $V_N^h(\br)$, we introduce the function $f^h(\br;\bR) = \sum_{j=1}^{n_h}q_j N_j(\br)$, with 
\begin{equation*}
q_j=\begin{cases}
\sum_{I=1}^{N_a}\frac{Z_I}{\modulus{\br_j-\bR_I}},&~\text{if}~j^{th}~\text{node~(positioned at}~\br_j)~\text{is a boundary node}\\  
0,\,&~\text{otherwise}\,,
\end{cases}
\end{equation*}
as an interpolation of the boundary conditions of Eqs.~\ref{eq:Poisson} into $X^{h,p}$. This allows us to define the discrete counterpart of the weak solution described in Eq. ~\ref{eq:WeakSolutionHartree} as $V_H^h[\rho^h](\br,t) = V_{H,0}^h[\rho^h](\br,t) + f^h(\br;\bR)$, with $V_{H,0}^h[\rho^h](\br,t) \in X^{h,p}_0\times [0,T]$, such that
\begin{equation} \label{eq:SemiDiscreteWeakSolutionHartree}
    (\nabla V_{H,0}^h,\nabla v^h)=4\pi(\rho^h,v^h)-(\nabla f^h,\nabla v^h), \quad \forall v^h \in X_0^{h,p}\times [0,T]\,, \text{and} ~ \forall t \in [0,T] \,. 
\end{equation}
Similarly, we define the discrete analog of the weak solution defined in Eq. ~\ref{eq:WeakSolutionNuclear} as $V_N^{ae,h}(\br;\bR)=V_{N,0}^{ae,h}(\br;\bR)-f^h(\br;\bR)$, with $V_{N,0}^h(\br;\bR)\in X^{h,p}_0$, such that 
\begin{equation} \label{eq:SemiDiscreteWeakSolutionNuclear}
    (\nabla V_{N,0}^{ae,h},\nabla v^h)=4\pi(b,v^h)+(\nabla f^h,\nabla v^h), \quad \forall v^h \in X_0^{h,p} \,.
\end{equation}

For the pseudopotential case, $V_N^h(\br;\bR)$ is same as the continuous function $V_N^{psp}(\br;\bR)$.

We now introduce the concept of Ritz projection, $\mathcal{P}_h$, which will be used in subsequent error estimates. The Ritz projection $\mathcal{P}_h:H_0^1(\Omega_T)\rightarrow X_0^{h,p}\times [0,T]$ is defined through the following Galerkin orthogonality condition,  
\begin{equation} \label{eq:Ritz}
    \left(\nabla (\psi - \mathcal{P}_h\psi),\nabla v^h\right) = 0, \quad \forall \psi \in H^1_0(\Omega_T), \,\forall v^h \in X_0^{h,p}\times [0,T],\, \text{and} ~ \forall t \in [0,T] \,.
\end{equation}
This allows us to use some standard finite-element error estimates ~\cite{Ciarlet2002} to bound $\norm{\psi-\mathcal{P}_h\psi}_{\Ltwo}$. 

In order to prove the bound of Eq. ~\ref{eq:SemiDiscreteErrorH1GS}, we, first, present a general case with no assumptions on the initial orbitals $\psi\al(\br,0)$. We then present the special case of the initial orbitals being ground-state Kohn-Sham orbitals, as a corollary to the general case. Furthermore, we note that the an error estimate for $\norm{\psi\al-\psi\al^h}_{\Hone}$, in turn, requires an estimate for $\norm{\psi\al-\psi\al^h}_{\Ltwo}$. Therefore, in our subsequent analysis we report estimates for both $\norm{\psi\al-\psi\al^h}_{\Ltwo}$ $\norm{\psi\al-\psi\al^h}_{\Hone}$. We emphasize that, although the numerical studies presented in this work have used hexagonal elements, the following results apply to other shapes of finite element, and hence, in our analysis we denote the mesh using the generic term `triangulation'~\cite{Ciarlet2002}. In particular, we take a triangulation $\mathcal{T}^{h,p}$ of $p^{th}$ order finite-elements covering the domain $\Omega$.
\begin{prop} \label{prop:SemiDiscrete}
    Assuming uniqueness and existence of the solution to Eqs. ~\ref{eq:WeakSolution} and ~\ref{eq:SemiDiscreteWeakSolution}, we obtain the following bounds on the finite-element semi-discrete approximation error to the Kohn-Sham orbitals: 

    \begin{subequations} \label{eq:SemiDiscreteErrorAppendix}
        \begin{equation} \label{eq:SemiDiscreteErrorL2}
            \begin{split}
                \sum_{\alpha=1}^{N_e}\norm{\psi\al-\psi\al^h}_{\Ltwo} (t) &\leq C_{1} e^{C_2 t} (t+1) \sum_e h_e^{p+1}\sum_{\alpha=1}^{N_e}\left(\modulus{\psi\al}_{p+1,\Omega_e}(s_{1,\alpha}) +  \modulus{\psi\al}_{p+1,\Omega_e}(s_{2,\alpha})  +  \modulus{\psi\al}_{p+3,\Omega_e}(s_{2,\alpha})\right) \\
                & \quad + C_{1} e^{C_2 t} t \sum_e h_e^{p+1} \left(\modulus{V_H[\rho^h]}_{p+1,\Omega_e}(s_3) + \modulus{V_N}_{p+1,\Omega_e}\right) + e^{C_2 t} \sum_{\alpha=1}^{N_e}\norm{\psi\al-\psi\al^h}_{\Ltwo}(0)\,, 
            \end{split}
    \end{equation}
        \begin{equation} \label{eq:SemiDiscreteErrorH1}
            \begin{split}
                \sum_{\alpha=1}^{N_e}\norm{\psi\al-\psi\al^h}_{\Hone} (t) & \leq  C_{3} e^{C_2 t} (t+1) \sum_e h_e^{p}\sum_{\alpha=1}^{N_e}\left(\modulus{\psi\al}_{p+1,\Omega_e}(s_{1,\alpha}) + \modulus{\psi\al}_{p+1,\Omega_e}(s_{2,\alpha})  +  \modulus{\psi\al}_{p+3,\Omega_e}(s_{2,\alpha})\right) \\
        & \quad + C_{3} e^{C_2 t} t \sum_e h_e^{p} \left(\modulus{V_H[\rho^h]}_{p+1,\Omega_e}(s_3) + \modulus{V_N}_{p+1,\Omega_e}\right)\\
        & \quad + C_{3} e^{C_2t} h_{min}^{-1}\sum_{\alpha=1}^{N_e}\norm{\psi\al-\psi\al^h}_{\Ltwo}(0) \,,
            \end{split}
    \end{equation}
\end{subequations}
    where $e$ denotes a finite-element of mesh size $h_e$ and cover $\Omega_e$ in the triangulation $\mathcal{T}^{h,p}$, $h_{min}$ represents the smallest element in the triangulation $\mathcal{T}^{h,p}$, and $\modulus{.}_{p,\Omega_e}$ is the semi-norm in $H^{p}(\Omega_e)$. The arguments $s_{1,\alpha},s_{2,\alpha},~\text{and}~s_3$ are defined as
    \begin{equation} \label{eq:SArg}
    \begin{gathered}
        s_{1,\alpha} = \argmax_{0 \leq s \leq t}\norm{\psi\al-\mathcal{P}_h\psi\al}_{\Ltwo}(s), \quad  s_{2,\alpha} = \argmax_{0 \leq s \leq t}\norm{\frac{\partial{\psi\al}}{\partial t} - \mathcal{P}_h\frac{\partial{\psi\al}}{\partial t}}_{\Ltwo}(s),\text{~and~} \\
        \quad s_3 = \argmax_{0 \leq s \leq t}\norm{{V_H}[\rho^h]-V_H^h[\rho^h]}_{\Ltwo}(s)\,.
    \end{gathered}
    \end{equation}
\end{prop}

\begin{proof}
Taking $v=v^h \in X_0^{h,p} \times [0,T]$ in Eq. ~\ref{eq:WeakSolution} (continuous solution) and subtracting it from Eq. ~\ref{eq:SemiDiscreteWeakSolution} (semi-discrete solution), we get
\begin{equation} \label{eq:DiffWeakSolution1}
            i\left(\frac{\partial\left(\psi\al-\psi\al^h\right)}{\partial t}, v^h\right) = 
            \frac{1}{2}\left(\nabla\left(\psi\al-\psi\al^h\right), \nabla v^h\right) + \left(V_{KS}[\rho]\psi\al-V_{KS}^h[\rho^h]\psi\al^h,v^h\right),\, \forall v^h \in X_0^{h,p}\times [0,T] \,.
\end{equation}
We rewrite $\psi\al - \psi\al^h = (\psi\al - \mathcal{P}_h\psi\al) + (\mathcal{P}_h\psi\al - \psi\al^h)$ and derive bounds on each of the terms. For simpler notation, we use $u\al = \psi\al - \mathcal{P}_h\psi\al$ and $w\al= (\mathcal{P}_h\psi\al - \psi\al^h)$. 
Thus, using $\psi\al-\psi\al^h = u\al + w\al$, we rewrite Eq. ~\ref{eq:DiffWeakSolution1} as
\begin{equation} \label{eq:DiffWeakSolution2}
        i\left(\frac{\partial w\al}{\partial t}, v^h\right) = -i\left(\frac{\partial u\al}{\partial t}, v^h\right) + \frac{1}{2}\left(\nabla u\al, \nabla v^h\right) + \frac{1}{2}\left(\nabla w\al, \nabla v^h\right) + \left(V_{KS}[\rho]\psi\al-V_{KS}^h[\rho^h]\psi\al^h,v^h\right)\,.
\end{equation}
Taking $v^h=w\al$, we have
\begin{equation} \label{eq:DiffWeakSolution3}
        i\left(\frac{\partial w\al}{\partial t}, w\al\right) = -i\left(\frac{\partial u\al}{\partial t}, w\al\right) + \frac{1}{2}\left(\nabla u\al, \nabla w\al\right) + \frac{1}{2}\left(\nabla w\al, \nabla w\al\right) + \left(V_{KS}[\rho]\psi\al-V_{KS}^h[\rho^h]\psi\al^h,w\al\right)\,.
\end{equation}
Noting that
    \begin{equation}
       \frac{1}{2}\frac{d}{dt}\norm{w\al}_{\Ltwo}^2 =\text{Re}\left\{\left(\frac{\partial}{\partial t}w\al,w\al\right)\right\}\,,
    \end{equation}
and comparing the imaginary parts of Eq. ~\ref{eq:DiffWeakSolution3}, we have
    \begin{equation} \label{eq:DiffWeakSolution4}
        \begin{split}
            \frac{1}{2}\frac{d}{dt}\norm{w\al}_{\Ltwo}^2 &= -\text{Re}\left\{\left(\frac{\partial u\al}{\partial t}, w\al\right)\right\} + 
            \frac{1}{2} \text{Im}\left\{\left(\nabla u\al,\nabla w\al\right)\right\} + \frac{1}{2}\text{Im}\left\{\left(\nabla w\al,\nabla w\al\right)\right\} \\
            & \quad + \text{Im}\left\{\left(V_{KS}[\rho]\psi\al-V_{KS}^h[\rho^h]\psi\al^h,w\al\right)\right\}
        \end{split}
    \end{equation}
In the above equation, we note that $(\nabla u\al,\nabla w\al)=0$, as a consequence of Eq. ~\ref{eq:Ritz}. Furthermore, $(\nabla w\al,\nabla w\al)$ is real. Thus, Eq. ~\ref{eq:DiffWeakSolution4} simplifies to,
\begin{equation} \label{eq:DiffWeakSolution5}
            \begin{split}
                \frac{1}{2}\frac{d}{dt}\norm{w\al}_{\Ltwo}^2 & = -\text{Re}\left\{\left(\frac{\partial u\al}{\partial t}, w\al\right)\right\} +\text{Im}\left\{\left(V_{KS}[\rho]\psi\al-V_{KS}^h[\rho^h]\psi\al^h,w\al\right)\right\} \\
                & \leq \modulus{\left(\frac{\partial u\al}{\partial t}, w\al\right)} + \modulus{\left(V_{KS}[\rho]\psi\al-V_{KS}^h[\rho^h]\psi\al^h,w\al\right)}\,.
            \end{split}
\end{equation}
\\%
We now decompose $V_{KS}$ into its components to rewrite the second term on the right of the above equation as 
\begin{equation} \label{eq:VKSPsiDiff}
        \begin{split}
            \left(V_{KS}[\rho]\psi\al-V_{KS}^h[\rho^h]\psi\al^h,w\al)\right) &=
            \left(V_{XC}[\rho]\psi\al-V_{XC}[\rho^h]\psi\al^h,w\al\right) + \left(V_{H}[\rho]\psi\al-V_{H}[\rho^h]\psi\al^h,w\al\right) \\
            & \quad + \left(\left(V_{H}[\rho^h]-V_{H}^h[\rho^h]\right)\psi\al^h,w\al\right) + \left(V_{N}\psi\al-V_{N}\psi\al^h,w\al\right) \\
			& \quad + \left(V_{N}\psi\al^h-V_{N}^h\psi\al^h,w\al\right) + \left(V_{field}\psi\al-V_{field}\psi\al^h,w\al\right).
        \end{split}
    \end{equation}
We note that the term $\left(V_{N}\psi\al^h-V_{N}^h\psi\al^h,w\al\right)$, on the right side of the above equation, is relevant only in the all-electron case (i.e., zero for the pseudopotential case as $V_N=V_N^h$).
Combining the results from Eqs. ~\ref{eq:HartreeInequality}, ~\ref{eq:VXCInequality}, ~\ref{eq:VNAllElectronInequality},~\ref{eq:VPSPInequality}, and ~\ref{eq:VFieldInequality}, with $v=w\al$, and using the fact that $\psi^h \in L^{\infty}(\Omega)$, it is straightforward to show that
\begin{equation} \label{eq:VKSInequality}
    \begin{split}
        \modulus{\left(V_{KS}[\rho]\psi\al-V_{KS}^h[\rho^h]\psi\al^h,w\al)\right)} &\leq 
        C_0\norm{\psi\al-\psi\al^h}_{\Ltwo}\norm{w\al}_{\Ltwo} + C_1 \left(\norm{\psi\al-\psi\al^h}_{\Ltwo} + \norm{\rho-\rho^h}_{\Lone}\right)\norm{w\al}_{\Ltwo}\\ 
        & \quad + C_2\norm{V_H[\rho^h]-V_H^h[\rho^h]}_{\Ltwo}\norm{w\al}_{\Ltwo} + C_3\norm{V_N-V_N^h}_{\Ltwo}\norm{w\al}_{\Ltwo}  \,. 
    \end{split}
\end{equation}
Using the above result in Eq. ~\ref{eq:DiffWeakSolution5}, we obtain
\begin{equation} \label{eq:TimeDerivativeW1}
    \begin{split}
        \frac{d}{dt}\norm{w\al}_{\Ltwo} &\leq \norm{\frac{\partial u\al}{\partial t}}_{\Ltwo} + C_0\norm{\psi\al-\psi\al^h}_{\Ltwo} + C_1 \left(\norm{\psi\al-\psi\al^h}_{\Ltwo} + \norm{\rho-\rho^h}_{\Lone}\right) \\
        & \quad + C_2\norm{V_H[\rho^h]-V_H^h[\rho^h]}_{\Ltwo} + C_3\norm{V_N-V_N^h}_{\Ltwo} \\
        &\leq \norm{\frac{\partial u\al}{\partial t}}_{\Ltwo} + C_0\norm{\psi\al-\psi\al^h}_{\Ltwo} \\
        & \quad + C_2\norm{V_H[\rho^h]-V_H^h[\rho^h]}_{\Ltwo} + C_3\norm{V_N-V_N^h}_{\Ltwo} + C_4 \sum_{\beta=1}^{N_e} \norm{\psi\be-\psi\be^h}_{\Ltwo}\,,
    \end{split}
\end{equation}
where we have used Eq. ~\ref{eq:RhoInequality1} in the second line to simplify the term involving $\norm{\rho-\rho^h}_{\Lone}$. Summing the above equation over all index $\alpha$, we have
\begin{equation} \label{eq:TimeDerivativeWSum1}
    \begin{split}
        \frac{d}{dt}\sum_{\alpha=1}^{N_e}\norm{w\al}_{\Ltwo} &\leq \sum_{\alpha=1}^{N_e}\left(\norm{\frac{\partial u\al}{\partial t}}_{\Ltwo} + C_5\norm{\psi\al-\psi\al^h}_{\Ltwo}\right)
        + C_6\norm{V_H[\rho^h]-V_H^h[\rho^h]}_{\Ltwo} + C_7\norm{V_N-V_N^h}_{\Ltwo} \\
        &\leq \sum_{\alpha=1}^{N_e}\left(\norm{\frac{\partial u\al}{\partial t}}_{\Ltwo} + C_5\norm{u\al}_{\Ltwo} + C_5\norm{w\al}_{\Ltwo}\right) \\
        &\quad + C_6\norm{V_H[\rho^h]-V_H^h[\rho^h]}_{\Ltwo} + C_7\norm{V_N-V_N^h}_{\Ltwo}\,, 
    \end{split}
\end{equation}
where in the second line we have split $\psi\al-\psi\al^h$ into $u\al$ and $w\al$. Now, integrating the above equation, gives 
\begin{equation} \label{eq:WSum0}
    \begin{split}
        \sum_{\alpha=1}^{N_e}\norm{w\al}_{\Ltwo}(t) &\leq \sum_{\alpha=1}^{N_e}\norm{w\al}_{\Ltwo}(0) + C_5\int_0^t \sum_{\alpha=1}^{N_e}\norm{w\al}(s)\,ds 
        + C_5 \int_0^t \sum_{\alpha=1}^{N_e}\left(\norm{\frac{\partial u\al}{\partial t}}_{\Ltwo}(s) + \norm{u\al}_{\Ltwo}(s)\right) \,ds \\ 
        & \quad + C_8\int_0^t \left(\norm{V_H[\rho^h]-V_H^h[\rho^h]}_{\Ltwo}(s) + \norm{V_N-V_N^h}_{\Ltwo}\right)\,ds \,.
    \end{split}
\end{equation}
Noting that $u\al=\psi\al-\mathcal{P}_h\psi\al$, $\frac{\partial u\al}{\partial t}=\frac{\partial \psi\al}{\partial t}-\mathcal{P}_h \frac{\partial \psi\al}{\partial t}$, and using the definitions of $s_{1,\alpha}$, $s_{2,\alpha}$, and $s_3$ (cf. Eq. ~\ref{eq:SArg}), we can simplify the above equation as
\begin{equation} \label{eq:WSum1}
    \begin{split}
        \sum_{\alpha=1}^{N_e}\norm{w\al}_{\Ltwo}(t) &\leq \sum_{\alpha=1}^{N_e}\norm{w\al}_{\Ltwo}(0) + C_5\int_0^t \sum_{\alpha=1}^{N_e}\norm{w\al}(s)\,ds 
        + C_5 t \sum_{\alpha=1}^{N_e}\left(\norm{\frac{\partial u\al}{\partial t}}_{\Ltwo}(s_{2,\alpha}) + \norm{u\al}_{\Ltwo}(s_{1,\alpha})\right) \\ 
        & \quad + C_8 t \left(\norm{V_H[\rho^h]-V_H^h[\rho^h]}_{\Ltwo}(s_3) + \norm{V_N-V_N^h}_{\Ltwo}\right) \,.
    \end{split}
\end{equation}
Invoking the Gr\"onwall's inequality on the above equation yields
\begin{equation} \label{eq:WSum2}
    \begin{split}
        \sum_{\alpha=1}^{N_e}\norm{w\al}_{\Ltwo} (t) &\leq e^{C_5 t}\left[\sum_{\alpha=1}^{N_e}\norm{w\al}_{\Ltwo}(0) + C_5 t \sum_{\alpha=1}^{N_e}\left(\norm{\frac{\partial u\al}{\partial t}}_{\Ltwo}(s_{2,\alpha}) + \norm{u\al}_{\Ltwo}(s_{1,\alpha})\right)\right] \\ 
        & \quad + C_8 e^{C_5 t} t \left(\norm{V_H[\rho^h]-V_H^h[\rho^h]}_{\Ltwo}(s_3) + \norm{V_N-V_N^h}_{\Ltwo}\right) \,.
    \end{split}
\end{equation}
Noting that $\norm{w\al}_{\Ltwo}(0)\leq\norm{\psi\al-\psi\al^h}_{\Ltwo}(0)$, we rewrite the above equation as 
\begin{equation} \label{eq:WSum3}
    \begin{split}
        \sum_{\alpha=1}^{N_e}\norm{w\al}_{\Ltwo} (t) &\leq C_5 e^{C_5 t} t \sum_{\alpha=1}^{N_e}\left(\norm{\frac{\partial u\al}{\partial t}}_{\Ltwo}(s_{2,\alpha}) + \norm{u\al}_{\Ltwo}(s_{1,\alpha})\right) \\ 
        & \quad + C_8 e^{C_5 t} t \left(\norm{V_H[\rho^h]-V_H^h[\rho^h]}_{\Ltwo}(s_3) + \norm{V_N-V_N^h}_{\Ltwo}\right) + e^{C_5t}\sum_{\alpha=1}^{N_e}\norm{\psi\al-\psi\al^h}_{\Ltwo} (0)\,.
    \end{split}
\end{equation}
Bounds on the terms involving $\norm{u\al}_{\Ltwo}$, $\norm{\frac{\partial u\al}{\partial t}}_{\Ltwo}$, $\norm{V_H[\rho^h]-V_H^h[\rho^h]}_{\Ltwo}$, and $\norm{V_N-V_N^h}_{\Ltwo}$, can now be obtained using the Ce\'a's lemma ~\cite{Ciarlet2002}--- a standard finite-element error estimates. The Ce\'a's lemma, in simple terms, is stated as follows. Let $\phi \in H^1(\Omega_T)$ and $\phi^h \in V^h \subseteq X^{h,p}$. If $y=\phi-\phi^h$ satisfies the following Galerkin orthogonality condition,
\begin{equation} \label{eq:GalerkinOrtho}
(\nabla y,\nabla v^h)(t)=0,\quad \forall v^h \in V^h ~\text{and}~ \forall t \in [0,T]\,,
\end{equation}
then 
\begin{subequations}\label{eq:Cea}
\begin{equation} \label{eq:CeaL2}
\norm{y}_{\Ltwo} \leq C \sum_e h_e^{p+1}\modulus{\phi}_{p+1,\Omega_e}, \quad \text{and}
\end{equation}
\begin{equation} \label{eq:CeaH1}
\norm{y}_{\Hone} \leq C \sum_e h_e^{p}\modulus{\phi}_{p+1,\Omega_e}\,.
\end{equation}
\end{subequations}
By definition of Ritz projection (Eq. ~\ref{eq:Ritz}), $y=u\al=\psi\al-\mathcal{P}_h\psi\al$ satisfies the Eq. ~\ref{eq:GalerkinOrtho}. Further, taking the time-derivative of Eq. ~\ref{eq:Ritz}, it is easy to verify that $y=\frac{\partial u\al}{\partial t}=\frac{\partial \psi\al}{\partial t} - \mathcal{P}_h\frac{\partial \psi\al}{\partial t}$ also satisfies the Eq. ~\ref{eq:GalerkinOrtho}.
Thus, applying the Ce\'a's lemma (Eq. ~\ref{eq:CeaL2}) to $u\al$ and $\frac{\partial u\al}{\partial t}$ yields
\begin{equation} \label{eq:UiInequality}
   \norm{u\al}_{\Ltwo} \leq C\sum_e h_e^{p+1}\modulus{\psi\al}_{p+1,\Omega_e}, \quad \text{and}
\end{equation}
\begin{equation} \label{eq:RitzTimeDerivativeBound}
    \norm{\frac{\partial u\al}{\partial t}}_{\Ltwo}\leq C\sum_e h_e^{p+1} \modulus{{\frac{\partial \psi\al}{\partial t}}}_{p+1,\Omega_e}\,.
\end{equation}
We further simplify the above inequality, by using Eq. ~\ref{eq:TDKS}
\begin{equation} \label{eq:RitzTimeDerivativeBound2}
    \begin{split}
        \norm{\frac{\partial u\al}{\partial t}}_{\Ltwo} & \leq C\sum_e h_e^{p+1} \modulus{{\frac{\partial \psi\al}{\partial t}}}_{p+1,\Omega_e} = C \sum_e h_e^{p+1} \modulus{-\frac{1}{2}\nabla^2\psi\al + V_{KS}[\rho]\psi\al}_{p+1,\Omega_e}\\
        & \leq C \sum_e h_e^{p+1} \left(\modulus{\psi\al}_{p+3,\Omega_e}+\modulus{(V_H + V_N + V_{XC} + V_{field})\psi\al}_{p+1,\Omega_e}\right)\\
        & \leq C \sum_e h_e^{p+1} \left(\modulus{\psi\al}_{p+3,\Omega_e}+\modulus{\psi\al}_{p+1,\Omega_e}\right)\,, 
    \end{split}
\end{equation}
 which follows from the definition of the $|.|_{p+3}$ semi-norm and the boundedness assumptions on $V_N$, $V_{XC}$, and $V_{field}$ (assumptions ~\ref{assm:VNAllElectron}--\ref{assm:VField}). 
Lastly, it is straightforward to observe that both $y=V_H[\rho^h]-V_H^h[\rho^h]$ and $y=V_N-V_N^h$ satisfy Eq. ~\ref{eq:GalerkinOrtho} (take the difference of Eqs. ~\ref{eq:SemiDiscreteWeakSolutionHartree} and ~\ref{eq:WeakSolutionHartree}; and Eqs. ~\ref{eq:SemiDiscreteWeakSolutionNuclear} and ~\ref{eq:WeakSolutionNuclear}, respectively). Thus, once again, applying the Ce\'a's lemma (Eq. ~\ref{eq:CeaL2}), we get 
\begin{subequations} \label{eq:PoissonDiscreteInequality}
\begin{equation} \label{eq:HartreeDiscreteInequality}
    \norm{V_H[\rho^h]-V_H^h[\rho^h]}_{\Ltwo} \leq C \sum_e h_e^{p+1} \modulus{V_H[\rho^h]}_{p+1,\Omega_e}\,.
\end{equation}
\begin{equation} \label{eq:VNDiscreteInequality}
    \norm{V_N-V_N^h}_{\Ltwo} \leq C \sum_e h_e^{p+1} \modulus{V_N}_{p+1,\Omega_e}\,.
\end{equation}
\end{subequations}
Using Eqs. ~\ref{eq:UiInequality}, ~\ref{eq:RitzTimeDerivativeBound2}, and ~\ref{eq:PoissonDiscreteInequality} in Eq. ~\ref{eq:WSum3}, we have
\begin{equation} \label{eq:WSum4}
    \begin{split}
        \sum_{\alpha=1}^{N_e}\norm{w\al}_{\Ltwo} (t) &\leq C_9 e^{C_5 t} t \sum_e h_e^{p+1}\sum_{\alpha=1}^{N_e}\left(\modulus{\psi\al}_{p+1,\Omega_e}(s_{1,\alpha}) +  \modulus{\psi\al}_{p+1,\Omega_e}(s_{2,\alpha})  +  \modulus{\psi\al}_{p+3,\Omega_e}(s_{2,\alpha})\right) \\
        & \quad + C_9 e^{C_5 t} t \sum_e h_e^{p+1} \left(\modulus{V_H[\rho^h]}_{p+1,\Omega_e}(s_3) + \modulus{V_N}_{p+1,\Omega_e}\right) + e^{C_5t} \sum_{\alpha=1}^{N_e}\norm{\psi\al-\psi\al^h}_{\Ltwo}(0) \,. 
    \end{split}
\end{equation}
Finally, expressing $\psi\al-\psi\al^h=w\al+u\al$ and using the result of Eq. ~\ref{eq:UiInequality} in the above equation, we obtain
\begin{equation} \label{eq:PsiSumiLtwo}
    \begin{split}
        \sum_{\alpha=1}^{N_e}\norm{\psi\al-\psi\al^h}_{\Ltwo} (t) &\leq C \sum_e h_e^{p+1}\sum_{\alpha=1}^{N_e}\modulus{\psi\al}_{p+1,\Omega_e}(s_{1,\alpha}) \\
        & \quad + C_9 e^{C_5 t} t \sum_e h_e^{p+1}\sum_{\alpha=1}^{N_e}\left(\modulus{\psi\al}_{p+1,\Omega_e}(s_{1,\alpha}) +  \modulus{\psi\al}_{p+1,\Omega_e}(s_{2,\alpha})  +  \modulus{\psi\al}_{p+3,\Omega_e}(s_{2,\alpha})\right) \\
        & \quad + C_9 e^{C_5 t} t \sum_e h_e^{p+1} \left(\modulus{V_H[\rho^h]}_{p+1,\Omega_e}(s_3) + \modulus{V_N}_{p+1,\Omega_e}\right) + e^{C_5t} \sum_{\alpha=1}^{N_e}\norm{\psi\al-\psi\al^h}_{\Ltwo}(0)\\
        &\leq C_{10} e^{C_5 t} (t+1) \sum_e h_e^{p+1}\sum_{\alpha=1}^{N_e}\left(\modulus{\psi\al}_{p+1,\Omega_e}(s_{1,\alpha}) +  \modulus{\psi\al}_{p+1,\Omega_e}(s_{2,\alpha})  +  \modulus{\psi\al}_{p+3,\Omega_e}(s_{2,\alpha})\right) \\
        & \quad + C_{10} e^{C_5 t} t \sum_e h_e^{p+1} \left(\modulus{V_H[\rho^h]}_{p+1,\Omega_e}(s_3) + \modulus{V_N}_{p+1,\Omega_e}\right) + e^{C_5 t} \sum_{\alpha=1}^{N_e}\norm{\psi\al-\psi\al^h}_{\Ltwo}(0) \,. 
    \end{split}
\end{equation}
This concludes the proof of Eq. ~\ref{eq:SemiDiscreteErrorL2}.

In order to derive estimates for $\sum_{\alpha=1}^{N_e}\norm{\psi\al-\psi\al^h}_{\Hone}(t)$, we use the inverse estimate ~\cite{BrennerScott2007} for $w\al = (\mathcal{P}_h\psi\al-\psi\al^h) \in X_0^{h,p}$ to obtain  
\begin{equation} \label{eq:WkH1}
    \norm{w\al}_{H^1(\Omega)}(t) \leq C h_{min}^{-1}\norm{w\al}_{L^2(\Omega)}(t)\,. 
\end{equation}
%
%
Additionally, applying the Ce\'a's lemma (Eq. ~\ref{eq:CeaH1}) on $u\al=(\psi\al-\mathcal{P}_h\psi\al)$, we have
\begin{equation}\label{eq:UkH1}
    \norm{u\al}_{\Hone}(t) \leq C \sum_e h_e^p\modulus{\psi\al}_{p+1,\Omega_e} (t)\,.
\end{equation}
Combining Eqs. ~\ref{eq:WkH1} and ~\ref{eq:UkH1}, we get
\begin{equation} \label{eq:PsiSumHone1}
    \begin{split}
        \sum_{\alpha}^{N_e}\norm{\psi\al-\psi\al^h}_{\Hone} (t) & \leq \sum_{\alpha=1}^{N_e}\left(\norm{u\al}_{\Hone}(t)+\norm{w\al}_{\Hone}(t)\right) \\
        & \leq C_{11} \sum_e h_e^p\sum_{\alpha=1}^{N_e} \modulus{\psi\al}_{p+1,\Omega_e}(t) + C_{12} h_{min}^{-1} \sum_{\alpha=1}^{N_e}\norm{w\al}_{\Ltwo}(t)\,.
    \end{split}
\end{equation}
Finally, using the inequality obtained in Eq. ~\ref{eq:WSum4} in the above equation and using the fact that $h_e/h_{min} \leq C$ for all the elements in $\mathcal{T}^{h,p}$, yields
\begin{equation}
    \begin{split}\label{eq:PsiSumHone2}
        \sum_{\alpha}^{N_e}\norm{\psi\al-\psi\al^h}_{\Hone} (t) & \leq 
        C_{11} \sum_e h_e^p\sum_{\alpha=1}^{N_e} \modulus{\psi\al}_{p+1,\Omega_e}(s_{1,\alpha}) \\
        & \quad + C_{13} e^{C_5 t} t \sum_e h_e^{p}\sum_{\alpha=1}^{N_e}\left(\modulus{\psi\al}_{p+1,\Omega_e}(s_{1,\alpha}) + \modulus{\psi\al}_{p+1,\Omega_e}(s_{2,\alpha})  +  \modulus{\psi\al}_{p+3,\Omega_e}(s_{2,\alpha})\right) \\
        & \quad + C_{13} e^{C_5 t} t \sum_e h_e^{p} \left(\modulus{V_H[\rho^h]}_{p+1,\Omega_e}(s_3) + \modulus{V_N}_{p+1,\Omega_e}\right) + C_{12} e^{C_5t} h_{min}^{-1}\sum_{\alpha=1}^{N_e}\norm{\psi\al-\psi\al^h}_{\Ltwo}(0)\\
        & \leq C_{14} e^{C_5 t} (t+1) \sum_e h_e^{p}\sum_{\alpha=1}^{N_e}\left(\modulus{\psi\al}_{p+1,\Omega_e}(s_{1,\alpha}) + \modulus{\psi\al}_{p+1,\Omega_e}(s_{2,\alpha})  +  \modulus{\psi\al}_{p+3,\Omega_e}(s_{2,\alpha})\right) \\
        & \quad + C_{14} e^{C_5 t} t \sum_e h_e^{p} \left(\modulus{V_H[\rho^h]}_{p+1,\Omega_e}(s_3) + \modulus{V_N}_{p+1,\Omega_e}\right) + C_{14} e^{C_5t} h_{min}^{-1}\sum_{\alpha=1}^{N_e}\norm{\psi\al-\psi\al^h}_{\Ltwo}(0)
        \,.
    \end{split}
\end{equation}
%
This concludes the proof for Eq. ~\ref{eq:SemiDiscreteErrorH1}.
\end{proof}

\begin{corollary} \label{corr:SemiDiscreteErrorGS}
    If the initial orbitals $\psi\al(\br,0)$ are obtained from a ground-state DFT calculation, wherein ~\cite{Motamarri2013} 
    \begin{equation} \label{eq:GSInequality}
        \norm{\psi\al-\psi\al^h}_{\Ltwo}(0) \leq C\sum_e h_e^{p+1}\left(\modulus{\psi\al}_{p+1,\Omega_e} + \modulus{V_H[\rho^h]}_{p+1,\Omega_e} + \modulus{V_N}_{p+1,\Omega_e}\right)\,,
    \end{equation}
    the results of Proposition ~\ref{prop:SemiDiscrete} can be simplified, $\forall t \in [0,T]$, to 
    \begin{subequations} \label{eq:SemiDiscreteErrorGSAppendix}
        \begin{equation} \label{eq:SemiDiscreteErrorL2GSAppendix}
            \begin{split}
                \sum_{\alpha=1}^{N_e}\norm{\psi\al-\psi\al^h}_{\Ltwo} (t) &\leq C'_{1} e^{C_2 t} (t+1) \sum_e h_e^{p+1}\sum_{\alpha=1}^{N_e}\left(\modulus{\psi\al}_{p+1,\Omega_e}(s_{1,\alpha}) +  \modulus{\psi\al}_{p+1,\Omega_e}(s_{2,\alpha})  +  \modulus{\psi\al}_{p+3,\Omega_e}(s_{2,\alpha})\right) \\
                & \quad + C'_{1} e^{C_2 t} (t+1) \sum_e h_e^{p+1} \left(\modulus{V_H[\rho^h]}_{p+1,\Omega_e}(s_3) + \modulus{V_N}_{p+1,\Omega_e}\right) \,. 
            \end{split}
    \end{equation}
        \begin{equation} \label{eq:SemiDiscreteErrorH1GSAppendix}
            \begin{split}
                \sum_{\alpha}^{N_e}\norm{\psi\al-\psi\al^h}_{\Hone} (t) & \leq  C'_{3} e^{C_2 t} (t+1) \sum_e h_e^{p}\sum_{\alpha=1}^{N_e}\left(\modulus{\psi\al}_{p+1,\Omega_e}(s_{1,\alpha}) + \modulus{\psi\al}_{p+1,\Omega_e}(s_{2,\alpha})  +  \modulus{\psi\al}_{p+3,\Omega_e}(s_{2,\alpha})\right) \\
                & \quad + C'_{3} e^{C_2 t} (t+1) \sum_e h_e^{p} \left(\modulus{V_H[\rho^h]}_{p+1,\Omega_e}(s_3) + \modulus{V_N}_{p+1,\Omega_e}\right)\\
            \end{split}
    \end{equation}
\end{subequations}
\end{corollary}
The last equation concludes the proof of Eq. ~\ref{eq:SemiDiscreteErrorH1GS}.

\subsection{Derivation of time discretization error estimate} \label{sec:appendC}
Before proceeding to the proof for Eq. ~\ref{eq:PsiTimeDiscreteError}, we note that for an exponential operator of the form $e^{\bL(t)}$, the partial derivative with respect to $t$ is given by ~\cite{Blanes2009},
\begin{equation} \label{eq:DExpDt}
     \frac{\partial}{\partial t} e^{\bL(t)}=  \text{dexp}_{\bL(t)}\left(\dot{\bL}(t)\right)e^{\bL(t)},
\end{equation}
where $\text{dexp}_{\bX}(\bY) = \sum_0^{\infty}\frac{1}{(k+1)!}\text{ad}_{\bX}^k(\bY)$. The operator $\text{ad}_{\bX}^k(\bY)$ is defined recursively as 
\begin{equation} \label{eq:AdXY}
    \text{ad}_{\bX}^k(\bY)=\text{ad}_{\bX}\left(\text{ad}_{\bX}^{k-1}(\bY)\right)\,,
\end{equation}
with $\text{ad}_{\bX}^1(\bY)=\bX\bY - \bY\bX$, and $\text{ad}_{\bX}^0(\bY)=\bY$. 

We now present the proof for Eq. ~\ref{eq:PsiTimeDiscreteError} in the following Proposition. In the following analysis, we assume each time interval $[t_{n-1},t_{n}]$ to be of length $\Delta t$. Moreover, for simpler terminology, we term $e^{\bAtilde_n}$ (cf. Eq. ~\ref{eq:ApproxMagnus}) as the second-order Magnus propagator without explicitly spelling out the mid-point integration rule invoked in it.  
\begin{prop} \label{prop:TimeDiscrete}
    For a second-order Magnus propagator with a mid-point integration rule, we obtain the following bound for the time-discretization error in $\psi\al^h$ 
    \begin{equation} \label{eq:PsiTimeDiscreteErrorAppendix}
        \norm{\psi\al^h(t_n)-\psi\al^{h,n}}_{\Ltwo} \leq C (\Delta t)^2 t_n \max_{0\leq t \leq t_n}\norm{\psi\al^h(t)}_{\Hone}\,,
    \end{equation}
\end{prop}
\begin{proof} 
To begin with, we introduce the following operators 
\begin{subequations} \label{eq:RS}
    \begin{equation} \label{eq:S}
        \bS_0^k=e^{\bA_k}e^{\bA_{k-1}} \ldots e^{\bA_1} = \prod_{l=0}^{k-1}e^{\bA_{k-l}} \quad  \text{for}\,\, 0 < k\leq n\,, \quad \bS_0^0 = I 
    \end{equation}
    \begin{equation} \label{eq:R}
        \bR_k^n=e^{\bAtilde_n}e^{\bAtilde_{n-1}} \ldots e^{\bAtilde_{k+1}} = \prod_{l=0}^{n-k-1}e^{\bAtilde_{n-l}} \quad \text{for}\,\, 0\leq k < n\,, \quad \bR_n^n = I \,.
    \end{equation}
\end{subequations}
To elaborate, $S_0^k$ denotes the exact Magnus propagator from $t_0$ to $t_k$, and $R_{k}^n$ denotes the second-order Magnus propagator from $t_k$ to $t_n$.
Let $\bpsi\al^h(t_n)$ and $\bpsi\al^{h,n}$ denote the vectors containing the finite-element expansion coefficients for $\psi\al^h(t_n)$ and $\psi\al^{h,n}$, respectively. Further, let $\bpsibar\al(t_n)=\bM^{1/2}\bpsi\al(t_n)$ and $\bpsibar\al^n=\bM^{1/2}\bpsi\al^n$. Thus, we can rewrite the time-discretization error in $\bpsibar\al(t_n)$ in terms of the following telescopic series,
    \begin{equation} \label{eq:PsiTimeDiscreteErrorTelescopic}
        \begin{split}
            \bpsibar\al(t_n)-\bpsibar\al^n &= \left(\bR_n^n\bS_0^n-\bR_0^n\bS_0^0\right)\bpsibar\al(0) = \sum_{k=1}^n\left(\bR_k^n\bS_0^{k}-\bR_{k-1}^n\bS_0^{k-1}\right)\bpsibar\al(0)\,. \\
        \end{split}
    \end{equation}
Noting that $\bS^{k}_0=e^{\bA_{k}}\bS^{k-1}_0$ and  $\bR_{k-1}^n=\bR_{k}^ne^{\bAtilde_{k}}$, we rewrite the above equation as
    \begin{equation}\label{eq:PsiTimeDiscreteErrorSplit}
        \begin{split}
            \bpsibar\al(t_n)-\bpsibar\al^n &= \sum_{k=1}^n\left(\bR_k^n\bS_0^{k}-\bR_{k-1}^n\bS_0^{k-1}\right)\bpsibar\al(0)
            = \sum_{k=1}^n\left(\bR_{k}^n e^{\bA_{k}} \bS_0^{k-1}-\bR_{k}^n e^{\bAtilde_{k}}\bS_0^{k-1}\right)\psibar\al(0)\\
            & = \sum_{k=1}^n\bR_{k}^n\left(e^{\bA_{k}}-e^{\bAtilde_{k}}\right)\bS_0^{k-1}\bpsibar\al(0)\\
            & = \sum_{k=1}^n\bR_{k}^n\left(e^{\bA_{k}}-e^{\bAtilde_{k}}\right)\bpsibar\al(t_{k-1})\,.
        \end{split}       
    \end{equation}
Since $\bR_k^n$ is a unitary operator, bounding $\left(\bpsibar\al(t_n)-\bpsibar\al^n\right)$ reduces to finding the bound on $\left(e^{\bA_k}-e^{\bAtilde_k}\right)\bpsibar\al(t_{k-1})$. To this end, we  extend the proof presented in ~\cite{Hochbruck2003} to the non-linear case of the TDKS equations. To begin with, we split $\left(e^{\bA_k}-e^{\bAtilde_k}\right)\bpsibar\al(t_{k-1})$ as   
\begin{equation} \label{eq:ak}
        \left(e^{\bA_k}-e^{\bAtilde_k}\right)\bpsibar\al(t_{k-1})=\left(e^{\bA_k}-e^{\bAbar_k}\right)\bpsibar\al(t_{k-1}) + \left(e^{\bAbar_k}-e^{\bAtilde_k}\right)\bpsibar\al(t_{k-1})\,, 
\end{equation}
where $\bAbar_k=\int_{t_{k-1}}^{t_k}-i\bHbar[\rho(t)]dt$. The two terms on the right hand side of the above equation denote the error due to truncation of the Magnus expansion and the time integral approximation, respectively.

In order to bound the error in the first term on the right side of Eq. ~\ref{eq:ak}, we introduce the following auxiliary function 
\begin{equation} \label{eq:xi}
        \bxi^k\al(t)=e^{\bB_k(t)}\bpsibar\al(t_{k-1}), \qquad \forall t \in [t_{k-1},t_k] \,,
\end{equation}
    where $\bB_k(t)=\int_{t_{k-1}}^{t}-i\bHbar[\rho(\tau)]d\tau$. We remark that $\bxi\al^k(t)$ denotes the time-evolution of $\bpsibar\al(t_{k-1})$ using the truncated the Magnus expansion, in the time interval $[t_{k-1},t_k]$.
Differentiating the above equation and using the result of Eq. ~\ref{eq:DExpDt} gives
\begin{equation} \label{eq:xiDifferentialEq}
        \dot{\bxi}\al^k(t)= \text{dexp}_{\bB_k(t)}(\dot{\bB}_k(t))e^{\bB_k(t)}\bxi\al^k(t_{k-1})=-i\bGbar_k(t)\bxi\al^k(t)\,, \qquad \forall t \in [t_{k-1},t_k],
\end{equation}
    where $\bGbar_k(t)=i\text{dexp}_{\bB_k(t)}(\dot{\bB}_k(t))$. We observe that $\bGbar_k$ is Hermitian. This can be proven as follows. First, note  that for two Hermitian (or skew-Hermitian) matrices $\bX$, $\bY$, the operator $\text{ad}_{\bX}(\bY)$ is skew-Hermitian. Second, owing to the Hermiticity of
    $\bHbar$, both $\bB_k(t)=\int_{t_{k-1}}^{t}-i\bHbar[\rho(\tau)]d\tau$ and $\dot{\bB}_k(t)=-i\bHbar(t)$ ($\forall t \in [t_{k-1},t_k]$) are skew-Hermitian. Thus, by expanding $\text{dexp}_{\bB_k(t)}(\dot{\bB}_k(t))$ and using the above two arguments, it can be shown that $\bGbar_k$ is Hermitian.
We now introduce the function $\bgamma\al^k(t) = \bpsibar\al(t)-\bxi\al^k(t),~\forall t \in[t_{k-1},t_k]$. It is important to note that
\begin{equation} \label{eq:gammatk}
    \bgamma\al^k(t_k)=\bpsibar\al(t_k)-\bxi\al^k(t_k)=\left(e^{\bA_k}-e^{\bAbar_k}\right)\bpsibar\al(t_{k-1})\,,
\end{equation}
where the second equality follows from the definition of $\bxi\al^k$ (Eq.~\ref{eq:xi}) and the fact that $\bB_k(t_k)=\int_{t_{k-1}}^{t_k}-i\bHbar[\rho(\tau)]d\tau = \bAbar_k$. Thus, the problem of bounding $\left(e^{\bA_k}-e^{\bAbar_k}\right)\bpsibar\al(t_{k-1})$ (the first term in Eq.~\ref{eq:ak}) reduces to bounding $\bgamma\al^k(t_k)$. To this end, we proceed, by first expressing the time-derivative of $\bgamma\al^k$ as
\begin{equation} \label{eq:GammaDifferentialEq}
        \dot{\bgamma}\al^k(t)= -i\bGbar_k(t)\bgamma\al^k(t)-i(\bHbar(t)-\bGbar_k(t))\bpsibar\al(t), \qquad \forall t\in [t_{k-1},t_k] \,,
\end{equation}
which follows from Eqs. ~\ref{eq:TDKSMatVecStandard} and ~\ref{eq:xiDifferentialEq}.
Now, taking the dot product with $\bgamma\al^k(t)^\dag$ on both sides yields
    \begin{equation} \label{eq:GammaDotProd}
        \bgamma\al^k(t)^\dag\dot{\bgamma}\al^k(t)= -i\bgamma\al^k(t)^\dag\bGbar_k(t)\bgamma\al^k(t) -i\bgamma\al^k(t)^\dag\left(\bHbar(t)-\bGbar_k(t)\right)\bpsibar\al(t)\,.
    \end{equation}
We note that $2\text{Re}\{\bgamma\al^k(t)^\dag\dot{\bgamma}\al^k(t)\}=\frac{d}{dt}\norm{\bgamma\al^k}^2$, where $\norm{.}$ represents the Euclidean norm of a vector. Further, we note $\bgamma\al^k(t)^\dag\bGbar_k(t)\bgamma\al^k(t)$ is real, owing to the Hermiticity of $\bGbar$. Thus, comparing the real parts of the above equation results in
    \begin{equation} \label{eq:GammaModDifferential}
        \frac{1}{2}\frac{d}{dt}\norm{\bgamma\al^k}^2=\text{Im}\{\bgamma\al^k(t)^\dag\left(\bHbar(t)-\bGbar_k(t)\right)\bpsibar\al(t)\}\,.
    \end{equation}
Consequently,
\begin{equation} \label{eq:GammaModDifferentialInequality}
    \frac{d}{dt}\norm{\bgamma\al^k}\leq \norm{\left(\bHbar(t)-\bGbar_k(t)\right)\bpsibar\al(t)}\,.
\end{equation}
Time integrating the above equation yields
\begin{equation} \label{eq:GammaModInequality}
    \norm{\bgamma\al^k}(t_k)=\norm{\bpsibar\al(t_k)-\bxi\al^k(t_k)}=\norm{\left(e^{\bA_k}-e^{\bAbar_k}\right)\bpsibar\al(t_{k-1})}\leq\int_{t_{k-1}}^{t_k}\norm{\left(\bHbar(\tau)-\bGbar_k(\tau)\right)\bpsibar\al(\tau)}d\tau\,,
\end{equation}
where we have used the result of Eq. ~\ref{eq:gammatk} along with the fact that $\norm{\bgamma\al^k}(t_{k-1})=\bpsibar\al(t_{k-1})-\bxi\al^k(t_{k-1}) = 0$ (by the definition of $\bxi\al^k(t)$, cf. Eq. ~\ref{eq:xi}). Thus, the problem of bounding $\norm{\left(e^{\bA_k}-e^{\bAbar_k}\right)\bpsibar\al(t_{k-1})}$ further simplifies to finding a bound for $\int_{t_{k-1}}^{t_k}\norm{\left(\bHbar(\tau)-\bGbar_k(\tau)\right)\bpsibar\al(\tau)}d\tau$. To this end, we use the fact that $\bGbar_k(\tau) = i\text{dexp}_{\bB_k(\tau)}(\dot{\bB}_k(\tau))$ and the definition of the operator $\text{dexp}_{\bX}(\bY)$, to obtain
    \begin{equation} \label{eq:HGDiff}
        \bHbar(\tau)-\bGbar_k(\tau) = -\frac{i}{2} [\bB_k(\tau),\dot{\bB}_k(\tau)] + h.o.t = -\frac{i}{2}\int_{t_{k-1}}^{\tau}[\bHbar(\tau),\bHbar(\sigma)]d\sigma + h.o.t. \,.
    \end{equation}
%
In order to bound $[\bHbar(\tau),\bHbar(\sigma)]$, we begin by rewriting $\bHbar$ in terms of $\bUbar$ and $\bVbar$, i.e., its kinetic and Kohn-Sham potential components. To elaborate, $\bUbar=\bM^{-1/2}\bU\bM^{-1/2}$ and $\bVbar=\bM^{-1/2}\bV\bM^{-1/2}$, with $U_{jk}=\int_{\Omega}\nabla N_j(\br).\nabla N_k(\br)\dr$ and $V_{jk}=\int_{\Omega}V_{KS}^h[\rho^h](\br,t)N_j(\br)N_k(\br)\dr$. Noting that $\bUbar$ is time-independent, we Taylor expand $\bHbar(\sigma)$ about $\tau$ to rewrite $[\bHbar(\tau),\bHbar(\sigma)]$ as 
    \begin{equation} \label{eq:HCommutator2}
        [\bHbar(\tau),\bHbar(\sigma)] = [\bHbar(\tau),\bVbar'(\tau)](\sigma-\tau) + \order((\sigma-\tau)^2)\,,
    \end{equation}
where $\bVbar'(\tau)=\frac{d}{dt}(\bVbar(t))\rvert_{\sigma=\tau}$. Thus, using the above relation in Eq. ~\ref{eq:HGDiff} we get 
    \begin{equation} \label{eq:HGPsiDiff}
        (\bHbar(\tau)-\bGbar_k(\tau))\bpsibar\al(\tau)= \frac{i}{4}\left([\bHbar,\bVbar'(\tau)]\bpsibar\al(\tau)\right)(\tau-t_{k-1})^2 + \order((\tau-t_{k-1})^3)\,.
    \end{equation}
We now invoke the boundedness assumption on $\bVbar'$ (assumption ~\ref{assm:VKSTimeDerBound}), and the norm equivalence of $\bUbar$ and $\bHbar$ (assumption ~\ref{assm:NormEqui}), to obtain 
\begin{equation} \label{eq:HGPsiDiffBound}
    \norm{(\bHbar(\tau)-\bGbar_k(\tau))\bpsibar\al(\tau)} \leq C (\tau-t_{k-1})^2 \norm{\bpsibar\al(\tau)}_{\bUbar} +  \order((\tau-t_{k-1})^3).
\end{equation}
%
Thus, substituting the above result into Eq. ~\ref{eq:GammaModInequality} provides the following bound
    \begin{equation} \label{eq:GammaModInequalitytk2}
        \norm{\left(e^{\bA_k}-e^{\bAbar_k}\right)\bpsibar\al(t_{k-1})} \leq C (\Delta t)^3 \max_{t_{k-1} \leq t \leq t_k} \norm{\bpsibar\al(t)}_{\bUbar}\,. 
    \end{equation}
This provides a bound for the first term (truncation error) on the right side of Eq. ~\ref{eq:ak}. In order to bound the second term on the right side of Eq. ~\ref{eq:ak}, i.e., the error due to mid-point quadrature rule, we begin with the following identity \begin{equation} \label{eq:expAbarAtilde}
    e^{\bAbar_k}-e^{\bAtilde_k} = \int_0^1 \frac{d}{dx}\left(e^{(1-x)\bAtilde_k}e^{x\bAbar_k}\right)\;dx= \int_0^1 e^{(1-x)\bAtilde_k}(\bAbar_k-\bAtilde_k)e^{x\bAbar_k}\,dx.
\end{equation}
Furthermore, we note that for a function $f(x)$ if $F_{1/2}$ denotes the midpoint approximation to $F = \int_a^b f(x)\; dx$, then $\modulus{F-F_{1/2}}\leq C (b-a)^3 f''(\eta)$, for some $\eta \in [a,b]$. Thus, for the mid-point integration rule, $\norm{\bAbar_k-\bAtilde_k}\leq C (\Delta t)^3\norm{\frac{d^2}{dt^2}(\bHbar)\rvert_{t'}}$ for some $t' \in [t_{k-1},t_k]$. Using this result along with the unitarity of the operators $e^{(1-x)\bAtilde_k}$ and $e^{x\bAbar_k}$, we obtain
   \begin{equation} \label{eq:expAbarAtildeBound}
       \norm{\left(e^{\bAbar_k}-e^{\bAtilde_k}\right)\bpsibar\al(t_{k-1})} \leq \norm{(\bAbar_k-\bAtilde_k)} \leq C(\Delta t)^3\norm{\frac{d^2}{dt^2}(\bHbar)\rvert_{t'}},\, \text{for some } t' \in [t_{k-1},t_k]\,.
   \end{equation}
Noting that $\frac{d^2}{dt^2}\bHbar = \frac{d^2}{dt^2}\bVbar, \,\forall t \in [t_{k-1},t_k]$, and invoking the boundedness assumption on $\frac{d^2}{dt^2}\bVbar$ (assumption ~\ref{assm:VKSTimeDerBound}), we get 
\begin{equation} \label{eq:MidPointBound}
        \norm{\left(e^{\bAbar_k}-e^{\bAtilde_k}\right)\bpsibar\al(t_{k-1})} \leq C (\Delta t)^3\,. 
\end{equation}
Thus, using the results of Eqs. ~\ref{eq:GammaModInequalitytk2} and ~\ref{eq:MidPointBound} in Eq.~\ref{eq:PsiTimeDiscreteErrorSplit} along with the unitarity of the operators $\bR_k^n$, yields
   \begin{equation} \label{eq:PsiTimeDiscreteErrorBound}
       \norm{\bpsibar\al(t_n)-\bpsibar\al^n} \leq C (\Delta t)^2 t_n \max_{t_0 \leq t \leq t_n} \norm{\bpsibar\al(t)}_{\bUbar}.
   \end{equation}
Finally, noting that the coefficient vectors for the spatial fields $\psi\al^h(\br,t_n)$ and $\psi\al^{h,n}(\br)$ are given by $\bM^{-1/2}\bpsibar\al^h(t_n)$ and $\bM^{-1/2}\bpsibar\al^{h,n}$, respectively, it is now trivial to arrive at Eq.~\ref{eq:PsiTimeDiscreteError} from the above equation.

%
%
\end{proof}
\twocolumngrid
\bibliography{ref}
\bibliographystyle{apsrev4-1}
\end{document}